\documentclass[11pt]{article}
\usepackage[margin=1in]{geometry}
\usepackage{graphicx}   
\usepackage{xspace}
\usepackage{enumitem}
\usepackage{amsmath,amssymb,amsthm}
\usepackage{mathtools}
\usepackage{adjustbox}
\usepackage{tablefootnote}
\usepackage[dvipsnames]{xcolor}
\usepackage{hyperref}

\usepackage{indentfirst}
\usepackage[utf8]{inputenc}
\usepackage{multirow}

\SetLabelAlign{parright}{\strut\smash{\parbox[t]\labelwidth{\raggedleft#1}}}

\theoremstyle{definition}
\newtheorem{definition}{Definition}[section]
 
\newtheorem{remark}[definition]{Remark} 
\newtheorem{construction}[definition]{Construction}
\newtheorem{question}[definition]{Question}

\theoremstyle{plain}
\newtheorem{proposition}[definition]{Proposition}
\newtheorem{corollary}[definition]{Corollary}
\newtheorem{theorem}[definition]{Theorem}

\newtheorem{lemma}[definition]{Lemma} 
\newtheorem{deflem}[definition]{Definition/Lemma} 
\newtheorem{fact}[definition]{Fact} 
\newtheorem{observation}[definition]{Observation}

\newtheorem{maintheorem}{Theorem}

\newcommand{\FF}{\mathbb{F}}

\DeclareMathOperator{\im}{im}

\AtBeginDocument{\selectfont} 

\title{Group Isomorphism and the Polylogarithmic-Time Hierarchy: \\ Depth-2$\frac{1}{2}$ Circuits and Lower Bounds}

\author{
Joshua A. Grochow\thanks{Department of Computer Science, University of Colorado Boulder; Department of Mathematics, University of Colorado Boulder. Email: \texttt{jgrochow@colorado.edu}}
\and
G\"ulce Karde\c{s}\thanks{Department of Computer Science, University of Colorado Boulder; Santa Fe Institute. Email: \texttt{gulce.kardes@colorado.edu}}
\and
Michael Levet\thanks{Department of Computer Science, College of Charleston. Email: \texttt{levetm@cofc.edu}}
}

\date{}
\usepackage{amsmath,amssymb}
\usepackage[most]{tcolorbox}
\newcommand{\betacc}[1]{\ifthenelse{\equal{#1}{1}}{\exists^{\log n}}{\exists^{\log^{#1}n}}} 
\newcommand{\alphacc}[1]{\ifthenelse{\equal{#1}{1}}{\forall^{\log n}}{\forall^{\log^{#1}n}}}

\DeclareMathOperator{\poly}{poly}
\DeclareMathOperator{\GL}{GL}
\DeclareMathOperator{\Mat}{Mat}
\DeclareMathOperator{\SL}{SL}
\DeclareMathOperator{\Sz}{Sz}
\DeclareMathOperator{\SU}{SU}
\DeclareMathOperator{\PSU}{PSU}

\DeclareMathOperator{\Tr}{Tr}
\DeclareMathOperator{\Nm}{Nm}

\usepackage{graphicx} 
\usepackage{multirow} 
\usepackage{listings}
\usepackage{amsmath,amssymb,amsfonts} 
\usepackage{amsthm} 
\usepackage{mathrsfs} 
\usepackage[title]{appendix} 
\usepackage{algorithm}
\usepackage{algpseudocode} 
\usepackage{xcolor} 
\usepackage{textcomp} 
\usepackage{manyfoot} 
\usepackage{booktabs} 
\usepackage{algorithm} 
\usepackage{algorithmicx} 
\usepackage{algpseudocode} 
\usepackage{listings} 
\usepackage{soul}      
\sethlcolor{blue!15} 
\newcommand{\algprobm}[1]{\textsc{#1}\xspace}

\newcommand{\Aut}{\operatorname{Aut}}

\newcommand{\Span}{\operatorname{Span}}

\newcommand*{\ComplexityClass}[1]{\ensuremath{\mathsf{#1}}\xspace}

\renewcommand{\P}{\ComplexityClass{P}}
\newcommand*{\coNP}{\ComplexityClass{coNP}}

\newcommand{\DTISP}{\ComplexityClass{DTISP}}

\newcommand{\DTISPpll}{\ensuremath{\DTISP(\mathrm{polylog}(n),\mathrm{log}(n))}\xspace}

\newcommand{\DTIMEpl}
{\ensuremath{\mathsf{DTIME}(\mathrm{polylog}(n))}\xspace}

\newcommand*{\LogSpace}{\ComplexityClass{L}}

\newcommand{\NC}{\ComplexityClass{NC}}
\newcommand{\AC}{\ComplexityClass{AC}}
\newcommand{\SAC}{\ComplexityClass{SAC}}

\newcommand{\ACz}{\ComplexityClass{AC^0}}

\newcommand{\qACz}{\ComplexityClass{quasiAC^0}}
\newcommand{\FOLL}{\ComplexityClass{FOLL}}

\newcommand{\DLOGTIME}{\ComplexityClass{DLOGTIME}}
\newcommand{\LOGSPACE}{\ComplexityClass{L}}

\newcommand{\DTIME}{\ComplexityClass{DTIME}}

\newcommand{\polylog}{\operatorname{polylog}}
\usepackage{titlesec}

\titleclass{\subsubsubsection}{straight}[\subsubsection]
\newcounter{subsubsubsection}[subsubsection]

\makeatother

\begin{document}

\maketitle
\begin{abstract}
In this paper, we investigate the low-depth circuit complexity of \algprobm{Group Isomorphism} in the multiplication (Cayley) table model. We prove the first circuit lower bounds for \algprobm{Group Isomorphism}: namely, we show that every family of depth-$2$ Boolean circuits deciding \algprobm{Group Isomorphism} requires quasipolynomial-size. (We also extend this to a limited class of depth-3 circuits, namely $\textsf{CNF}$-of-$\textsf{XOR}$s.) We complement this with upper bounds of uniform depth-$2\frac{1}{2}$ circuits of quasipolynomial-size: depth 3 circuits with poly-logarithmic bottom fan-in.

A sequence of previous results from 1970--2025 progressively reduced the circuit depth from polynomial to $3\frac{1}{2}$; all of these results relied on the generator-enumerator strategy and, in fact, applied more generally to quasigroups. In contrast, our depth-$2\frac{1}{2}$ construction follows a fundamentally different strategy that exploits structure more specific to groups.
We guess a composition series for each group, together with generators for its terms and the isomorphism types of its composition factors. We then inductively verify that the corresponding extensions at each level of the two composition series are compatible. A central part in this approach 
brings to bear the extensive work on the Short Presentation Conjecture, in tandem with the algorithmic theory of group extensions and cohomology.
\end{abstract}

\setcounter{tocdepth}{1}
\tableofcontents

\thispagestyle{empty}
\setcounter{page}{0}
\newpage

\section{Introduction}
Isomorphism problems ask whether two objects are ``the same'' up to some natural notion of relabeling or transformation. They have played a central role in computational complexity, with \algprobm{Graph Isomorphism} perhaps being the most well-studied (see \cite{AllenderDas} for more on its history and role in computational complexity). In this paper, we focus on its algebraic sibling, \algprobm{Group Isomorphism} (\algprobm{GpI}); in fact, since \algprobm{Group Isomorphism} (when the input is given by multiplication tables) reduces to \algprobm{Graph Isomorphism} \cite{ZKT}, the former can be considered a particular algebraic sub-problem of the latter. From the algebraic perspective, the algorithmic \algprobm{Group Isomorphism} problem is a very natural question of interest in its own right, having been asked by Dehn over a century ago, in 1912 \cite{Dehn}; it has since become a central question in computational complexity and computational algebra. \nocite{DehnTranslation}

We focus on the parallel, low-depth uniform circuit complexity of this problem when the inputs are given by their multiplication (Cayley) tables. In computer algebra systems, groups are almost always given succinctly either by a presentation, or generating sets of permutations or matrices, which are much more succinct than the Cayley table. However, even in these more succinct encodings, the runtime of the best-known algorithms for (finite) \algprobm{Group Isomorphism} are $|G|^{\Theta(\log |G|)}$ \cite{BEO02, ELGO02, BE99, CH03}  (see \cite[p.~2]{WilsonSubgroupProfiles}; see also \cite{Rosenbaum2013BidirectionalCD,LuksCompositionSeriesIso} and \cite[Sec.~2.2]{GR16} for the best-known worst-case analysis in the Cayley table model), a runtime that depends on the size of the group, but could be exponential relative to the more succinct input size. In contrast, in the Cayley table model $|G|$ is polynomially related to the input length, so the above running time is expressed (as a quasipolynomial) in terms of the input length.

In the Cayley table model, the $|G|^{O(\log |G|)}$ bound comes from the generator-enumerator technique \cite{FN, MillerTarjan},\footnote{Miller attributes the result to Tarjan.} which has a short explanation: by Lagrange's Theorem, the order of a subgroup divides the order of the group, and a straightforward exercise then shows that every group $G$ is generated by at most $\log_2 |G|$ elements. The generator-enumerator isomorphism technique is to find a generating set of size $\leq \log |G|$ for $G$, then to try all generating sets of $H$ of the same size and to check whether one can map the generating set of $G$ to that of $H$ in a way that extends to an isomorphism. Since it was introduced in 1970, a series of results \cite{FN,MillerTarjan,LiptonSnyderZalcstein, Wolf, WagnerThesis, ChattopadhyayToranWagner,
TangThesis, CGLW} has shown that the generator-enumerator technique can be implemented in smaller and smaller complexity classes, albeit always with quasipolynomial time or circuit size (see Section~\ref{sec:related} for details). The most recent advance \cite{CGLW} showed that this technique could be implemented by unbounded fan-in uniform circuits of $n^{O(\log n)}$ size and constant depth (the natural circuit class \qACz \cite{BarringtonQuasipolynomial}), in fact depth only 4.\footnote{Throughout the paper, unless otherwise noted, circuit classes denote their uniform versions.}

Perhaps the biggest question about \algprobm{Group Isomorphism} is whether it can be solved in polynomial time. From the parallel complexity perspective, the analogous big question is:
\begin{quotation}
Is \algprobm{Group Isomorphism} in $\ACz$?
\end{quotation}
We would be shocked if \algprobm{Group Isomorphism} were in $\ACz$. Despite this, circuit lower-bounds even at depth-$2$ have been elusive. In this paper, we make the first progress on circuit lower bounds for this question, as well as improve the depth of the upper bound (albeit still quasipolynomial size---that is, we do not improve the worst-case serial runtime). We cover these two directions in each of the next two subsections.

\subsection{Lower bounds}
It was noted in \cite{CGLW} that we did not have lower bounds even against depth-2 circuits for \algprobm{GpI}. And there is good reason for this difficulty: Chattopadhyay, Tor\'an, and Wagner
\cite{ChattopadhyayToranWagner} exhibited an upper bound of
$
    \exists^{\log^2 n}\FOLL \cap \exists^{\log^2 n}\textsf{L}
$
for the more general \(\algprobm{Quasigroup Isomorphism}\) problem
(\(\algprobm{QGpI}\)). Notably, \(\exists^{\log^2 n}\FOLL\) cannot compute
\(\algprobm{Parity}\) or \(\algprobm{Majority}\)
\cite{FSS, ChattopadhyayToranWagner}. Consequently, the standard strategy of
obtaining $\ACz$ lower bounds via low-depth circuit reductions from \(\algprobm{Parity}\) or
\(\algprobm{Majority}\) cannot work for \(\algprobm{GpI}\) in this setting.

Our first main result is the first lower bound against Boolean circuits for \algprobm{GpI} of which we are aware:

\begin{maintheorem} \label{thm:MainLowerBounds}
Any (uniform or non-uniform) family of depth-$2$ circuits deciding \algprobm{Group Isomorphism} requires size $|G|^{\Omega(\log |G|)}$.    
\end{maintheorem}

We show more strongly that DNFs require exponential size (Theorem~\ref{thm:DNF}), while CNFs require size $|G|^{\Omega(\log |G|)}$ (Theorem~\ref{thm:CNF}). Theorem~\ref{thm:MainLowerBounds} is thus the first step towards the longstanding open problem of whether \algprobm{GpI} is in \ACz (see, e.\,g., \cite{ChattopadhyayToranWagner}\footnote{As they state in their conclusion, their research ``started originally trying to prove that \algprobm{QGroupIso} is hard for $\textsf{NC}^1$,'' yet to date we are not aware of any previous progress towards such a Boolean lower bound.}). 

This is in sharp contrast with the setting of \algprobm{Graph Isomorphism}. where circuit lower bounds have long been known. To the
best of our knowledge, the first complexity-theoretic lower bound for
\(\algprobm{GI}\) was established by Cai, F\"urer, and Immerman \cite{CFI},
who gave an \(\ACz\)-reduction from \(\algprobm{Parity}\) to
\(\algprobm{GI}\). The best known lower bound for \(\algprobm{GI}\) is
\(\textsf{DET}\) \cite{Toran}. By contrast, although some lower bounds for \algprobm{GpI} have been proved in semantically restricted models, we are not aware of prior lower bounds for Boolean circuits for \algprobm{GpI} (see Related Work for more discussion).

Given that the ``hard instances'' for CNFs in Theorem~\ref{thm:MainLowerBounds} are Abelian 2-groups, it is natural to wonder whether $\mathsf{Mod_2}=\mathsf{XOR}$ gates could help in determining their isomorphism, and thereby that our lower bound ultimately relies ``only'' on the exponential hardness of $\mathsf{XOR}$ for $\ACz$ circuits \cite{FSS,Yao,hastad}. Intuitively, this is not the case because our reduction only seems to reduce from a $\log n$-bit $\mathsf{XOR}$, which is trivial for CNFs to compute (see Remark~\ref{rmk:parity}). Formally, we can show this is not the case by showing that our same reduction, with some additional work, also proves that $\algprobm{GpI}$ for Abelian 2-groups cannot be computed by polynomial-size $\mathsf{CNF}$s-of-$\mathsf{XOR}$s:  conjunctions ($\mathsf{AND}$) of disjunctions ($\mathsf{OR}$s) of parities ($\mathsf{XOR}$s) of input bits. We denote this class of circuits $\ComplexityClass{CNF \circ XOR}$: 

\begin{maintheorem}[{=Theorem~\ref{thm:CNF-XOR}}] \label{thm:MainLowerBoundXOR}
Any (uniform or non-uniform) family of $\ComplexityClass{CNF \circ XOR}$ circuits deciding \algprobm{Group Isomorphism} requires size $|G|^{\Omega(\log |G|)}$.  
\end{maintheorem}

In combination with our reduction, the key new ingredient to establish Theorem~\ref{thm:MainLowerBoundXOR} is our proof of what we believe is a new bound (Theorem~\ref{thm:space-missing-k-1}) on the dimension of affine subspaces that contain at least one invertible matrix but no matrix of co-rank exactly $1$.

Just as our $\mathsf{CNF}$ lower bound can be viewed as a first step towards showing $\algprobm{GpI} \notin \ACz$, this result could be viewed as a first step towards showing $\algprobm{GpI} \notin \ACz[\oplus]$ (or at least $\ACz \circ \mathsf{XOR}$). We can also view $\ComplexityClass{CNF \circ XOR}$ as a particular class of depth-3 circuits ($\mathsf{AND} \circ \mathsf{OR} \circ \mathsf{XOR}$), albeit one that is incomparable to $\AC$ circuits of depth $3$.\footnote{Already $\mathsf{XOR}$ is not in $\ACz$ of any depth \cite{FSS,Yao}. In the other direction, \cite[Theorem~3]{RSS} gives a function in $\mathsf{AND} \circ \mathsf{OR} \circ \mathsf{AND}$ that requires exponential-size $\mathsf{DNF} \circ \mathsf{XOR}$ circuits; taking negations gives us a depth-3 $\ACz$ function that requires exponential-size $\mathsf{CNF} \circ \mathsf{XOR}$ circuits.} Although we do not rule out the utility of $\mathsf{XOR}$ gates higher in the circuit, we note that a long-open question at the frontier of circuit complexity is to prove $\ComplexityClass{XOR \circ AND} \not\subseteq \ComplexityClass{AC^0 \circ XOR}$ (with particular focus on the Inner Product function mod 2).

\paragraph{Methods.} Our proof of the CNF lower bound is significantly more involved than that for the DNF lower bound. 
The DNF lower bound follows from a relatively straightforward observation about Cayley tables of (quasi)groups: No two such Cayley tables can differ in exactly one entry. Our proof of the CNF lower bound combines (1) a novel reduction from \algprobm{Matrix Rank} to \algprobm{GpI} with (2) results by Meshulam \cite{Meshulam1985} and de Seguins Pazzis \cite{deSeguinsPazzis2010} on the largest dimension of a(n affine) linear space of matrices of bounded rank. Because of the aforementioned result \cite{ChattopadhyayToranWagner}, this reduction is necessarily not an $\ACz$ reduction, which was an obstruction to previous attempts at such lower bounds (see Remark~\ref{rmk:parity} for more details). We get around this obstruction in that our reduction is an affine projection reduction with exponential blow-up: each output bit is an $\mathbb{F}_2$-affine linear combination of the input bits. 

In more detail, for the reduction, we associate with each $k \times k$ matrix $M$ over $\mathbb{F}_2$ an Abelian group $G_M$ on $\mathbb{F}_2^k \times \mathbb{F}_2^k$ with multiplication
$(u,z)\cdot_M(v,w)=(u+v,z+w+M(u\odot v))$ where $\odot$ denotes coordinate-wise (Hadamard) multiplication. Squaring gives $(u,z)^2=(0, Mu)$ so the subgroup of squares has order $2^{\operatorname{rank}(M)}$. Due to the classification of finite Abelian groups, $G_M \cong(\mathbb{Z}_4)^{\operatorname{rank}(M)}\times \mathbb{Z}_2^{2k-2\operatorname{rank}(M)}$. So the rank of $M$ determines and is determined by the isomorphism type of $G_M$, and for a fixed rank-$r$ matrix $R$, testing whether $\operatorname{rank}=r$ is equivalent to testing whether $G_M \cong G_R.$ This step pinpoints a restricted family on which the yes-instances have a simpler linear-algebraic description. An especially helpful aspect of this construction is that in the coordinate labels each bit of the Cayley table of $G_M$
is an affine linear function of at most one row of $M$.

For the CNF lower bound, fix $s=\lfloor k/2 \rfloor$ and a matrix $R$ of rank-($s+1$) on the restricted CNF of GpI to inputs $(G_M, G_R)$. We focus on the falsifying sets for clauses, and notice that the matrices falsifying any one clause satisfy independent affine constraints on their rows. Each nonempty falsifying set is thus a “row-cube” $A_1 \times \cdots \times A_k$ where each $A_j \subseteq \mathbb{F}_2^k$ is an affine subspace. Correctness imposes that these row-cubes avoid every rank-$(s+1)$ matrix and altogether cover every matrix of rank at most $s$; moreover, any such
row-cube meeting rank at most $s$ contains only matrices of rank at most $s$, since a path changing one row at a time would otherwise pass through rank $s+1$
(Lemma~\ref{lem:meshulam}). The bounded rank dimension theorem mentioned above \cite{deSeguinsPazzis2010} then bounds the dimension of each such row-cube. Comparing this bound with the number of matrices of rank at most $s$ that the row-cubes must cover yields the lower bound.

For the $\ComplexityClass{CNF \circ XOR}$ lower bound, we use the same reduction as above, but instead of row-cubes, we now must tackle arbitrary affine subspaces of matrices. To do this, we now focus on matrices of rank $k-1$ (rather than $k/2$ as above). Because our spaces are no longer row-cubes, the affine spaces we must deal with must avoid rank $k-1$, but can have matrices both above and below that rank. Thus, instead of the aforementioned results \cite{deSeguinsPazzis2010} on affine spaces of bounded rank, we prove what we believe is a new result (Theorem~\ref{thm:space-missing-k-1}, which works over arbitrary fields), bounding the dimension of any affine space of matrices that contains a rank-$k$ matrix, but no matrix whose rank is exactly $k-1$.\footnote{We remark that such spaces with a ``gap'' in ranks---avoiding rank $k-1$ but having both matrices of rank $k$ and matrices of rank $\leq k-2$---exist. That is, we cannot divide the spaces into ``affine spaces of invertible matrices'' and then separately ``affine spaces of rank $\leq k-2$.'' For example, when $k\geq 2$, the affine line $\{\operatorname{diag}(t,t,1,\ldots,1):t \in  \mathbb{F}\}$ contains matrices of rank $k$ and $k-2$, but none of rank $k-1.$}

\begin{remark}[On $\ACz$ reductions from \algprobm{Parity}] \label{rmk:parity}
Since depth-$d$ $\ComplexityClass{AC}$ circuits for $n$-bit \algprobm{Parity} require size $2^{\Omega\left(n^{1/(d-1)}\right)}$ \cite{hastad}, it follows from \cite{ChattopadhyayToranWagner} unconditionally that any $\ACz$ reduction from \algprobm{Parity} to \algprobm{Group Isomorphism} on groups of order $n$ can be from parities on at most $\polylog(n)$ bits; the tighter upper bound of depth-$4$ circuits of size $n^{O(\log n)}$ from \cite{CGLW} implies any $\mathsf{AND}$-reduction can be from parities on at most $O(\log^6 n)$ bits. However, we do not know any reduction from \algprobm{Parity} on more than $O(\log n)$ bits to \algprobm{GpI} for groups of order $n$, and such a reduction gives no notable lower bounds even on CNFs, since all functions on $O(\log n)$ bits have (at least non-uniform; and in the case of \algprobm{Parity} these can be made uniform) decision trees of depth $O(\log n)$, hence CNFs and DNFs of $\poly(n)$ size. 

In contrast, we give a reduction from $k \times k$ \algprobm{Matrix Rank} over $\FF_2$ to \algprobm{GpI} on groups of order $n=4^k$. Whereas $\algprobm{Parity}$ is in $\ComplexityClass{NC^1} \subseteq \LOGSPACE$, \algprobm{Matrix Rank} over $\FF_2$ is $\ComplexityClass{Mod_2 L}$-complete \cite{BuntrockDHM92}.\footnote{Although they phrase their results in terms of $\NC^1$ reductions, the reductions are in fact already in $\ACz$ (in fact, we believe even in $\NC^0$, but we do not need that stronger statement for our discussion here so we have not verified it carefully). The relevant reductions are the $\ComplexityClass{Mod_2 L}$-completeness of what they call the \algprobm{Mod-2-Gap} problem, and the reduction from the latter to \algprobm{Matrix Rank}, which follows immediately from their Proposition 9.} In particular, $k$-bit \algprobm{Parity} reduces to $k \times k$ \algprobm{Matrix Rank} over $\FF_2$, so our reduction still does not reduce from a parity on more than $O(\log n)$ bits. By reducing from the harder problem of \algprobm{Matrix Rank} instead, we still incur an exponential size increase, but are able to get our quasi-polynomial lower bound via a combinatorial covering argument (see Section~\ref{sec:ReturningToCNF}).
\end{remark}

\begin{remark}[Lower bounds on \algprobm{Latin Square Isotopy}] \label{rmk:isotopy}
Closely related to \algprobm{(Quasi)Group Isomorphism} is the \algprobm{Latin Square Isotopy} problem, which is essentially a special case of \algprobm{Strongly Regular Graph Isomorphism} (see, e.g., \cite{MillerTarjan, Wolf}) that we now define. We say that two quasigroups are \emph{isotopic} if there exist bijections $\alpha, \beta, \gamma : Q_1 \to Q_2$ such that for all $x, y, z \in Q_1$, $\gamma(xy) = \alpha(x)\beta(y)$. The \algprobm{Latin Square Isotopy} problem takes as input two quasigroups $Q_1, Q_2$, and asks if they are isotopic. There has been some notable work on \algprobm{Latin Square Isotopy} \cite{MillerTarjan, Wolf, LevetLatinSquares, CGLW}. In particular, Collins, Grochow, Levet, and Wei\ss \ \cite{CGLW} established the bound of $
\exists^{\log^2 n}\forall^{\log n}\exists^{\log n}\DTISPpll$,
for \algprobm{Latin Square Isotopy}, which is the same bound they established for \algprobm{Quasigroup Isomorphism}. They also give a more careful analysis of their algorithm which implies that \algprobm{Latin Square Isotopy} is decidable by a uniform family of depth-$4$ $\qACz$ circuits with size $n^{O(\log n)}$. Prior to our work, lower bounds against polynomial-size, depth-$2$ circuits were open for \algprobm{Latin Square Isotopy}, but they follow from our main results.

Namely, Albert \cite{Albert} established that two groups are isotopic if and only if they are isomorphic. Thus, in light of Theorems~\ref{thm:MainLowerBounds} and \ref{thm:MainLowerBoundXOR}, we obtain the following immediate corollary:
\begin{corollary} \label{cor:isotopy}
Any family of depth-$2$ circuits---or $\ComplexityClass{CNF \circ XOR}$ circuits---deciding \algprobm{Latin Square Isotopy} requires size $|G|^{\Omega(\log |G|)}$.         
\end{corollary}
\end{remark}

\subsection{Depth-$3$ Circuits for Group Isomorphism} \label{sec:intro:upper}
We complement our lower bound by also improving the depth of the upper bound on circuits for \algprobm{Group Isomorphism}. To state our result, we recall that circuits of depth $d+1$ with bottom fan-in at most $\poly(\log n)$ are sometimes called ``depth $d + \frac{1}{2}$'' \cite{BBhalf}; the aforementioned circuits \cite{CGLW} are in fact depth $3\frac{1}{2}$ in this nomenclature.

\begin{maintheorem}[High-level version of Theorem~\ref{thm:MainUpperBound}] \label{thm:main}
\algprobm{Group Isomorphism} in the Cayley table model can be solved by uniform circuits of depth $2\frac{1}{2}$ and quasipolynomial size.
\end{maintheorem}

We actually obtain a slightly stronger upper bound, that places \algprobm{Group Isomorphism} into the second level of the polylogarithmic-time hierarchy; see Theorem~\ref{thm:MainUpperBound} for details. Collins, Grochow, Levet, and Wei\ss \ previously established that the more general \algprobm{Quasigroup Isomorphism} problem belongs to the third level of the polylogarithmic time hierarchy, which yielded depth-$3\frac{1}{2}$ circuits. The trade-off compared to previous results is that they get size $n^{O(\log n)}$ \cite{CGLW}, whereas our quasipolynomial size bound is slightly larger.

At a high level, rather than following the generator-enumerator strategy of previous results, our technique is instead to guess a polylogarithmic-length encoding of one group, and then verify that both groups satisfy the same encoding (hence must be isomorphic). As there are at most $n^{(2/27 + o(1))\mu(n)^2}$ groups of order $\leq n$ \cite{Pyber, Higman, Sims}, where $\mu(n) \leq \log_2 n$ is the maximum exponent of any prime-power divisor of $n$, it is at least information-theoretically plausible that groups of order $n$ could be specified by $O(\log^3 n)$ many bits.\footnote{And in general this is asymptotically tight, as there are at least $p^{(2/27) m^3 - O(m^2)}$ groups of order $p^m$ \cite{Higman, Sims}.} As our strategy is to guess this specification---in terms of a circuit, that would be an \textsf{OR} gate of fan-in $2^{O(\log^3 n)}$---then in order to achieve a depth-3 circuit, we would also need to be able to verify those guesses, i.\,e., to check that the specification matches both groups, by a depth-3 circuit of quasipolynomial size that is an \textsf{OR} of \textsf{AND}s of \textsf{OR}s, as then the top \textsf{OR} of the verification circuit can then be merged with the \textsf{OR} gate for the guesses. 

(Already at this high level of description we can note a significant difference from the generator-enumerator approach: whereas the latter also works for \algprobm{Quasigroup Isomorphism}, our approach here cannot give quasipolynomial size circuits for \algprobm{Quasigroup Isomorphism}, by a standard counting argument; see Observation~\ref{obs:quasigroups}.)

We show that a \emph{presentation} (in the group-theoretic sense) of a group can be verified as efficiently as needed for the above strategy to work. This uses the idea from \cite{CGLW} that basic group operations can be computed in deterministic polylogarithmic time from the multiplication table, and the fact that polylogarithmic time can be simulated by quasipolynomial-size depth-2 circuits. However, it is a long-open conjecture, known as the Short Presentation Conjecture, whether presentations of $O(\log^3 |G|)$ size---or even just $O(\log^c |G|)$ for any fixed $c$---exist \cite{BabaiSzemeredi, BGKLP}. 

Babai, Goodman, Kantor, Luks, and Pálfy \cite{BGKLP} introduced several strengthenings of the Short Presentation Conjecture that they anticipated would be useful for algorithmic applications. In particular, they conjectured not only did finite simple groups have such short presentations, but that those presentations could be constructed from the ``standard names'' (Def.~\ref{def:standard-name}) of the finite simple groups in polylogarithmic time (in the order of the group), which they refer to as the Uniform Short Presentation Conjecture \cite[Conjecture~3]{BGKLP}.

\paragraph{The status of the (Uniform) Short Presentation Conjecture.} \label{par:status}
Babai, Goodman, Kantor, Luks, and Pálfy proved \cite[Theorem~8.3]{BGKLP} that if finite simple groups admit presentations of size $O(\log^c n)$, then all groups admit  presentations of size $O(\log^{c+1} n)$; exactly what we would hope for in order to get our strategy above to work to get depth-3 circuits. 
They also proved that the Uniform Short Presentation Conjecture held with $c=2$ for all but three infinite families of finite simple groups \cite{BGKLP},\footnote{This result, as all the similar results we mention, was based on a case analysis, and in that sense, when we say ``all but three infinite families,'' the latter statement depends on the Classification of Finite Simple Groups.} namely the groups denoted
\[
^2 A_2(q) = PSU_3(q), \qquad ^2 B_2(q)=Sz(q), \qquad \text{and} \qquad ^2 G_2(q)=R(q),
\]
for infinitely many prime powers $q$.

The status of the Short Presentation Conjecture for finite simple groups is in a somewhat interesting state of affairs currently. Hulpke and Seress \cite{HulpkeSeress} showed that the first of the above three families satisfies the Short Presentation Conjecture with $c=2$ (though they did not prove uniformity; we prove uniformity of their presentation in Section~\ref{app:PSU3}). They also remark that J.~Thompson told W.~Kantor that Suzuki's original paper \cite{suzuki} on the groups $^2B_2(q)=Sz(q)$ (now known as the Suzuki groups) contained within it already a short presentation. Subsequently, Guralnick, Kantor, Kassabov, and Lubotzky \cite{GKKLquant} exhibited a shorter presentation using only $O(1)$ generators and relators for $Sz(q)$ (they actually established the analogous result for all finite simple groups except $^2G_{2}(q)$). We prove uniformity for this latter presentation in Section~\ref{app:suzuki}. This leaves only the Ree groups, $^2 G_2(q) = R(q)$.

A short presentation for the Ree groups---and thus completion of the proof of the Short Presentation Conjecture---due to Hulpke, Kassabov, Seress (posthumously), and Wilson \cite{HKSW} has been announced in various talks and conferences since 2020, but has not yet appeared in print. 

\paragraph{More on our methods: avoiding the Short Presentation Conjecture for the Ree groups.}
In order to get around this, we take advantage of two additional ingredients: (1) (an adaptation of) the isomorphism test for black-box groups in $\ComplexityClass{\Sigma_{2}TIME}(\polylog(|G|))$ \cite{BabaiSzemeredi}, and (2) a computational version of non-Abelian group cohomology \cite{GQCoho, dedecker, inassaridze}. 

We non-deterministically guess the names of the simple composition factors $SN_{i}$ (for $i=1,\dotsc,\log|G|$), as well as additional data that we use to verify isomorphism (for details see p.~\pageref{guessed-data}). We then verify isomorphism from the bottom up: for $i=1,\dotsc,\log|G|$, we inductively verify that $G_i \cong H_i$, where $G_i$ (resp. $H_i$) is a composition series for $G$: $G_i \unlhd G_{i+1}$ and $G_{i+1} / G_i$ is simple for all $i$ (but we don't compute the composition series, rather we non-deterministically guess the generators of its successive terms). 

For the non-Ree composition factors, we follow the strategy suggested by the aforementioned presentations from \cite{BGKLP}: we guess generators of $G_i / G_{i-1}$, uniformly generate the short presentation for $SN_i$, and then verify that the generators we guessed satisfy the relations from the presentation for $SN_i$. This ensures that $G_i / G_{i-1}$ is a quotient of $SN_i$, hence must be trivial or isomorphic to $SN_i$, since $SN_i$ is simple. To ensure it is actually isomorphic to $SN_i$, we co-nondeterministically try all words in the guessed generators of $G_{i-1}$ to make sure that at least one of the generators we guessed for $G_{i}$ is not in $G_{i-1}$, hence that $G_i / G_{i-1}$ is not the trivial group.
(Some of the additional data is then used to verify that not only are these ``top'' composition factors isomorphic, but in fact $G_i \cong H_i$.) 

But when there are Ree composition factors, we encounter a new obstacle. For factors for which we guessed $SN_i$ is a Ree group, in the absence of a (published) short presentation, we instead use standard matrix generators to specify ${}^2G_{2}(q)$ (see, e.\,g., \cite{Brnhielm2014}), and compare this to a composition factor $G_{i}/G_{i-1}$ from our non-deterministic guess. In this case, this check verifies that the quotient $G_i / G_{i-1}$ surjects onto the Ree group ${}^2 G_2(q)$, but not (yet) that they are isomorphic. This check is adapted from the $\ComplexityClass{\Sigma_{2}TIME}(\polylog(|G|))$ isomorphism test for black-box groups due to Babai and Szemer\'edi \cite{BabaiSzemeredi}. In their setting, Babai and Szemer\'edi assumed an oracle for deciding whether an element is \emph{not} the identity, but in our setting we can only verify when an element \emph{is} the identity, resulting in the preceding statement about a surjection rather than isomorphism. 
However, we can also check (using uniformly-generated formulas for the orders of simple groups, including the Ree groups) that the orders are correct: $\prod_i |SN_i| = |G|$. Since our checks ensure that $|G_i / G_{i-1}| \geq |SN_i|$ for all $i$, the fact that product correctly multiplies to $|G|$ then implies that we must have $|G_i / G_{i-1}| = |SN_i|$ for all $i$ (including when $SN_i$ is a Ree group), thus ensuring we have correctly identified the simple isomorphism types of the composition series.

The final ingredient is how to certify not just that $G_i / G_{i-1} \cong SN_i \cong H_i / H_{i-1}$, but that in fact $G_i \cong H_i$. In both the Ree and non-Ree cases this essentially relies on non-Abelian cohomology. In the non-Ree case, the needed results are packaged nicely in \cite{BGKLP}. To say what additional is needed in the Ree case, we outline the technique. The basic idea is this: if $G_i / G_{i-1}$ has a presentation $\langle x_1,\dotsc,x_k | r_1(\vec{x}) = \dotsb = r_m(\vec{x})=1 \rangle$, then the multiplication in $G_i$ is fully determined by the following data:
\begin{itemize}
\item The isomorphism type of $G_i / G_{i-1}$

\item The isomorphism type of $G_{i-1}$

\item The (outer) action of $G_i / G_{i-1}$ on $G_{i-1}$ by conjugation in $G_i$

\item A function $\gamma_i \colon [m] \to G_{i-1}$, that gives, for each word $r_i$ in the presentation of $G_i / G_{i-1}$, an associated element of $G_{i-1}$. 
\end{itemize}
Taken together, these let us compute multiplication in $G_i$ (and hence fully determine its isomorphism type) by a standard word-rewriting technique. When the simple group $G_i / G_{i-1}$ has a short presentation, then we can use poly-logarithmically many bits to guess the entire function $\gamma_i$. But in the Ree case we cannot.

In the Ree case, we need an additional uniform normal form for matrices in the standard presentation of the Ree groups (Lemma~\ref{lem:Ree-bijection}), and this seems crucial for the approach to work. It essentially allows us to guess and check a uniform version of the function $\gamma_i$, but where the relations $r_i$ are now \emph{all} relations of the form $abc^{-1} = 1$, i.\,e., encoding every entry $ab=c$ of the multiplication table of the Ree group.
In this case, we use more directly results from non-Abelian cohomology \cite{GQCoho, dedecker, inassaridze}. We also note that the uniform normal form of Lemma~\ref{lem:Ree-bijection} may be useful for proving uniformity of a (future) short presentation of the Ree groups; see Section~\ref{sec:conclusion} for more discussion on this point.

\paragraph{Difference compared to the generator-enumerator strategy.}
As further technical evidence of the difference between our approach and the previous generator-enumerator results, we recall a significant obstacle to extending our approach from groups to quasigroups. Namely, a quasigroup analogue of the Short Presentation Conjecture is false, by a standard counting argument that we include here for completeness:

\begin{observation} \label{obs:quasigroups}
Any injection from the set of isomorphism classes of quasigroups of order $n$ to binary strings of length $\ell$ must have $\ell \geq \Omega(n^2 \log n)$.
\end{observation}

In particular this says that quasigroups up to isomorphism cannot be specified, in the general case, more succinctly than by writing down their multiplication table, at least up to constant factors, regardless of the method of encoding (table, presentation, Turing machine, etc.). 

\begin{proof}
The number of Latin squares of order $n$ is at least $(n!)^{2n} / n^{n^2}$ (e.\,g., \cite[Thm.~17.2]{vanLintWilsonBook}), hence the number of quasigroups of order $n$ up to isomorphism is at least 
\begin{align*}
(n!)^{2n-1} / n^{n^2} & \sim  \left((n/e)^n \sqrt{2\pi n}\right)^{2n-1} / n^{n^2} \\
 & = n^{n^2- 1/2} \sqrt{2\pi}^{2n-1} / e^{2n^2 - n} \\
 & \geq (n/c)^{\Theta(n^2)}
\end{align*}
for some constant $c \approx e^2 / 1.01$. Thus, in order for there to be at least as many strings of length $\ell$ as there are quasigroups of order $n$ (up to isomorphism), we must have $\ell \geq \log (n/c)^{\Theta(n^2)} \geq \Omega(n^2 \log n)$.
\end{proof}

\subsection{Related Work} \label{sec:related}

\paragraph{Generator-enumeration strategy.} The generator-enumeration strategy was independently discovered by Felsch and Neubüser \cite{FN}
and by Tarjan; see \cite{MillerTarjan}. Miller \cite{MillerTarjan} extended this observation to the setting of quasigroups, obtaining the $\exists^{\log^2 n}\textsf{P}$ bound for \algprobm{Quasigroup Isomorphism}, and much of the subsequent work on this strategy also works at that level of generality. We summarize the history in Table~\ref{table:history}, which is from \cite{CGLW} (except for adding a new entry for the present paper). We note that the upper bound of Chattopadhyay, Tor\'an, Wagner \cite{ChattopadhyayToranWagner} was the first upper bound showing that \algprobm{Parity} did not $\ACz$-reduce to \algprobm{GpI} (allowing them to show that \algprobm{GI} did not $\ACz$-reduce to \algprobm{GpI}).

Recently, Collins--Grochow--Levet--Wei\ss\ \cite{CGLW} exhibited a bound of 
\[
\exists^{\log^2 n} \forall^{\log n} \exists^{\log n} \DTISPpll \subseteq \exists^{\log^2 n}\LogSpace \cap \exists^{\log^2 n}\FOLL,
\]
for \algprobm{Quasigroup Isomorphism}. In addition to improving upon the work of Chattopadhyay--Tor\'an--Wagner \cite{ChattopadhyayToranWagner}, the work of Collins--Grochow--Levet--Wei\ss\ yielded depth-$3\frac{1}{2}$ circuits of size $n^{O(\log n)}$. In contrast, our bound requires non-deterministically guessing $O(\log^c n)$ bits, for some constant $c > 2$ that we do not estimate carefully. Thus, even for \algprobm{Group Isomorphism}, although we improve the depth, it is at the cost of increasing the size.

\renewcommand{\arraystretch}{1.5}
\begin{table}[!htbp]
\small
\begin{center}
\begin{adjustbox}{max width=\textwidth}
\begin{tabular}{|r|c|c|l|}

\hline
\textbf{Year} & \textbf{Result} & \textbf{Depth} & \textbf{Citation} \\ \hline \hline
1970 & Generator-enumerator introduced & $\poly(n)$ & Felsch \& Neubüser \cite{FN}\tablefootnote{Complexity not analyzed there, but the same as Tarjan's algorithm \cite{MillerTarjan}.} \\ \hline
1978 & $\betacc{2} \P \subseteq \mathsf{DTIME}(n^{\log n + O(1)})$ & $\poly(n)$ & Tarjan (see Miller \cite{MillerTarjan})  \\ \hline
1977 & $\mathsf{DSPACE}(\log ^2 n)$ & $O(\log^2 n)$ & Lipton--Snyder--Zalcstein \cite{LiptonSnyderZalcstein}\tablefootnote{Despite the publication dates, this seems to have been independent of Tarjan's result. They note that Miller and Rabin had also observed this result independently.}  \\ \hline
 1994 & $\betacc{2}\mathsf{AC}^1$ & $O(\log n)$ & Wolf \cite{Wolf}\tablefootnote{\label{fn:wolf}Wolf only claims a bound of $\betacc{2}\mathsf{NC}^{2}$. However, if we replace his use of $\mathsf{NC}^{1}$ circuits to multiply two elements of a quasigroup with $\ACz$ circuits, we immediately get the $\betacc{2}\mathsf{AC}^1$ bound.} \\ \hline
 2010 & $\betacc{2}\mathsf{SAC}^1$ & $O(\log n)$ & Wagner \cite{WagnerThesis} \\ \hline
 2010 & $\betacc{2}\FOLL \cap \betacc{2} \LogSpace$ & $O(\log \log n)$ & Chattopadhyay--Torán--Wagner \cite{ChattopadhyayToranWagner}\tablefootnote{They do not claim the $\betacc{2}\LogSpace$ bound, but it follows immediately from their algorithm and results.} \\ \hline
 2013 & $\betacc{2}\mathsf{SC}^2 \cap \betacc{2}\LogSpace$ & $O(\log n)$ & Papakonstantinou--Tang--Qiao \cite{TangThesis}\tablefootnote{We have written the result this way, despite $\betacc{2} \LogSpace \subseteq \betacc{2}\mathsf{SC}^2$, because in Tang's thesis \cite{TangThesis}, the only place this is currently published, they only claim $\mathsf{NSC}^2$ using only $O(\log^2 n)$ bits of nondeterminism, which in our notation would be $\betacc{2}\mathsf{SC}^2$. However, their algorithm and results also immediately yields a $\betacc{2}\LogSpace$ bound.}  
 \\ \hline
 2024 & $\ComplexityClass{\Sigma_3 TISP}(\polylog(n),\log(n)) \subseteq \qACz$
 & $3\frac{1}{2}$ & Collins--Grochow--Levet--Wei\ss\xspace\cite{CGLW} \\ \hline
 2026 & 
 $\ComplexityClass{\Sigma_2 TIME}(\polylog(n)) \subseteq \qACz$
& $2\frac{1}{2}$ & This work \\ \hline
\end{tabular}
\end{adjustbox}
\end{center}

\caption{\label{table:history} History of the low-level circuit complexity of algorithms for \algprobm{(Quasi)Group Isomorphism} based on the generator-enumerator technique, in comparison to this paper (all but the last row of the table are from \cite[Table~1]{CGLW}). For non-circuit classes, we list their depth as the best-known depth of their simulation by circuits. The class in the bound obtained by Collins--Grochow--Levet--Wei\ss \ \cite{CGLW} is contained in all previous classes listed in the table, which also all apply to \algprobm{Quasigroup Isomorphism}. Compared to Collins--Grochow--Levet--Wei\ss, our work improves the circuit depth, but at the cost of increasing the size (in the exponent of the exponent of the quasipolynomial). 
}  
\end{table}

\paragraph{Short presentations.} See Section~\ref{sec:Presentations} for background on group presentations. We say that a presentation $P$ for a group $G$ is \emph{short} if $P$ has length $\polylog(|G|)$. Babai and Szemerédi \cite{BabaiSzemeredi} first introduced the \emph{Short Presentation Conjecture}, which asks whether there exists an absolute constant $c$, such that every finite simple group of order $n$ has length $\leq (\log n)^c$. Babai, Goodman, Kantor, Luks, and Pálfy \cite{BGKLP} subsequently investigated the Short Presentation Conjecture. Following them, we say that presentations of finite simple groups are \emph{uniform} if they can be constructed in $\DTIMEpl$ from the standard names for the finite simple groups. We have mentioned in Section~\ref{sec:intro:upper} the results on these conjectures most relevant for our paper \cite{BabaiSzemeredi,BGKLP,HulpkeSeress,suzuki,HKSW,GKKLquant}. 

Prior to this work, the existing short presentations for $\PSU_{3}(q)$ and $\Sz(q)$ had not been shown to be uniform. We establish uniformity for both $\PSU_{3}(q)$ (Proposition~\ref{prop:UniformityPSU}) and $\Sz(q)$ (Proposition~\ref{prop:UniformSuzuki}).

\paragraph{Parallel algorithms for special cases of \algprobm{Group Isomorphism}.} The work on $\textsf{NC}$ algorithms for \algprobm{GpI} and \algprobm{QGpI} is comparatively nascent compared to that of \algprobm{GI} (see \cite{LevetRombachSieger} for a survey of $\textsf{NC}$ algorithms for \algprobm{GI}). Indeed, much of the work involves parallelizing the generator-enumeration strategy (\emph{ibid.}), which has yielded bounds of $\textsf{L}$ \cite{TangThesis} for $O(1)$-generated groups and $\textsf{SAC}^{1}$ for $O(1)$-generated quasigroups \cite{WagnerThesis}. Other families of groups known to admit $\textsf{NC}$ isomorphism tests include Abelian groups \cite{ChattopadhyayToranWagner, GrochowLevetWL, CGLW}, graphical groups arising from the CFI graphs \cite{WLGroups, CollinsLevetWL, CollinsUndergradThesis}, coprime extensions $H \ltimes N$ where $H$ is $O(1)$-generated and $N$ is Abelian \cite{GrochowLevetWL} (parallelizing a result from \cite{QST11}), groups of almost all orders \cite{CGLW} (parallelizing \cite{DietrichWilson}), and Fitting-free groups where the number of non-Abelian simple factors of the socle is $O(\log n/\log \log n)$ \cite{GrochowLevetWL} (parallelizing a result from \cite{BCGQ}). Grochow, Johnson, and Levet recently exhibited an $\textsf{AC}^{3}$ isomorphism test for the class of all Fitting-free groups \cite{GrochowJohnsonLevet}. 

Johnson, Levet, Vojtěchovský, and Widholm \cite{JLVWQuasigroups} recently showed that in the multiplication table model, a fully-refined direct product decomposition of a group can be computed in $\textsf{AC}^{3}$. They also exhibited an $\textsf{NC}$ isomorphism test for central quasigroups. To the best of our knowledge, no other special family of quasigroups that are not groups (beyond $O(1)$-generated quasigroups) are known to admit even a polynomial-time isomorphism test.

Levet \cite{LevetCodeEq} recently exhibited an $\textsf{AC}^{3}$ isomorphism test for coprime extensions $H \ltimes N$ where $H$ is elementary Abelian and $N$ is Abelian (parallelizing the remainder of \cite{QST11}). Additionally, Levet exhibited an $\textsf{AC}^{3}$ isomorphism test for central-radical groups where $\text{Rad}(G)$ was elementary Abelian; and $G/\text{Rad}(G)$ was a direct product of either (i) non-Abelian simple groups, or (ii) perfect groups of bounded size (parallelizing a result of Grochow and Qiao \cite{GQCoho}). \\

\noindent \textbf{First-order encodings of groups: the work of Nies and Tent \cite{NiesTent2014}.} Instead of using a group presentation as its short encoding, one could use more general first-order sentences in the language of groups, in the vein of descriptive complexity. In that setting, Nies and Tent showed an analogue of Short Presentation Conjecture: that all groups have polylogarithmic-length first-order descriptions \cite{NiesTent2014}. 

We observe here that one can directly get a \qACz upper bound from their results; however, we will also discuss the obstacle we faced in getting their results to yield depth-3 circuits. Since it is not used elsewhere in our paper, we refer to their paper for any needed background and definitions.

A first-order sentence identifies a group $G$ if $G$ is the only group, up to isomorphism, that satisfies that sentence. Nies and Tent \cite[Thm.~7.3]{NiesTent2014} proved that there is a constant
\(m\) such that every finite group \(G\) has a \(\Sigma_m\)-sentence
\(\varphi_G\) of length \(O(\log^4 |G|)\) that identifies \(G\). Hence, for two
finite groups \(G,H\), we have the following:
\[
G\not\cong H \iff \exists \varphi \; \bigl[ |\varphi| \le c\log^4 n,\ \varphi\in\Sigma_m,\ G\models\varphi,\ H\not\models\varphi \bigr],
\]
where \(n\) denotes the input size and \(c\) is a suitable constant. Indeed, if \(G\not\cong H\), we may take \(\varphi=\varphi_G\).

\begin{observation}
Theorem 7.3 of Nies and Tent \cite{NiesTent2014} implies that \algprobm{Group Isomorphism} is in \qACz.
\end{observation}

\begin{proof}
For each fixed candidate sentence \(\varphi\), the predicates \(G\models\varphi\)
and \(H\models\varphi\) are computable by constant-depth quasipolynomial-size
circuits via the standard quantifier-to-gate translation: each existential
quantifier block becomes an \(\mathsf{OR}\), each universal block becomes an
\(\mathsf{AND}\), and atomic formulas are evaluated in constant depth from the
multiplication tables. Since \(|\varphi| \leq O(\log^4 n)\) and the number of
quantifier alternations is bounded by some constant \(m\), this yields a
\qACz circuit for the inner test. Finally, the outer existential over $\varphi$ ranges over only \(2^{(\log n)^{O(1)}}\) candidates, so it can be implemented by one additional $\mathsf{OR}$ gate of quasipolynomial fan-in, preserving \qACz. 
\end{proof}

We note, however, that they estimate a bound of $m \leq 10$ (for most of their paper they can get $\Sigma_3$ sentences, but not for the Ree groups). However, even if they could get $m$ down to $3$, we do not see how to get circuits better than depth $3\frac{1}{2}$ out of this approach: the outer existential to guess $\varphi$ can be merged with the top \textsf{OR} gate in evaluating $\varphi$, but then evaluating the atomic (unquantified) formulas at the ``bottom'' of $\varphi$ seems to require $\DTISPpll$, which increases the depth by one, to 4 ($3\frac{1}{2}$ because the bottom fan-in is only polylogarithmic).\footnote{In comparison to \cite{CGLW}, this would be depth $3\frac{1}{2}$ circuits of size $n^{O(\log^3 n)}$, whereas \cite{CGLW} get depth-$3\frac{1}{2}$ circuits of size $n^{O(\log n)}$.}

\paragraph{Lower bounds for \algprobm{Group Isomorphism}.} We are not aware of prior lower bounds on \algprobm{GpI} on any class of Boolean circuits; nonetheless, there are several lower bounds in more restricted algebraic models of computation, which we review here.
Gowers \cite{GowersSubgroupProfiles}, reiterated by Babai \cite[Page~81]{BabaiGraphIso}, asked whether there exists some fixed constant $k$, such that the profiles of subgroups generated by at most $k$ elements would suffice to place \algprobm{Group Isomorphism} into $\textsf{P}$. That is: are all groups determined up to isomorphism by the multiset of isomorphism types of their $k$-generated subgroups? Glauberman and Grabowski \cite{GlaubermanGrabowski} refuted this claim, by showing that $k \geq \sqrt{2 \log_{3} |G|} - 5/2$ was necessary. Wilson \cite{WilsonSubgroupProfiles} subsequently improved this bound to $k \geq \log_{3} |G| - 2$. It is remarkable how tight this bound is: every group is generated by at most $\log_p|G|$ generators, where $p$ is the smallest prime dividing $|G|$.

Grochow and Levet \cite{GrochowLevetWL} showed that there exists an infinite family of Abelian groups $(G_m)_{m \in \mathbb{N}}$ such that any $\textsf{FO}$ (first-order) formula identifying $G_{m}$ requires $\Omega(\log |G_{m}|)$ variables. Collins and Levet \cite{CollinsLevetWL} subsequently extended this result to a higher-arity analogue of the Ehrenfeucht--Fra\"iss\'e pebble game corresponding to $\textsf{FO}$. 

Bshouty \cite{BshoutyAbelianGpI} recently established a query lower bound of $\Omega(|G|)$ for isomorphism testing of Abelian groups given by their multiplication tables, in the ``algebraic'' model where queries are to group elements and operations are the natural group operations (not bit-wise).

The Abelian groups used for the aforementioned lower bounds in finite model theory \cite{GrochowLevetWL, CollinsLevetWL} and query complexity \cite{BshoutyAbelianGpI} are essentially the same family of groups we use to get our CNF lower bound: products of many copies of the cyclic groups $\mathbb{Z}_2$ and $\mathbb{Z}_4$. However, for our CNF lower bound we need them in a particular form, in order to get a reduction from \algprobm{Matrix Rank}.

\section{Preliminaries}
Throughout, for a finite group $G$ we write $n:=|G|$ for its order. The symmetric group, consisting of permutations of a set of size $m$, is denoted $\operatorname{Sym}(m)$. $\mathrm{GL}(d,q)$ denotes the general linear group consisting of the $d \times d$ invertible matrices over the finite field $\FF_q$. 
For matrix groups over finite fields, we write $q=p^e$, where $p$ is prime and $e \geq 1$, and we use $d$ for the matrix dimension, writing for example $\mathrm{GL}(d, q)$ and $ \mathrm{SL}(d, q)$. Other auxiliary integer parameters will be denoted by letters such as $r, s, t, k$, and not by $n,d,p,q,e$,  
which we try to reserve for the meanings above.

\subsection{Complexity theory}
We assume that the reader is familiar with standard complexity classes such as $\textsf{P}, \textsf{NP}, \textsf{L}$, and $\textsf{NL}$. 
For a standard reference on circuit complexity, see \cite{VollmerText}. We consider Boolean circuits using the gates \textsf{AND}, \textsf{OR}, and \textsf{NOT}. 

In this paper, we will consider $\DTIME(\log^c n)$-uniform circuit families $(C_{n})_{n \in \mathbb{N}}$, for some fixed $c \geq 1$. For this,
one encodes the gates of each circuit $C_n$ by bit strings of length $O(\log^c n)$. Then the circuit family $(C_n)_{n \geq 0}$
is called \emph{$\DTIME(\log^c n)$-uniform}  if (i) there exists a deterministic Turing machine that computes for a given gate $u \in \{0,1\}^*$
of $C_n$ ($|u| \in O(\log^c n)$) in time $O(\log^c n)$ the type of gate $u$, where the types are $x_1, \ldots, x_n$, \textsf{NOT}, \textsf{AND}, or \textsf{OR} gates,
and (ii) there exists a deterministic Turing machine that decides for two given gates $u,v \in \{0,1\}^*$ of $C_n$ ($|u|, |v| \in O(\log^c n)$) and a binary encoded integer $i$ with $O(\log^c n)$ many bits in time $O(\log^c n)$ whether $u$ is the $i$-th input gate for $v$. When $c = 1$, this notion of uniformity is referred to as $\textsf{DLOGTIME}$-uniformity. For circuit families of size $\poly(n)$, we will use $\textsf{DLOGTIME}$-uniformity.

\begin{definition}
Fix $k \geq 0$. We say that a language $L$ belongs to (uniform) $\textsf{AC}^{k}$ if there exist a (uniform) family of circuits $(C_{n})_{n \in \mathbb{N}}$ over the $\textsf{AND}, \textsf{OR}, \textsf{NOT}$ basis such that the following hold:
\begin{itemize}
\item The $\textsf{AND}$ and $\textsf{OR}$ gates have unbounded fan-in.
\item $C_{n}$ has depth $O(\log^{k} n)$ and uses (has size) $n^{O(1)}$ gates. Here, the implicit constants in the circuit depth and size depend only on $L$.

\item $x \in L$ if and only if $C_{|x|}(x) = 1$. 
\end{itemize}
Note that the AND and OR gates can have arbitrary fan-in.
\end{definition}

\noindent The complexity class $\NC^{k}$ is defined analogously as $\textsf{AC}^{k}$, except that the $\textsf{AND}, \textsf{OR}$ gates are required to have fan-in at most 2. The class $\SAC^k$ is defined analogously, in which the $\textsf{OR}$ gates have unbounded fan-in but the $\textsf{AND}$ gates must have fan-in $2$. We also allow circuits to compute functions by using multiple output gates. 

For every $k$, the following containments are well-known:
\[
\textsf{NC}^{k} \subseteq \SAC^k \subseteq  \AC^{k} \subseteq \textsf{NC}^{k+1}.
\]

\noindent In the case of $k = 0$, we have that:
\[
\textsf{NC}^{0} \subsetneq \AC^{0} \subsetneq \textsf{NC}^{1} \subseteq \LogSpace \subseteq \textsf{NL} \subseteq \textsf{SAC}^{1} \subseteq \AC^{1}.
\]

\noindent We note that functions that are $\textsf{NC}^{0}$-computable can only depend on a bounded number of input bits. Thus, $\textsf{NC}^{0}$ is unable to compute the $\textsf{AND}$ function. It is a classical result that $\AC^{0}$ is unable to compute \algprobm{Parity} \cite{FSS,Ajtai,Yao}. 

We will also be interested in $\ComplexityClass{NC}$ and $\ComplexityClass{AC}$ circuits of quasipolynomial size (i.\,e., $2^{O(\log^k n)}$ for some constant $k$). For a circuit class $\mathcal{C} \subseteq \ComplexityClass{NC}$, the analogous class permitting a quasipolynomial number of gates is denoted $\ComplexityClass{quasi}\mathcal{C}$. Note that \DLOGTIME uniformity does not make sense for $\ComplexityClass{quasiNC}$, as we cannot encode gate indices using $O(\log n)$ bits. Instead, we will use $\DTIMEpl$-uniformity for $\ComplexityClass{quasiNC}$ \cite{BarringtonQuasipolynomial,FerrarottiGonzalezScheweTurull}. 

\paragraph{Bounded nondeterminism.} For a complexity class $\mathcal{C}$, we define $\betacc{i}\mathcal{C}$ to be the set of languages $L$ such that there exists an $L' \in \mathcal{C}$ such that $x \in L$ if and only if there exists $y$ of length at most $O(\log^{i} |x|)$ such that $(x, y) \in L'$. Similarly, define $\alphacc{i}\mathcal{C}$ to be the set of languages $L$ such that there exists an $L' \in \mathcal{C}$ such that $x \in L$ if and only if for all $y$ of length at most $O(\log^{i} |x|)$, $(x,y) \in L'$. For any $i \geq 0$ and any $c \geq 0$, both $\betacc{i}\FOLL$ and $\alphacc{i}\FOLL$ are contained in  $\mathsf{quasi}\FOLL$, and so cannot compute \algprobm{Parity} \cite{ChattopadhyayToranWagner, Smolensky87algebraicmethods}. 

\paragraph{Time- and space-restricted Turing machines.}
When considering complexity classes defined by Turing machines with a time bound $t(n) \in o(n)$, we use Turing machines with random access to the input tape, which is read-only, and a separate address (or \emph{index}) tape. After writing an address, the machine can go to a query state reading the symbol from the input at the location specified by the binary-encoded integer on the address tape. As usual, the machines are otherwise allowed to have multiple work tapes. 

For functions $t(n), s(n) \in \Omega(\log n)$, the class $\DTISP(t(n),s(n))$ is defined to consist of decision problems computable by deterministic $t(n)$-time and $s(n)$-space-bounded Turing machines. Be aware that there must be one Turing machine that simultaneously satisfies the time and space bound. For details we refer to \cite[Section 2.6]{VollmerText}. For further reading on the connection to \qACz, we refer to \cite{BarringtonQuasipolynomial,FerrarottiGonzalezScheweTurull}.

We frequently use $\DTISPpll$ and $\DTIMEpl$.
However, because of how small the time and space bounds are for these classes, when we abuse notation to say some function (not necessarily decision problem) is computable in $\DTISPpll$, there is some ambiguity as to what this might mean. We use the following two equivalent definitions. First we need some setup. For a function $f$ from bit-strings to bit-strings, we define its \emph{bit function} $\text{bit-}f$ as follows. Write $f(x)_i$ to denote the $i$-th bit of $f(x)$. Then define
\[
\text{bit-}f(x,i) = \begin{cases}
f(x)_i & i \leq |f(x)| \\
\bot & i > |f(x)|.
\end{cases}
\]
Note that $\text{bit-}f$ is always a function which we may take to have at most two output bits (e.\,g., by encoding $f(x)_i$ by repeating the bit, and encoding $\bot$ by $01$), i.\,e., a pair of decision problems.

\begin{deflem}[{\cite[Definition/Lemma~2.4]{CGLW}}] \label{def:DTISP}
For any function $f \colon \{0,1\}^* \to \{0,1\}^*$ with $|f(x)| \leq O(\log |x|)$, the following are equivalent:
\begin{enumerate}
\item $f$ is computable by a Turing machine, using $\poly(\log n)$ time and $O(\log n)$ space, that halts with the result written in a specified place on its work tape.

\item The two bits of $\text{bit-}f$ are each decision problems in $\DTISPpll$, i.\,e., can each be decided by a Turing machine using $\poly(\log n)$ time and $O(\log n)$ space. 
\end{enumerate}
In either case, we say $f$ is computable in $\DTISPpll$.
\end{deflem}

The same result and proof work up to $|f(x)| \leq \poly(\log|x|)$, if we provide the machine with a separate output tape whose size is not counted towards the size bound.

\begin{fact}[{cf. \cite[Fact~2.5]{CGLW}}]\label{fact:DTISP} 
For any $t(n) \leq n$,\footnote{Fact~\ref{fact:DTISP} continues to hold when $t(n) > n$, but tighter statements are available in that setting. In this paper we will only use this in the setting where $t(n) \leq \polylog(n)$, so we do not bother specifying this in detail.} any decision problem in $\DTIME(t(n))$ is computable by  
decision trees of depth $t(n)$,
 hence by CNFs with $2^{t(n)}$ clauses each of width $t(n)$, as well as by DNFs with $2^{t(n)}$ terms each of width $t(n)$.
\end{fact}

We will now recall the polylogarithmic-time hierarchy.

\begin{definition}
For $k \in \mathbb{N}$, $\ComplexityClass{\Sigma_{k}TIME}(\polylog(n))$ denotes the class of languages $L$, such that:
\[
\omega \in L \iff \exists{x_1} \forall{x_2} \exists{x_3} \cdots Qx_k \, \textsf{DTIME}(\polylog(|\omega|)).
\]
Here, the quantifiers alternate starting with an existential quantifier. In particular, if $k$ is even, then $Q$ is a universal quantifier. Otherwise, $Q$ is an existential quantifier. 

We similarly define $\ComplexityClass{\Pi_{k}TIME}(\polylog(n))$ to be the class of languages $L$ such that:
\[
\omega \in L \iff \forall{x_1} \exists{x_2} \forall{x_3} \cdots Qx_k \, \textsf{DTIME}(\polylog(|\omega|)).
\]
Here, the quantifiers alternate starting with an universal quantifier. In particular, if $k$ is even, then $Q$ is an existential quantifier. Otherwise, $Q$ is a universal quantifier. 

In both the settings of $\ComplexityClass{\Sigma_{k}TIME}(\polylog(n))$ and $\ComplexityClass{\Pi_{k}TIME}(\polylog(n))$, we note that the certificates $x_1, \ldots, x_k$ all may be taken to have length $\polylog(n)$, so in particular we have:
\[
\ComplexityClass{\Sigma_{k}TIME}(\polylog(n)) = \exists^{\polylog(n)}\forall^{\polylog(n)} \dotsb Q^{\polylog(n)}\DTIMEpl.
\]

The \emph{polylogarithmic-time hierarchy} is: $\bigcup_{k \in \mathbb{N}} \ComplexityClass{\Sigma_{k}TIME}(\polylog(n))$.
\end{definition}

\paragraph{Reductions.} A \emph{many-one} reduction $A \leq_m B$ is a function $r$ such that for all strings $x$, $A(x) = B(r(x))$; the many-one reduction lies in a complexity class $\mathcal{C}$ if $r \in \mathcal{C}$. A many-one reduction $r$ is a \emph{projection} if each output bit of $r(x)$ is either a constant or some bit $x_i$ of the input (each $x_i$ may be used zero or more times). We say the function $n \mapsto \max_{|x|=n} |r(x)|$ is the \emph{stretch} of $r$. Note that if $r$ is a projection reduction of stretch $s(n)$ (of any complexity) and $B$ is solvable by non-uniform CNFs (resp., DNFs) of size $S(n)$, then $A$ is solvable by non-uniform CNFs (DNFs) of size $S(s(n))$. In particular, if $r$ has polynomial stretch, and $B$ has quasi-polynomial size CNFs, then so does $A$. If $r$ is uniformly computable, then the analogous statements holds for uniform CNFs (resp. DNFs). 

We say that a many-one reduction $r$ is an $\FF_2$-\emph{affine projection}
if, after identifying its input and output strings with vectors over
$\FF_2$, every output bit of $r(x)$ is an affine linear function of the
input bits. This means that for every output coordinate $j$ there are
$a_{j,0},a_{j,1},\ldots,a_{j,n}\in\FF_2$ such that
$
    r(x)_j=a_{j,0}+\sum_{i=1}^n a_{j,i}x_i
    \pmod 2.
$
Differently put, $r(x)=Ax+b$ for some matrix $A$ over $\FF_2$ and some
vector $b$. We call a many-one reduction whose reduction map has this
form an $\FF_2$-affine projection reduction.

\subsection{Group theory} \label{sec:GroupTheory}
\noindent For a standard reference, see \cite{Robinson1982}. All groups will be assumed to be finite, unless otherwise specified. For $g, h \in G$, the \emph{commutator} $[g, h] := ghg^{-1}h^{-1}$. The \emph{commutator subgroup} $[G, G] := \langle \{ [g, h] : g, h \in G \} \rangle$. 

A nontrivial group $G$ is \emph{simple} if the only normal subgroups are $1, G$. A \emph{composition series} for the group $G$ is an ascending chain of subgroups
$1=G_0 \triangleleft G_1 \triangleleft \cdots \triangleleft G_k = G$ such that for each $0 \leq i < k$, $G_{i+1}/G_{i}$ is simple. The multiset $\{\!\!\{ G_{i+1}/G_{i} : 0 \leq i < k\cdot \}\!\!\}$ is the collection of \emph{composition factors}. While a group may have many different composition series, the Jordan--H\"older
Theorem states that any two composition series of $G$ have the same length and same multiset of composition factors (up to isomorphism), though the composition factors need not appear in the same order between the two series.

The Classification of Finite Simple Groups says that every finite simple group is either alternating, of Lie type, or one of the 26 sporadic groups. Thus each finite simple group belongs to one of finitely many uniform families (each being indexed by one or two parameters: an integer $d$ and a prime power $q$), together with finitely many  sporadic cases. Because of its utility in algorithms, we recall a convention from \cite{BGKLP}:

\begin{definition}[{Standard names of finite simple groups, \cite{BGKLP}}] \label{def:standard-name}
The \emph{standard name} of a finite simple group $G(q)$ includes: (1) a string indicating which family it is from (e.g. $A$, $B$, $C$, $D$, $^2 G_2$, etc.), (2) if it is from a family of unbounded rank, the $d$ indexing which element of the family (viz. $A_d$, $B_d$, etc.), (3) the prime $p$ such that $q=p^e$, together with an irreducible polynomial over
$\FF_p$ defining the field occurring in the standard matrix realization of $G(q)$,
and a primitive root of that field. For untwisted groups this field is $\FF_q$;
for twisted groups it may instead be $\FF_{q^2}$ or $\FF_{q^3}$.
(The standard name of a sporadic group only includes part (1).)
\end{definition}

The alternating groups of degree $m$ have order $m!/2$. The Lie-type finite simple groups with parameters $(d,q)$ have order $q^{\Theta(d^2)}$, and those without the parameter $d$ have order $q^{O(1)}$. The standard name of a finite simple group $S$ thus always has bit length at most $O(\log |S|)$.

Let $G$ be a group and let $S=\{g_1,\dots,g_m\}\subseteq G$. A \emph{straight-line program} (SLP) over $S$ is a finite sequence $h_1,\dots,h_t$ of elements of $G$ such that each $h_i$ is one of the following:  a generator $g_j\in S$; the inverse of an earlier term, $h_r^{-1}$ for some $r<i$; or the product of two earlier terms, $h_rh_s$ for some $r,s<i$. The \emph{value} of the SLP is its final term $h_t$. We say that an element $g\in G$ has an SLP of length $t$ over $S$ if there is such a program with final value $g$, and we say a subset $T \subseteq G$ has an SLP of length $t$ over $S$ if there is such an SLP such that the set of elements appearing in the SLP contains $T$. Babai and Szemerédi \cite{BabaiSzemeredi} established a \emph{Reachability Lemma}, showing that every element $h \in H := \langle S \rangle$ can be generated by an SLP of length at most $(1 + \log |H|)^2$. 

\newcommand{\coNTISPpllMV}{\ComplexityClass{coNTISP}(\polylog(n),\log n) \ComplexityClass{MV}}
We record here for use later a lemma whose proof uses the same idea as \cite[Cor.~3.2]{CGLW}. Given a binary relation $R \subseteq X \times W$, we say that $R$ is computable in $\coNTISPpllMV$ (cf. $\ComplexityClass{NPMV}$, \cite{BookLongSelman, selman}) if there is a co-nondeterministic Turing machine using time $\polylog(n)$ and space $O(\log n)$ that, on input $x \in X$, each co-nondeterministic branch either outputs a special symbol $\dagger$, or outputs some $w \in W$ such that $(x,w) \in R$, and furthermore, if there exists $w \in W$ such that $(x,w) \in R$, then at least one co-nondeterministic branch outputs some such $w$. 

\begin{lemma}[Discrete logarithm in groups given by Cayley table] \label{lem:discrete-log}
Given the multiplication table of a group $G$, discrete logarithms in $G$ can be computed in $\coNTISPpllMV$. (That is, the binary relation $\{((G,g,h), k) : g, h \in G, g^k = h\}$ is computable in the stated complexity class.)
\end{lemma}

\begin{proof}[Proof (cf. {\cite[Cor.~3.2]{CGLW}})]
If $h=1$, simply output $k=0$. Otherwise, universally check all $k=1,\dotsc,n-1$. We can check in $\DTISPpll$ whether $g^k = h$ by \cite[Lem.~3.1]{CGLW}. If $g^k \neq h$, then that co-nondeterministic branch outputs $\dagger$. If $g^k=h$, then that co-nondeterministic branch outputs $k$. 

If $h \in \langle g \rangle$, then since $|\langle g \rangle| \leq |G|=n$, there must be some $k \in \{1,\dotsc,n-1\}$ such that $g^k = h$, and hence at least one co-nondeterministic branch of the computation outputs an integer.
\end{proof}

\subsection{Group Extensions and Cohomology} \label{sec:Cohomology}

We recall background on group extensions and non-Abelian cohomology. Much of the theory was originally developed in \cite{dedecker,inassaridze}, but \cite{GQCoho} independently redeveloped it (unaware of that prior work at the time), and put it in a form particularly convenient for algorithmic applications, so we follow the notation and development there.  Given a finite group $G$, we will consider a normal subgroup $N \trianglelefteq G$ when considering $G$ as an extension of $N$ by $Q := G/N$. Here, we denote this as $N \xhookrightarrow{\iota} G \overset{\pi}{\twoheadrightarrow} Q$, where $\iota : N \to G$ is an injection and $\text{Im}(\iota) = \text{ker}(\pi)$. We refer to $G$ as the \emph{total group} of the extension.

\paragraph{Actions.} For $g \in G$, let $c_{g} : N \to N$ be given by $c_{g}(n) = gng^{-1}$. The conjugation action $\theta' : G \to \Aut(N)$ is defined by $\theta'(g) := c_{g}$. When $N$ is non-Abelian (as will frequently be the case in our setting), $\theta'$ does not contain $N$ in its kernel. However, the action of $N$ on itself by conjugation is by inner automorphisms (by definition); we denote the group of inner automorphisms by $\text{Inn}(N)$. Thus, we obtain a well-defined homomorphism $G/N \to \Aut(N)/\text{Inn}(N) =: \text{Out(N)}$.
We refer to such a map as an \emph{outer action}.

In computations, rather than represent an outer action as a coset of $\text{Inn}(N)$ in $\Aut(N)$, we specify it by a representative automorphism. Consequently, we must remember that we may need to multiply by an arbitrary element of $\text{Inn}(N)$. We will use $T$ to denote an action rather than $\theta$, to remind the reader that the essential object is the outer action $\theta$ represented by $T$, despite the fact that we work with $T : Q \to \Aut(N)$. Two actions $T_1, T_2 : Q \to \Aut(N)$ are \emph{outer equivalent} if there exists a function $t' : Q \to \text{Inn}(N)$ and an automorphism $\alpha \in \Aut(N)$ such that $T_{1}(q) = \alpha^{-1} \circ t'(q) \circ T_{2}(q) \circ \alpha$, for all $q \in Q$.

\begin{definition}[\algprobm{Outer Action Compatibility} {\cite[Definition~3.2]{GQCoho}}]
Given two actions $T_{1}, T_{2} : Q \to \Aut(N)$, decide whether there exists $\beta \in \Aut(Q)$ such that $T_{1}$ and $T_{2} \circ \beta$ are outer equivalent; that is, whether there exists $(\beta, \alpha, t') \in \Aut(Q) \times (\Aut(N) \rtimes \text{Inn}(N)^{Q})$ such that $T_{1}(q) = \alpha^{-1} \circ t'(\beta(q)) \circ T_{2}(\beta(q)) \circ \alpha,$ for all $q \in Q$.
\end{definition}

\paragraph{Cohomology.} Let $\pi : G \to G/N \cong Q$ be the natural projection map. Any function $s : Q \to G$ such that $\pi(s(q)) = q$ for all $q \in Q$ is called a \emph{section} of $\pi$. Any such section gives rise to a function $f_{s} : Q \times Q \to N$ defined by $f_{s}(p,q) = s(p)s(q) \cdot s(pq)^{-1}$. We are free to choose $s(1) = \text{id}_{G}$, and then $f(1_Q,q) = f(q,1_Q) = 1_N$ for all $q \in Q$. Such sections are called \emph{normalized}. We will assume that all sections are normalized, unless otherwise stated.

A section gives rise to an action (not just an outer action) $T_{s} : Q \to \Aut(N)$, given by $T_{s}(q)(n) = s(q) \cdot n \cdot s(q)^{-1}$. The $2$-cocycle identity is:
\[
f_{s}(q_1, q_2) \cdot f_{s}(q_1 q_2, q_3) = T_{s}(q_1)(f_{s}(q_2, q_3)) \cdot f_{s}(q_1, q_2 q_3).
\]
Any function $f : Q \times Q \to N$ is called a \emph{$2$-cochain}. If $f$ satisfies the $2$-cocycle identity with respect to $T$, then $f$ is called a \emph{$2$-cocycle} with respect to $T$.

Observe that the $2$-cocycle identity depends not only on the action $T_{s}$ and $f_{s}$, but also on the relationship between $T_{s}$ and $f_{s}$ (namely, that they come from the same section). It is preferable to have a condition that does not depend on the ambient group extension. To this end, note that the action satisfies $T_{s}(q_1)T_{s}(q_2) = c_{f_{s}(q_1, q_2)}(T_{s}(q_{1}q_{2}))$, where $c_{m} : N \to N$ denotes the conjugation action sending $c_{m}(n) = mnm^{-1}$. We now recall the notion of extension data from Grochow and Qiao.

\begin{definition}[Extension Data {\cite[Definition~3.4]{GQCoho}}]
Let $Q$ and $N$ be groups. We say that a pair $(T, f)$ of an action $T : Q \to \Aut(N)$ and a function $f : Q \times Q \to N$ is \emph{extension data} if for all $q_i \in Q$, 
\begin{align*}
&T(q_1)T(q_2) = c_{f(q_1, q_2)}(T(q_{1}q_{2})), \text{ and } \\
&f(q_1, q_2) \cdot f(q_{1}q_{2}, q_{3}) = T(q_1)(f(q_2, q_3)) \cdot f(q_1, q_{2}q_{3}).
\end{align*}

\noindent In this case, we refer to $f$ as a $2$-cocycle with respect to the action $T$.
\end{definition}

We now recall what it means for two extension data to be equivalent.

\begin{definition}[Equivalence of Extension Data {\cite[Definition~3.5]{GQCoho}}]
Two extension data $(T_i, f_i)$ ($i = 1, 2$) are \emph{equivalent} if there exists a map $t : Q \to N$ such that $T_{1}(q) = c_{t(q)} T_{2}(q)$ for all $q \in Q$, and 
\[
f_{1}(q,r) = t(q)[T_{2}(q)(t(r))] \cdot c_{f_{2}(q,r)}(t(qr)^{-1}) \cdot f_{2}(q,r) := f_{2}^{t, T_{2}}(q,r).
\]
\end{definition}

\begin{definition}[Pseudo-Congruence of Extension Data {\cite[Definition~3.6]{GQCoho}}] \label{def:pseudocongruence}
Let $Q$ and $N$ be groups. For $i = 1, 2$, let $T_i : Q \to \Aut(N)$, $f_i : Q \times Q \to N$, and $(T_i, f_i)$ the corresponding extension data. We say that $(T_1, f_1)$ and $(T_2, f_2)$ are \emph{pseudo-congruent} if there exist $(\alpha, \beta) \in \Aut(N) \times \Aut(Q)$ and $t : Q \to N$ such that for all $q \in Q$ and all $n \in N$, the following hold:
\begin{align*}
&T_{1}(q)(n) = (\alpha^{-1} \circ c_{t(\beta(q))} \circ T_{2}(\beta(q)) \circ \alpha)(n), \text{ and for all }  q_1, q_2 \in Q \\
&f_{1}(q_1, q_2) = \alpha^{-1} \biggr[f_{2}^{t, T_2}(\beta(q_1), \beta(q_2)) \biggr] =: f_{2}^{(\alpha, \beta, t, T_2)}(q_1, q_2).
\end{align*}
\end{definition}

We now turn to discussing when two different extension data $(T_1, f_1), (T_2, f_2)$ yield isomorphic total groups. The following is essentially equivalent to \cite[Main Lemma~3.7]{GQCoho}, but using the assumption about isomorphisms extending from a given normal subgroup (which, in our setting, we will guess non-deterministically), rather than assuming the normal subgroups come from characteristic subgroup functions as in \cite{GQCoho}; the latter assumption is used in \cite{GQCoho} only to ensure that any isomorphism $G_1 \to G_2$ must send $N_1$ to $N_2$.

\begin{lemma}[{cf. \cite[Main~Lemma~3.7]{GQCoho}}] \label{lem:GQMain}
Given two finite groups $G_1$ and $G_2$, with normal subgroups $N_1 \unlhd G_1$ and $N_2 \unlhd G_2$, we have that there is an isomorphism $G_1 \to G_2$ sending $N_1$ to $N_2$ if and only if both of the following conditions hold:
\begin{enumerate}[label=(\alph*)]
\item $N_1 \cong N_2$ and $G_1 / N_1 \cong G_2 / N_2$.

\item Let $(T_i, f_i)$ be extension data for $N_i$ and $G_i/N_i$. The extension data $(T_1, f_1)$ is pseudo-congruent to $(T_2, f_2)$.
\end{enumerate}
\end{lemma}

\begin{remark} \label{rmk:t-trivial}
Following the introduction of $t \colon Q \to N$ starting on the bottom of p.~1168 of \cite{GQCoho}, we see that the role of $t$ is essentially to correct for having chosen sections in two groups that are not equal under the chosen candidate isomorphism. In particular, Lemma~\ref{lem:GQMain} holds even if we restrict the pseudo-congruence in part (2) to have $t$ being the trivial map ($t(q)=1$ for all $q \in Q$). That is, $G_1 \cong G_2$ by an isomorphism sending $N_1$ to $N_2$ iff (a) [as in Lemma~\ref{lem:GQMain}] and (b') there exist extension data $(T_i,f_i)$ for $G_i$ ($i=1,2$) that are pseudo-congruent by a pseudo-congruence with $t=1$.
\end{remark}

\subsection{Group Presentations} \label{sec:Presentations}

\noindent We recall some preliminaries regarding group presentations. For a standard reference see, e.\,g., \cite{CoxeterMoser}; our presentation follows closely that in \cite{DietrichWilson}. The \emph{free group} $F[X]$ on a given alphabet $X$ is obtained by creating a disjoint copy $X^{-}$ of the alphabet, and treating the elements of $F[X]$ as words over the disjoint union $X \dot \cup X^{-}$, including the empty word $1 \not \in X \cup X^{-}$. Formally, one can replace the given set $X$ by $\{ (x, 1) : x \in X \}$ and $X^{-}$ by $\{ (x, -1) : x \in X \}$. For simplicity, we abuse notation by identifying $x$ with $(x, 1)$, and $x^{-}$ with $(x, -1)$. The empty word serves as the identity, and concatenation is the group product. To impose the existence of inverses, we apply the rewriting rules $xx^{-} \mapsto 1$ and $x^{-}x \mapsto 1$, for each $x \in X$. For a set $M$, let $M^{X}$ denote the set of functions from $X$ to $M$. The elements of $M^{X}$ are naturally represented as tuples $\mathbf{m} = (\mathbf{m}_{x})_{x \in X}$. For a group $G$, a tuple $\textbf{g} \in G^X$, and a word $w \in F[X]$, we assign $w(\textbf{g}) \in G$ by replacing each variable $x^{\pm}$ in $w$ with the corresponding value of $\textbf{g}_{x} \in G$ or $\textbf{g}_{x}^{-1} \in G$, and then evaluating the corresponding product in $G$. Define the homomorphism $\widehat{\mathbf{g}} : F[X] \to G$ by sending $w \mapsto w(\mathbf{g})$. If $G$ is generated by $S := \text{Im}(\widehat{\mathbf{g}})$ and $R$ generates $\text{ker}(\widehat{\mathbf{g}})$, then the pair $\langle S | R \rangle$ is a \emph{presentation} of $G$, where $R$ is the set of \emph{relations} relative to $S$. Note that $\langle S | R \rangle$ describes $G$ up to isomorphism, in which case we write $G \cong \langle S | R \rangle$.

In this paper, we will follow the conventions of \cite{GKKLquant} in defining the length of a presentation.

\begin{definition}[Length of a Presentation] \label{def:PresentationLength}
The \emph{(word) length} of a presentation is the sum of the number of generators and the lengths of the relations as words in both the generators and their inverses. 
\end{definition}

\begin{remark}[On different notions of length]
\label{rem:presentation-length-conventions}
There are differing notions in the literature regarding the length of a presentation. Babai, Goodman, Kantor, Luks, and Pálfy  \cite{BGKLP} take the length of the presentation to be the total number of characters required to write down all of the generators and relations. While Babai, Goodman, Kantor, Luks, and Pálfy count each generator as a single symbol, they point out that in properly accounting for the total bit-length, one could instead count each generator as an indexed symbol where the index must be written as a word over a fixed finite alphabet (say, binary). If there are $k$ generators, then doing so requires $O(\log k)$ bits per generator. As $k \leq \lceil \log |G| \rceil$, this would only increase the length of the presentation by a factor of $O(\log k) = O(\log \log |G|)$.

In \cite{BGKLP} they also allow exponents in their relators to be written in binary, but point out that one can also avoid this without substantially changing the total length, as follows. Namely, one can obtain a new presentation of the same asymptotic length by introducing new variables and relations to simulate repeated squaring. For instance, the cyclic group of order $m$ has presentation $\langle x | x^m = 1 \rangle$, which has word length $m$, but bit-length $O(\log m)$. However, we can instead use the following presentation:
\[
\langle x_0, \ldots, x_k | x_{i+1} = x_{i}^{2} (0 \leq i < k), x_{0}^{e_{0}} x_{1}^{e_{1}} \cdots x_{k}^{e_{k}} = 1 \rangle,
\]
where $e_{i} \in \{0,1\}$ (for each $0 \leq i \leq k$) and $m = \sum_{i=0}^{k} 2^{i}e_{i}$, which has word length at most $4 \log(m) + O(1)$ \cite[Remark~1.3]{BGKLP}. 

Putting these all together, if a group $G$ has presentation of word length (Definition~\ref{def:PresentationLength}) $O(\log^c |G|)$ then it has a presentation of bit length $O(\log^{c} |G| \log\log|G|) \subseteq \widetilde{O}(\log^c |G|)$, which is smaller than $O(\log^{c+\varepsilon} |G|)$ for all fixed $\varepsilon > 0$. If the number of generators was initially $O(1)$ and the presentation's word length was $\Omega(\log|G|)$, then the factor of $\log\log|G|$ is unnecessary.
\end{remark}

Babai, Goodman, Kantor, Luks, and Pálfy \cite[Theorem~8.3]{BGKLP} showed that if every finite simple group admits a presentation of length $O(\log^c n)$, then every finite group admits a presentation of length $O(\log^{c+1} n)$. Their proof was constructive. We recall their construction here as we will use it in our main result. 

Following their terminology, we call a set $S$ of generators of a group $G$ \emph{efficient} if every element of $G$ can be written as a straight-line program over $S$ of length at most $2\log_2|G|$. Babai and Szemerédi \cite{BabaiSzemeredi} showed that from any generating set for a group $G$, there is a SLP of length at most $\log^2 |G|$ that computes an efficient generating set of $G$ (see \cite[Lemma~8.2]{BGKLP}). (We follow much, but not all, of the notation there for ease of reference; one notable exception is that our indices increase---if $i < j$ then $G_i \leq G_j$---whereas theirs decrease.)

\begin{construction}[{\cite[Proof~of~Theorem~8.3]{BGKLP}}] \label{construction:BGKLP}
Suppose there is $c \geq 2$ such that all finite simple groups of order $n$ have presentations of length $O(\log^c n)$. Let $G$ be a finite group with composition series 
$1 = G_0 \unlhd G_1 \unlhd \dotsb \unlhd G_m = G$.  
We will describe their construction of a presentation for $G$ of length $O(\log^{c+1} |G|)$. 

For each $i$, let $\langle \Gamma_i | R_i\rangle$ be the presentation for $G_i / G_{i-1}$ of length $O(\log^c |G_i / G_{i-1}|)$. Given a surjective homomorphism $\pi\colon X \to Y$, we say that $x \in X$ is a \emph{lift} of $y \in Y$ if $\pi(x)=y$, and we say that $\widehat{S} \subseteq X$ is a lift of $S \subseteq Y$ if $|\widehat{S}| = |S|$ and $\pi(\widehat{S}) = S$.

The generators in the presentation of $G$ will be as follows:
\begin{itemize}
\item[(Gen1)] For $i=1,\dotsc,m$, let $\Gamma_i^0$ be a minimum-size subset of $\Gamma_i$ that generates $G_i / G_{i-1}$ (in particular, $|\Gamma_i^0| \leq \log |G_i / G_{i-1}|$). Let $\widehat{\Gamma}_i$ be lifts of $\Gamma_i$ in $G_i$ (hence in $G$), along the natural quotient map $G_i \to G_i / G_{i-1}$. Let $\widehat{\Gamma}_i^0$ be the subset of $\widehat{\Gamma}_i$ that is a lift of $\Gamma_i^0$. Let $S_i$ be a short SLP computing an efficient set of generators of $G_i$ from $\Gamma_1^0 \cup \Gamma_2^0 \cup \dotsb \cup \Gamma_i^0$, and let 
\[
S := \bigcup_{i=1}^m S_i.
\]

\item[(Gen2)] For each $i$ and for each $\gamma \in \widehat{\Gamma}_i$, let $D(\gamma)$ be a an SLP of length $O(\log |G_i|)$ (since $S_i$ is an efficient generating set) computing $\gamma$ from $S_i$. Let 
\[
D := \bigcup_{i=1}^m \bigcup_{\gamma \in \widehat{\Gamma}_i} D(\gamma).
\]

\item[(Gen3)] For each $i$ and each $\rho \in R_i$ giving the relation $\rho(\gamma_{i1},\dotsc,\gamma_{ir})=1$ in the generators $\Gamma_i = \{\gamma_{i1},\dotsc,\gamma_{ir}\}$, let $\widehat{\gamma}_{im} \in \widehat{\Gamma}_i$ be the lift of $\gamma_{im}$ chosen above. When we apply $\rho$ to the lifted elements, we find $\rho(\widehat{\gamma}_{i1}, \ldots, \widehat{\gamma}_{ir})=z_\rho$ for some element $z_\rho \in G_{i-1}$. Let $P(\rho)$ be a short straight-line program computing $z_{\rho}$ from $S_{i-1}$. Let:
\[
P := \bigcup_{i=1}^m \bigcup_{\rho \in R_{i}} P(\rho).
\]

\item[(Gen4)] As $G_{i-1}$ is normal in $G_i$, we have for each $n \in G_{i-1}$ and each $g \in G_i$, that $gng^{-1} \in G_{i-1}$. For each $n \in \bigcup_{j<i}\widehat{\Gamma}_j^0 =: \widehat{\Gamma}^0_{< i}$ and each $g \in \widehat{\Gamma}_i^0$, let $C(n,g)$ be a SLP of length $O(\log |G_i|)$ computing $gng^{-1}$ from $S_{i-1}$. Let:
\[
C := \bigcup_{i=1}^m\bigcup_{n \in \widehat{\Gamma}_{< i}^0} \bigcup_{g \in \widehat{\Gamma}_i^0} C(n,g).
\]
\end{itemize}

\noindent Let $T := S \cup P \cup D \cup C$ be the set of generators of $G$ ($S$ alone generates $G$, but the additional generators will enable a shorter presentation). 

In order to construct a presentation of $G$, they associate for each $t \in T$, an abstract generator symbol $x(t)$. Relative to this generating set, the following relations define $G$: 
\begin{itemize}

\item[(Rel1)] If some $t \in T$ arose as $t = uw$ or $t = w^{-1}$ in the course of one of the SLPs $S_i$, $D(\gamma)$, $P(\rho)$, or $C(a,g)$, then include the relation $x(t) = x(u)x(w)$ or $x(t)x(w) = 1$, respectively.

\item[(Rel2)] For each relator $\rho \in R_{i}$ lifted to $\widehat{\rho}(\widehat{s}_{i1}, \ldots, \widehat{s}_{ir}) = z_{\rho}$, include the relation $\widehat{\rho}(x(\widehat{s}_{i1}), \ldots, x(\widehat{s}_{im})) = x(z_{\rho})$.

\item[(Rel3)] For each $n \in \widehat{\Gamma}_{i-1}^0$ and $g \in \widehat{\Gamma}_i^0$, include the relation $x(gng^{-1}) = x(g) x(n)x(g)^{-1}$. 
\end{itemize}
This completes the construction, whose correctness is proved in \cite[Theorem~8.3]{BGKLP}.
\end{construction}

\subsection{Additional preliminaries for the Ree groups $^2G_{2}(q)$} \label{sec:ReeGroupPrelims}
We recall the specific facts about the Ree groups that will be
used below. For the reader wishing to better understand these facts and put them in the larger context of finite simple groups, and Lie and algebraic groups, we refer to  \cite{Curtis1965, TitsBook, carter-book}, but we will not need such depth here.

The Ree groups $^2 G_2(q)$ \cite{Ree}\footnote{Sometimes called ``small Ree groups'' to distinguish them from the other groups discovered by Ree, $^2 F_4(q)$, which have higher rank as algebraic groups. We avoid calling them ``small'' here because that word might be mistaken for other uses of ``small'' in our complexity-theoretic context, such as ``polylogarithmic.'' Throughout this paper, when we say ``Ree group'', we mean the ``small Ree groups $^2 G_2(q)$.''} are finite simple groups of order $q^3(q^3+1)(q-1)$, where $q=3^{2k+1}$ is an odd power of $3$. 
Rather than give an abstract definition, we use the $7 \times 7$ matrix representation of $^2 G_2(q)$ from \cite{KLM}, in the convenient form presented in  \cite[Sec.~3]{Brnhielm2014}.

Let $\ell = 3^k$; this parameter will show up in exponents frequently in this construction.\footnote{For context / to whet the reader's appetite to learn more: the so-called Frobenius automorphism of $\FF_q$ is given by $x \mapsto x^3$; as an automorphism this has order $2k+1 = \log_3 q$. The function $x \mapsto x^{\ell}$ is the $k$-th power of the Frobenius automorphism, and has the property that its square is equal to the inverse of the Frobenius: $((x^\ell)^\ell)^3 = x^{3^{2k+1}} = x$. This is relevant to how $^2 G_2(q)$ is constructed from the exceptional algebraic group $G_2(q)$.} Then $^2G_2(q)$ is generated by
\begin{align*}
&h(\lambda) := \text{diag}(\lambda^{\ell}, \lambda^{1-\ell}, \lambda^{2\ell-1}, 1, \lambda^{1-2\ell}, \lambda^{\ell-1}, \lambda^{-\ell}) \qquad (\lambda \in \FF_q^\times)\\    
&\Upsilon := \text{antidiag}(-1, -1, -1, -1, -1, -1, -1).
\end{align*}
along with a Sylow $3$-subgroup of $^2 G_2(q)$, which, when viewed as a subgroup of $\GL(7,q)$, is the intersection of $^2 G_2(q)$ with the upper unitriangular matrices, and is generated by, for all $x \in \FF_q$:
\begin{align*}
&\alpha(x) = \begin{pmatrix}
1 & x^{\ell} & 0 & 0 & -x^{3\ell+1} & -x^{3\ell+2} & x^{4\ell+2} \\
0 & 1 & x & x^{\ell+1} & -x^{2\ell+1} & 0 & -x^{3\ell+2} \\
0 & 0 & 1 & x^{\ell} & -x^{2\ell} & 0 & x^{3\ell+1} \\
0 & 0 & 0 & 1 & x^{\ell} & 0 & 0 \\ 
0 & 0 & 0 & 0 & 1 & -x & x^{\ell+1} \\ 
0 & 0 & 0 & 0 & 0 & 1 & -x^{\ell} \\
0 & 0 & 0 & 0 & 0 & 0 & 1\\
\end{pmatrix}
\end{align*}

\begin{align*}
&\beta(x) = \begin{pmatrix}
1 & 0 & -x^{\ell} & 0 & -x & 0 & -x^{\ell+1} \\
0 & 1 & 0 & x^{\ell} & 0 & -x^{2\ell} & 0 \\
0 & 0 & 1 & 0 & 0 & 0 & x \\
0 & 0 & 0 & 1 & 0 & x^{\ell} & 0 \\
0 & 0 & 0 & 0 & 1 & 0 & x^{\ell} \\
0 & 0 & 0 & 0 & 0 & 1 & 0 \\
0 & 0 & 0 & 0 & 0 & 0 & 1\\
\end{pmatrix}
\end{align*}

\begin{align*}
\gamma(x) = \begin{pmatrix}
1 & 0 & 0 & -x^{\ell} & 0 & -x & -x^{2\ell} \\
0 & 1 & 0 & 0 & -x^{\ell} & 0 & x \\
0 & 0 & 1 & 0 & 0 & x^{\ell} & 0 \\
0 & 0 & 0 & 1 & 0 & 0 & -x^{\ell} \\
0 & 0 & 0 & 0 & 1 & 0 & 0 \\ 
0 & 0 & 0 & 0 & 0 & 1 & 0 \\
0 & 0 & 0 & 0 & 0 & 0 & 1\\
\end{pmatrix}    
\end{align*}
This completes our definition/recollection of the Ree groups. 

From the matrices above, one can directly verify the following facts (see also \cite{Ward1966} for more):
\begin{itemize}
\item As functions of $x$ and $y$, $\beta$ and $\gamma$ are ``additive''; more precisely, for all $x,y \in \FF_q$, we have $\beta(x+y)=\beta(x)\beta(y)$ and $\gamma(x+y)=\gamma(x)\gamma(y)$.

\item $\langle \beta(x), \gamma(y) : x,y \in \FF_q \rangle \cong \FF_q^2 \cong \FF_3^{4k+2}$ is an elementary Abelian group, which we denote $[U,U]$.

\item For all $x, y \in \FF_q$, there exists elements $b_{x,y} \in [U,U]$ such that $\alpha(x+y)=\alpha(x)\alpha(y)b_{x,y}$.
\end{itemize}

There are a few additional key structural facts about the group that we will recall and use. Let $U$ be the Sylow $3$-subgroup $U := \langle \alpha(x), \beta(y), \gamma(z) : x,y,z \in \FF_q \rangle$; let $H = \{ h(\lambda) : \lambda \in \FF_q^\times \} = \langle h(\omega) \rangle$, where $\omega$ is a generator of the multiplicative group $\FF_q^\times$ (a.k.a. a primitive element of $\FF_q$). We recall the following facts:
\begin{itemize}
\item $U = \{\alpha(a)\beta(b)\gamma(c) : a,b,c \in \FF_q\}$ (that is, every element of $U$ can be represented as $\alpha(a)\beta(b)\gamma(c)$).

\item (Bruhat decomposition) $^2 G_2(q) = (HU) \sqcup (HU \Upsilon U)$, where we highlight that this is a \emph{disjoint} union. 

\item $H$ normalizes $U$, and every element of $HU$ has a unique expression of the form $h(\lambda) \alpha(a)\beta(b) \gamma(c)$ with $\lambda \in \FF_q^\times$ and $a,b,c \in \FF_q$.

\item Every element of $^2 G_2(q)$ outside of $HU$ has a \emph{unique} expression of the form $h(\lambda) u \Upsilon u'$ with $u,u' \in U$, and every element of the form $h(\lambda) u \Upsilon u'$ with $u,u' \in U$ is not in $HU$.
\end{itemize}

The above gives us a nice bijection or parametrization of the Ree groups: 
\begin{equation} \label{eq:Ree-bijection}
\left(\FF_q^\times \times \FF_q^3 \right) \sqcup  \left(\FF_q^\times \times \FF_q^6 \right) \to {}^2 G_2(q),
\end{equation}
given by
\[
\FF_q^\times \times \FF_q^3 \ni (\lambda; a,b,c) \mapsto h(\lambda) \alpha(a)\beta(b) \gamma(c) \in HU
\]
and
\[
\FF_q^\times \times \FF_q^6 \ni (\lambda; a,b,c,x,y,z) \mapsto h(\lambda) \alpha(a)\beta(b)\gamma(c) \Upsilon \alpha(x) \beta(y) \gamma(z) \in HU\Upsilon U.
\]
The facts recalled above immediately imply that this is a bijection.

In our proof, we will need an efficiently computable form of this bijection, relative to a particular generating set, which we give in Lemma~\ref{lem:Ree-bijection}.

\section{Depth-3 Circuits for Group Isomorphism}

In this section, we will establish the following theorem. We first introduce a small notation: for two classes of Boolean functions $\mathcal{C},\mathcal{D}$, we use $\mathcal{C} \wedge \mathcal{D}$ to denote the class of Boolean functions expressible as a conjunction of a function in $\mathcal{C}$ and a function in $\mathcal{D}$. 

\begin{theorem}[Full version of Theorem~\ref{thm:main}] \label{thm:MainUpperBound} 
In the Cayley-table model, \algprobm{Group Isomorphism} belongs to
\[
\exists^{\log^{c_{1}} (n)} \forall^{\log^{c_{2}} (n)} \DTIMEpl \land \forall^{\log n}\DTISPpll
\]
for some constants $c_{1}, c_{2}$. Consequently, \algprobm{Group Isomorphism} can be decided by uniform depth-$3$ circuits of quasipolynomial size and polylogarithmic bottom fan-in.
\end{theorem}

We will accomplish this in two steps. 
We will first show that we can verify whether a multiplication table satisfies the group axioms in $\forall^{\log n}\DTISPpll$. We will show, more strongly, that this check can be implemented using a CNF of polynomial-size and logarithmic width (Lemma~\ref{lem:valid-cayley-tables}). 
Our second step will be to show that, given two multiplication tables of groups $G$ and $H$, we can decide whether $G$ and $H$ are isomorphic in $\exists^{\polylog(n)} \forall^{\polylog(n)} \DTIMEpl$ (Proposition~\ref{prop:GpI}).

\subsection{The total language and validity of Cayley tables}
\label{subsec:total-language-validity}
It is more natural to talk about our circuits under the assumption or promise that the input Cayley tables are in fact the multiplication tables of \emph{groups}. Here we show that this assumption can be checked by a polynomial-size CNF, so that later we can talk only in the language of groups, while still solving the total \textsc{GpI} problem rather than the promise problem.

\begin{lemma}[Validity of Cayley tables]
\label{lem:valid-cayley-tables}
The language \(\mathrm{VALID}_n\) of \(n\times n\) tables whose entries encode
 a group operation on \([n]\) is decidable in $\forall^{\log n}\DTISP(\log n \log \log n, \log n)$, as well as by a uniform CNF of polynomial size and logarithmic clause width.
\end{lemma}
\begin{proof} 
Assume \(n\ge 2\). Let \(\ell=\lceil\log_2 n\rceil\). The verifier checks that the table is an associative Latin square, i.\,e., it checks the following four universally quantified formulas testing: range, row-distinctness, column-distinctness, and associativity. 

All universally quantified labels are encoded using \(\ell\) bits. Since some \(\ell\)-bit strings may not represent elements of \([n]\), any universal instance involving such an invalid label is treated as vacuous. Actual table entries, however, are explicitly checked to lie in \([n]\). The conditions are as follows. 

\begin{itemize}
\item First, every table entry is in range: for all $x,y\in[m]$, $T(x,y)\in [n]$. 

\item Second, every row has distinct entries: for all $x,y,z\in[n]$, if $y\neq z$ then $T(x,y)\neq T(x,z).$

\item Third, every column has distinct entries: for all $x,y,z\in[n]$, if $y\neq z$, then $T(y,x)\neq T(z,x).$

\item Finally, for all $x,y,z\in[n]$, associativity holds. Let $a:=T(x,y),\,\, b:=T(y,z).$ 
The verifier checks $T(a,z)=T(x,b),$ that is, $T(T(x,y),z)=T(x,T(y,z)).$
\end{itemize}

\noindent In order to check these conditions, we will use $3\ell \in O(\log n)$ universally quantified (co-nondeterministic) bits. Given the co-nondeterministic bits, we may check each of the above conditions in \\ $\DTISP(\log n \log\log n, \log n)$, where the main runtime comes from computing indices of the form $ni+j$ where $n,i,j$ are all $\log n$-bit numbers, which can be multiplied in time $O(\log n \log \log n)$ \cite{IntegerMult} and added in time linear in their bit-length. Thus, our total work is $\forall^{\log n}\DTISP(\log n \log \log n, \log n)$. By Fact~\ref{fact:DTISP}, this can thus checked by a CNF of size $n^{O(\log \log n)}$. A more careful analysis yields a polynomial-size CNF. We will now proceed with this analysis. 

In each case, we are only comparing table entries to the bound of $n$, or comparing two table entries to each other. A comparison between an \(\ell\)-bit table entry and a fixed \(\ell\)-bit number $n$ is expressed bit-by-bit. Thus, for fixed \(c\), the condition $T(x,y)\neq c$ is a single clause of width $\ell$, while $T(x,y)=c$ is a conjunction of \(\ell\) literals. Similarly, equality of two table entries, such as \(T(a,z)=T(x,b)\), is a conjunction of \(O(\ell)\) constant-size CNFs, namely the bitwise \textsf{XNOR}s (equality). 

The range condition is imposed by including, for each \(x,y\in[n]\) and each \(\ell\)-bit string \(c\notin[n]\), the clause $T(x,y)\neq c.$ For the row- and column-distinctness checks, inequality of two table entries is imposed by forbidding each possible common value. Thus row-distinctness is expressed by the clauses $T(x,y)\neq c \;\vee\; T(x,z)\neq c$ for all \(x,y,z,c\in[n]\) with \(y\neq z\), and column-distinctness is expressed similarly by the clauses $T(y,x)\neq c \;\vee\; T(z,x)\neq c.$ 
For associativity, we expand over the possible values \(a,b\in[n]\) of \(T(x,y)\) and \(T(y,z)\). For each \(x,y,z,a,b\in[n]\) and each bit position \(r\in[\ell]\), we include the two clauses $T(x,y)\neq a \;\vee\; T(y,z)\neq b\;\vee\; \neg T(a,z)_r \;\vee\; T(x,b)_r,$ and $T(x,y)\neq a \;\vee\; T(y,z)\neq b\;\vee\; T(a,z)_r \;\vee\; \neg T(x,b)_r.$ These clauses say that if \(T(x,y)=a\) and \(T(y,z)=b\), then the \(r\)-th bits of \(T(a,z)\) and \(T(x,b)\) agree. 

Altogether there are polynomially many clauses: \(O(n^3)\) range clauses, \(O(n^4)\) row-distinctness clauses, \(O(n^4)\) column-distinctness clauses, and \(O(n^5\ell)\) associativity clauses. Each clause has width \(O(\ell)\). Hence the checks yield a uniform polynomial-size CNF. 

These conditions characterize group tables. If \(T\) is a group table, then all entries lie in \([n]\), every row and every column is a permutation of \([n]\), and associativity holds. Hence the verifier accepts. Conversely, suppose that the verifier accepts. The range condition ensures that \(T\) defines a binary operation on \([n]\). The row-distinctness and column-distinctness conditions imply that every row and every column contains \(n\) distinct elements of \([n]\). Thus every row and every column contains each element of \([n]\) exactly once, or equivalently, \(T\) is a Latin square (the multiplication table of a quasigroup). Since the associativity check also holds and \([n]\) is nonempty, \(T\) is a nonempty associative quasigroup. It is a standard fact that every nonempty associative quasigroup is a group, thus $T$ is the Cayley table of a group. Therefore the verifier accepts exactly the valid Cayley tables of groups.
\end{proof}

\subsection{The promise language and testing for isomorphism in the second level of the Polylogarithmic-Time Hierarchy}

In this section, we will establish the following.

\begin{proposition} \label{prop:GpI}
Given two groups $G$ and $H$ by their multiplication tables, we can decide whether $G \cong H$ in $\exists^{\polylog(n)} \forall^{\polylog(n)} \DTIMEpl$. 
\end{proposition}

Note that Proposition~\ref{prop:GpI} operates under the promise that the multiplication tables provided as input indeed satisfy the group axioms. Because the details of the proof take 7+ pages, but the idea of the proof is much shorter, we start by presenting the idea:

\begin{proof}[Proof idea]
We first present the proof idea for those composition factors that satisfy the Uniform Short Presentation Conjecture; then we present the additional ideas that we use to handle the Ree groups, for whom the Uniform Short Presentation Conjecture remains open (see ``The status of the (Uniform) Short Presentation Conjecture'' on p.~\pageref{par:status} for more details).

    \paragraph{Proof idea assuming composition factors satisfy the Uniform Short Presentation Conjecture.} The existential quantifier is used to non-deterministically guess a presentation for $G$ that is of the form given in Construction~\ref{construction:BGKLP}, along with associated elements of $G$ that play the role of the generators in that presentation, and associated SLPs as in Construction~\ref{construction:BGKLP}. 
    Note that these generators come in stages, corresponding to the composition factors of $G$.
    The verifier then checks that
    the guessed generators in both groups satisfy this presentation. A priori, this would imply that the subgroup of $G$ generated by the guesses and the subgroup of $H$ generated by the guesses are both quotients of the presented group. (At this point, it would be possible that $G$ contains a quotient of $H$ and $H$ contains a quotient of $G$.) To ensure it is genuinely an isomorphism, the algorithm also checks that the guessed generators at stage $i$ do not belong to the subgroup at stage $i-1$, and also that the orders of the simple composition factors it guesses---which can be computed using well-known formulas directly from the standard names---multiply out to $|G|=|H|$. This lets us therefore conclude that the presented group and the subgroups generated by the guesses all have  order exactly $n=|G|=|H|$, so if both $G$ and $H$ satisfy the presentation, then they are both isomorphic to the presented group, and hence to each other. The correctness essentially follows from \cite[Theorem~8.3]{BGKLP}.

\paragraph{Additional ideas to handle Ree group composition factors.} In general, we allow our non-deterministic guesses to guess that composition factors correspond to Ree groups as well, using their standard name. However, we do not have available to us uniform short presentations of the Ree groups. Similar to the above case, we also use the non-deterministic bits to guess a generating set of size $O(\log n)$ (corresponding to a generating set we specify in Lemma~\ref{lem:Ree-bijection}). However, rather than verifying a cohomology class based on a small number of short relators, we instead use a co-nondeterministic $\forall^{\polylog}$ quantifier to essentially treat the entire multiplication table of the Ree group as relators, that is, we look at relators of the form $q \cdot r \cdot (qr)^{-1}$. We use the universal quantifier both to enumerate over all pairs $(q,r)$---for which we take advantage of an easily computable and co-nondeterministically invertible bijection (see Lemma~\ref{lem:Ree-bijection})---as well as, for each $(q,r)$ to pick an SLP in terms of the previous terms in the composition series. 

There are two additional tricks we highlight here. The first is that the natural approach to the latter SLP would be to nondeterministically guess the SLP corresponding to $(q,r)$, but this would put us back in the third level of the polylogarithmic hierarchy. Instead, we try \emph{all} SLPs of length $O(\log^2 n)$, with the bound being guaranteed by the Babai--Szemerédi Reachability Lemma. If the SLP is not equal to the element it is supposed to, then that co-nondeterministic branch simply accepts, essentially leaving the final accept/reject decision to the branches where the SLP \emph{is} equal to what it is supposed to. And since we are guaranteed that such a branch exists, the check is still guaranteed to happen. Indeed, this is why we introduced the class $\coNTISPpllMV$.

The second is why we needed the efficiently invertible bijection of Lemma~\ref{lem:Ree-bijection} at all. The issue here is the following: when we instead had a small number of relators, we can simply guess the SLPs---and hence the cohomology class---all at once and verify them. But when we are using the universal quantifier to quantify over relators, we have no way of knowing that our choices of section, and hence cohomology class, on the different co-nondeterministic branches were consistent with one another. But this consistency is precisely provided by the bijection of Lemma~\ref{lem:Ree-bijection}. The latter bijection also crucially takes advantage of another universal quantifier in order to solve discrete log efficiently in a group given by its multiplication table (Lemma~\ref{lem:discrete-log}).

With all these pieces in place, even for the Ree groups, the final verification checks are analogous to those for non-Ree groups coming from Construction~\ref{construction:BGKLP} \cite{BGKLP}, but now the correctness instead relies on the non-Abelian cohomology setup of Lemma~\ref{lem:GQMain} \cite{GQCoho}.
\end{proof}

Before we come to the details of the proof of Proposition~\ref{prop:GpI}, we need one additional lemma about the Ree groups. For notation see Section~\ref{sec:ReeGroupPrelims}. Suppose we are given $\FF_q$ as $\FF_3[x]/f(x)$. For $a \in \FF_3^{2k+1}$, let us write 
\[
\overline{\alpha}^a := \alpha(x^0)^{a_0} \alpha(x^1)^{a_1} \dotsb \alpha(x^{2k})^{a_{2k}},
\]
and analogously for $\overline{\beta}^b$ and $\overline{\gamma}^c$. The following lemma gives us a ``standard'' generating set for $^2 G_2(q)$ of size $2 + 3\log_3 q \leq O(\log |^2 G_2(q)|)$, as well as normal forms for words for each element of the group, in terms of that generating set.

\begin{lemma} \label{lem:Ree-bijection}
Suppose we are given an irreducible polynomial $f(x) \in \FF_3[x]$ of degree $2k+1$ and a primitive element $\omega \in \FF_q = \FF_3[x] / f(x)$. Let
\[
B = \{h(\omega), \Upsilon\} \cup \{\alpha(x^i), \beta(x^i), \gamma(x^i) : 0 \leq i \leq 2k\}.
\]
\begin{enumerate}
\item The following function is a bijection:
\[
\varphi \colon ([q-1] \times \FF_3^{3(2k+1)}) \sqcup ([q-1] \times \FF_3^{6(2k+1)}) \to {}^2 G_2(q)
\]
where $\varphi(i; a,b,c) = h(\omega)^i \overline{\alpha}^a \overline{\beta}^b \overline{\gamma}^c$ and $\varphi(i; a,b,c,a',b',c') = h(\omega)^i \overline{\alpha}^a \overline{\beta}^b \overline{\gamma}^c \Upsilon \overline{\alpha}^{a'} \overline{\beta}^{b'} \overline{\gamma}^{c'}$

\item The bijection $\varphi$ is computable in $\DTISPpll$, and 
\[
\varphi^{-1} \in \coNTISPpllMV.
\]
\end{enumerate}
\end{lemma}

\begin{proof}
1. The bijection is nearly the same as that of \eqref{eq:Ree-bijection}. Using the facts recalled in the preliminaries that $\beta(y+z)=\beta(y)\beta(z)$ for all $y,z \in \FF_q$ and similarly for $\gamma$, we have that $\overline{\beta}^b = \beta(b_0 + b_1 x + b_2 x^2 + \dotsb + b_{2k} x^{2k})$, and similarly for $\gamma$. However, for $\alpha$, we only have that for all $y,z \in \FF_q$, there exists $b_{y,z} \in [U,U]$ such that $\alpha(y+z)=\alpha(y)\alpha(z)b_{y,z}$. Thus we get that $\overline{\alpha}^a = \alpha(a_0 + a_1 x + a_2 x^2 + \dotsb + a_{2k} x^{2k}) \beta(f_1(a)) \gamma(f_2(a))$ where $f_1, f_2 \colon \FF_3^{2k+1} \to \FF_q$ are the unique functions to make this equation hold, that is, $\beta(f_1(a)) \gamma(f_2(a))$ is equal to the product of the appropriate $b_{y,z}$ terms that arise from rewriting $\overline{\alpha}^a$ as above.

Now, from the fact that \eqref{eq:Ree-bijection} is a bijection, it will suffice to show that 
\[
\{\alpha(a) \beta(b) \gamma(c) : a,b,c \in \FF_q\} = \{\overline{\alpha}^a \overline{\beta}^b \overline{\gamma}^c : a,b,c \in \FF_3^{2k+1}\} \qquad \text{(to show)}.
\]
Towards this end, let $\pi \colon \FF_3^{2k+1} \to \FF_q$ be the $\FF_3$-linear vector space isomorphism given by $\pi(a) = a_0 + a_1 x + \dotsb + a_{2k} x^{2k}$. Then from the preceding paragraph, we have
\[
\overline{\alpha}^a \overline{\beta}^b \overline{\gamma}^c   = \alpha(\pi(a)) \beta(f_1(a)) \gamma(f_2(a))\beta(\pi(b)) \gamma(\pi(c)).
\]
Since $[U,U]$ is Abelian and using the multiplicative property of $\beta,\gamma$, the latter is the same as
\[
\alpha(\pi(a)) \beta(f_1(a) + \pi(b)) \gamma(f_2(a) + \pi(c)).
\]
But now, as $b$ ranges over $\FF_3^{2k+1}$, $\pi(b)$ bijectively ranges over $\FF_q$. And since the additive group of $\FF_q$ is, in particular, a group, $f_1(a) + \pi(b)$ bijectively ranges over $\FF_q$ as well. Thus $\{\beta(f_1(a) + \pi(b)) : b \in \FF_3^{2k+1}\} = \{\beta(\pi(b)) : b \in \FF_3^{2k+1}\} = \{\beta(b) : b \in \FF_q \}$, and similarly for $\gamma$. This proves that $\varphi$ is indeed a bijection.

2. That $\varphi$ is computable in $\DTISPpll$ is relatively straightforward: we may compute the $i$-th power of $h(\omega)$ by repeated squaring of $7 \times 7$ matrices, which takes only $O(\log i) \leq O(\log q) \leq O(\log|G|)$ many iterations, each of which involves $O(1)$ operations over $\FF_q$. In total that uses $\polylog(q)$ time and $O(\log q)$ space. We may then construct and multiply the appropriate $\alpha,\beta,\gamma,\Upsilon$ matrices with similar complexity.

Computing $\varphi^{-1}(M)$ is slightly more involved. We split into two cases based on whether $M$ is upper-triangular or not. 
\begin{itemize}
\item \textbf{Case 1:} Suppose first that $M$ is upper-triangular. Then $M = h(\omega)^i \alpha(a)\beta(b')\gamma(c')$ for some $i \in [q-1], a,b',c' \in \FF_q$, and we must compute $i$ and $a,b,c \in \FF_3^{2k+1}$ such that $M = h(\omega)^i \overline{\alpha}^a \overline{\beta}^b \overline{\gamma}^c$. In this case, $M_{1,1} = (\omega^i)^\ell$, where, following the notation of Section~\ref{sec:ReeGroupPrelims} we have $\ell = 3^k$. We may raise $M_{1,1}$ to the $3\ell$-th power by repeated squaring in $\FF_q$ to get $\omega^i$. However, note that we do not yet have $i$, rather we have $\omega^i$ written as a polynomial in $x$. We use the co-nondeterministic bits to either output $\dagger$ or $i$ as in Lemma~\ref{lem:discrete-log}. On the branches of the discrete log subroutine that output $\dagger$, our machine also outputs $\dagger$; on the branches that output $i$, the machine continues as follows. First it computes $M' = h(\omega)^{-i} M$; then $M' = \alpha(a)\beta(b')\gamma(c')$ (but the algorithm does not yet know $a,b',c,'$). From the formulas for the matrices $\alpha(a),\beta(b'),\gamma(c')$ recalled in Section~\ref{sec:ReeGroupPrelims}, one can see that $a = -M'_{3,6}$, that is, $a$ can simply be read off as an element of $\FF_q$ from one of the entries of $M'$. 

The algorithm then computes $M'' = \left(\overline{\alpha}^{a}\right)^{-1} M'$, which now has the form $\beta(b)\gamma(c)$ for some $b,c \in \FF_q$ to be computed. Again, examining the formulas for the matrices, we find $b = -M''_{1,5}$, compute $M''' = \left(\overline{\beta}^{b}\right)^{-1} M'' = \beta(b)^{-1} M''$, and we find $c = -M'''_{2,7}$. This completes the computation of $\varphi^{-1}$ in the case that $M$ is upper triangular.

\item \textbf{Case 2:} Finally, suppose instead that $M$ is not upper triangular. Then by facts recalled in Section~\ref{sec:ReeGroupPrelims}, if $M$ is in $^2 G_2(q)$ at all, it must be of the form $h(\omega)^i \alpha(a)\beta(b)\gamma(c) \Upsilon \alpha(a') \beta(b') \gamma(c')$ for appropriate $i \in [q-1]$ and $a,b,c,a',b',c' \in \FF_q$. This time we proceed similarly (but we used Mathematica \cite{Mathematica} to check our calculations symbolically, see Mathematica code in Appendix~\ref{app:mathematica}). In this case, we have that $-(\omega^i)^{-\ell} = M_{7,1}$. To compute $\omega^{-i}$, we thus compute $(-M_{7,1})^{3\ell}$, and then to compute $i$ we use Lemma~\ref{lem:discrete-log} as before. 

After left-multiplying by $h(\omega)^{-i}$ we find that $a^\ell$ is in the $(6,1)$ entry of the resulting matrix, so we can compute $a$ by taking the $3\ell$-th power of this entry. After left-multiplying by $(\overline{\alpha}^a)^{-1}$, we find that the $(3,1)$ entry is $-b$. Left multiplying by $\beta(b)^{-1}$, we can then find several more parameters at once: the $(2,1)$ entry is $-c$, the $(7,2)$ entry is $-(a')^\ell$. Finally, left multiplying by $(\gamma(c)\Upsilon \overline{\alpha}^{a'})^{-1}$, we are now in the situation where we have a matrix of the form $\beta(b')\gamma(c')$, and we can recover $b'$ and $c'$ as in the upper triangular case above.
\end{itemize}

Thus, every co-nondeterministic branch that does not output $\dagger$ indeed outputs $\varphi^{-1}(M)$, and at least one such co-nondeterministic branch outputs $\varphi^{-1}(M)$.
\end{proof}

We now flesh out the details of the above proof idea for Proposition~\ref{prop:GpI}:

\begin{proof}[Proof of Proposition~\ref{prop:GpI}]
Our machine will non-deterministically guess $\polylog(|G|)$-many bits. When we encounter a composition factor of $^2G_{2}(q)$, we will also utilize $\polylog(|G|)$-many co-nondeterministic bits to handle isomorphism testing of the composition factor (using ideas from the $\ComplexityClass{\Sigma_{2} TIME}(\polylog(|G|))$ isomorphism test of Babai and Szemerédi \cite{BabaiSzemeredi}) and cohomology class isomorphism for the extension (Lemma~\ref{lem:GQMain}). We describe these as follows.

\begin{itemize}
\item[(ND1)] \label{guessed-data} We will guess lists $T_1,\dotsc,T_m$ of non-identity elements of $G$ where all the elements appearing in all $T_i$ are distinct (including elements in $T_i$ versus those in $T_j$, $i \neq j$), and corresponding lists $T'_1, \dotsc, T'_m \subseteq H$ with $|T_i| = |T_{i}'|$ for all $i$ and $\sum_{i=1}^m |T_i| \in O(\log^{c+1}|G|)$.

For convenience later, we will define $T_{0} = T_{0}' = \emptyset$.

(For intuition, we remark that later in the proof, we will check for each $i \in [r]$, that the elements in $T_{i}$ induce the composition factor $G_{i}/G_{i-1}$ and similarly for $T_{i}'$ and $H_i / H_{i-1}$.)

\item[(ND2)] For each $i \in [r]$, we will guess the standard name $SN_i$
(Definition~\ref{def:standard-name}) of a finite simple group
(the intended meaning of the guess here is that this corresponds to the
isomorphism type of $G_i/G_{i-1}$). In addition, if the standard name
includes field data for a field $\FF_q$, we will non-deterministically
guess the prime factorization of $q-1$. 

\item[(ND3)] \label{ND3} (Essentially the remaining data from Construction~\ref{construction:BGKLP}) For each $i \in [m]$ where $SN_{i}$ does not specify a  Ree group $^2G_{2}(q)$:
\begin{itemize}
\item[(ND3-Gen1)] Guess an SLP $S^{nd}_{i}$ in $\leq \sum_{j=1}^i |T_i|$ many variables (this will play the role of $S_i$ from Construction~\ref{construction:BGKLP})

\item[(ND3-Gen2)] For each $\gamma \in T_i$ guess an SLP $D^{nd}(\gamma)$ in at most $|S^{nd}_i|$ many variables

\item[(ND3-Gen3)] For each $\rho \in [O(\log^c |G|)]$ guess an SLP $P^{nd}(\rho)$ in at most $|S^{nd}_{i-1}|$ many variables

\item[(ND3-Gen4)] For each $n \in T_1 \cup \dotsb \cup T_{i-1}$ and each $g \in T_i$ guess an SLP $C^{nd}(n,g)$ in at most $|S^{nd}_{i-1}|$ many variables.
\end{itemize}

\item[(ND4)] \label{ND4} For each $i \in [m]$ where $SN_{i}$ is a Ree group $^2G_{2}(q)$, we guess the following.
\begin{itemize}
\item[(ND4-Action)] (Same as (ND3-Gen4).)

\item[(coND4-Quotient)] Our machine will co-nondeterministically try all SLPs of length $O(\log^2 n)$ in $|T_{i}|$ many variables, as well as all SLPs of length $O(\log^2 n)$ in $|T_1 \cup \dotsb \cup T_{i-1}|$ many variables.

\item[(coND4-CCI)] The machine will co-nondeterministically try, for each pair $q,r \in {}^2 G_2(q)$ and each SLP $P_{q,r}$ of length $O(\log^2 n)$ over $T_1 \cup \dotsb \cup T_{i-1}$. 

\end{itemize}

\item[(coND1)] The machine will co-nondeterministically try all SLPs of length $O(\log^2 n)$ over the $T_i$. (Intention: these will be used to verify that $\langle T_1 \cup \dotsb \cup T_{i+1} \rangle \neq \langle T_1 \cup \dotsb \cup T_i \rangle$.)
\end{itemize}

\noindent We now show how to use these nondeterministic guesses and co-nondeterministic tries to verify whether $G \cong H$.

\paragraph{Verifying standard names.}
First, we verify the basic aspects of the standard names that were guessed: we are given a standard name $(L,d,q,p,f,\alpha)$, where $L$ is the label of the group, $(d,q)$ are its parameters, $q=p^e$ for some $e \geq 1$,  $f$ is a polynomial in $\FF_p[x]$ of degree $e$, and $\alpha$ is a polynomial in $\FF_p[x]$ of degree $< e$. And we must verify the following: 
\begin{itemize}
    \item[(0)] $L$ is a valid label, but there are only finitely many labels so this is trivial. 
    
    \item[(1)] We next check that $p$ is prime. This can be done in $\poly(\log p) \leq \poly(\log |G|)$ time \cite{AKS}.\footnote{Alternatively, we could instead use the much older, easier, and now-classical result that \algprobm{Primes} is in $\mathsf{NP}$, and along with $p$ we could non-deterministically guess the witness that $p$ is prime, which would have size $\poly(\log p) \leq \poly(\log |G|)$. This would have the effect of relying on easier results and improving the exponent of the verification runtime, but at the expense of increasing the number of non-deterministic bits.}

    \item[(2)] We next verify that $q$ is a power of $p$. The same bound from (1) in terms of $|G|$ applies to checking that $q$ is a power of $p$, because $q \leq |G|$ so $q$ is a number of at most $\log |G|$ many bits. 

    \item[(3)] We next check that $f(x)$ is irreducible over $\FF_{p}[x]$. To verify that $f$ is irreducible we use Rabin's irreducibility test \cite[Lemma 1]{Rabin}: $f$ is irreducible if and only if $f$ divides $x^{p^e} -x$ and $gcd(f,x^{p^{e_i}} - x)=1$ for all $e_i$ that are maximal divisors of $e$ (that is, of the form $e$ divided by a single prime). First the algorithm computes the maximal divisors $e_i$. Since $e \leq \log|G|$, the bit-length of $e$ is $O(\log \log |G|)$, so trial division lets us factor $e$ in time $\widetilde{O}(\sqrt{e}) \leq \widetilde{O}(\log^{1/2} |G|)$. The algorithm can thus find the primes $p_i$ dividing $e$, and in $O(\log^3 e) \leq O((\log \log |G|)^3)$ time can compute the maximal divisors $e_i$. 

Now the algorithm checks that $f$ divides $x^{p^e}-x = x^q - x$. For this, it computes $x^{p^e} \pmod{f}$ by repeated squaring, taking the remainder modulo $f$ at each step to keep the degree at most $e$. If the final result is not $x$, the algorithm rejects. There are $O(\log p^e) = O(\log q) \leq O(\log|G|)$ rounds of repeated squaring, each of which takes time $\poly(e,\log p) = \poly(\log q) \leq \poly(\log |G|)$. 

Finally, for each maximal divisor $e_i$ of $e$ computed above, the algorithm checks that $gcd(f,x^{p^{e_i}}-x) = 1$. Before using the standard Euclidean algorithm, as above the algorithm first computes $x^{p^{e_i}} \pmod{f}$ by repeated squaring, taking the remainder modulo $f$ at each step to keep the degree low. The algorithm then uses the standard Euclidean algorithm to compute $gcd(f,(x^{p^{e_i}} \pmod{f})-x) = gcd(f,x^{p^{e_i}}-x)$. As these are polynomials of degree $e$, this takes at most $e$ many rounds each of which takes at most $\poly(e,\log p)$ time. If the gcd is not 1, the algorithm rejects. 

\item[(4)] It remains to check that $\alpha$ is a primitive root in $\FF_{p}[x]/(f(x))$, that is, that $\alpha$ generates the group of units of $\FF_p[x]/f(x) \cong \FF_q$. For this, now that it has verified that $f$ is irreducible (and hence that $\FF_p[x]/f(x)$ is a field of order $q=p^e$), it suffices to verify that the multiplicative order of $\alpha$ is exactly $q-1$. As part of (ND2) we have guessed the prime factorization of $q-1$. The algorithm first verifies this prime factorization: it checks that each listed number is prime (as above---note these numbers have at most $\log q \leq \log|G|$ many bits) and that their product is exactly $q-1$. The algorithm then uses repeated squaring modulo $f$ (as above) to check that (1) $\alpha^{q-1} \equiv 1 \pmod{f}$ and (2) for each prime $p_i$ dividing $q-1$, $\alpha^{(q-1)/p_i} \not\equiv 1 \pmod{f}$. This completes the verification of the basic data guessed along with the standard names.
\end{itemize}

\paragraph{Verifying isomorphism.} We now proceed to checking isomorphism. First the verifier checks that the orders of the guessed simple composition factors indeed multiply to $n = |G| = |H|$. The order of a finite simple group can be computed in $\DTISPpll$ from its standard name by standard formulas (which can be found in many textbooks or, e.\,g., Wikipedia): $|A_k| = k!/2$, $|A_d(q)| = |PSL_{d+1}(q)| = \frac{q^{d(d+1)/2}}{gcd(d+1,q-1)} \prod_{i=1}^d (q^{i+1}-1)$, etc. The algorithm then multiplies out these numbers and checks that the result is equal to $|G|$ and $|H|$, and if not, rejects. (Although this might seem like a simple check that we do ``merely'' for convenience, it actually turns out to be crucial.) 

Now, we will proceed inductively along the composition series specified by the lists $(T_{i})_{i \in [m]}$ and $(T_{i}')_{i \in [m]}$; we will essentially verify that these lists indeed specify composition series for $G$ and $H$ respectively. Since the Ree groups $^2G_{2}(q)$ are not known to have a uniform short presentation, we will consider two cases: when $SN_{i}$  specifies a Ree group $^2G_{2}(q)$, and when $SN_{i}$ does not specify a Ree group.\footnote{And if it turns out the Uniform Short Presentation Conjecture holds for all finite simple groups, then our analysis of the case with uniform short presentations would also apply to the Ree groups.}

Formally, for $i \in [m]$, let $N_{i} = \langle T_1 \cup \dotsb \cup T_{i-1} \rangle$ and $N_{i}' = \langle T'_1 \cup \dotsb \cup T_{i-1}' \rangle$. Let $G_{i} = \langle N_i, T_{i} \rangle$ and $H_{i} = \langle N'_i, T_{i}' \rangle$. Now let $Q_{i} = G_i /  N_{i}$ and $Q_{i}' = H_{i} / N_{i}'$. Our algorithm will verify that the map\footnote{\label{fn:lists}Here we think of $T_i$ and $T'_i$ as lists of elements, so when we say ``the'' map we mean: the one that maps the first element of $T_i$ to the first element of $T'_i$, the second element of $T_i$ to the second element of $T'_i$, and so on} $T_{i} \mapsto T_{i}'$ extends to an isomorphism between $G_{i}$ and $H_{i}$. For our base case, we define for convenience $T_{0} = T_{0}' = \emptyset$, and $G_{0} = H_{0} = \{1\}$. 

For convenience, define $\Phi_{i}$ to be the map sending $T_{j} \mapsto T_{j}'$ for all $0 \leq j \leq i$ (see Footnote~\ref{fn:lists}). If the algorithm has correctly decided that $\Phi_i$ does not extend to an isomorphism $G_i \to H_i$, then on that nondeterministic branch it rejects and need not continue further checking.

So now suppose instead that for some fixed $0 \leq i < m$, our algorithm has correctly decided that $\Phi_{i}$ extends to an isomorphism between $G_{i}$ and $H_{i}$. We will show that our algorithm can correctly decide whether the map $\Phi_{i+1}$ extends to an isomorphism between $G_{i+1}$ and $H_{i+1}$. 

\begin{itemize}
\item \textbf{Case 1:} Suppose that $SN_{i+1}$ does not specify a Ree group. To check that $G_{i+1}$ and $H_{i+1}$ are isomorphic, we essentially verify that the guessed data provides presentations of both $G_{i+1}$ and $H_{i+1}$ of the form given in Construction~\ref{construction:BGKLP}. 

The rest of the verifier's algorithm is to check that $G_{i+1}$ and $H_{i+1}$ in fact satisfy the guessed presentation. From the standard name $SN_{i+1}$, using the uniform short presentation, we construct a presentation $\langle \Gamma_{i+1} | R_{i+1} \rangle$ of the corresponding finite simple group. Precisely, we use Proposition~\ref{prop:UniformityPSU} for $\PSU_{3}(q)$, Proposition~\ref{prop:UniformSuzuki} for $\Sz(q)$, and \cite{BGKLP} for the remaining finite simple groups (excluding $^2G_{2}(q)$).

 If $|\Gamma_{i+1}| \neq |T_{i+1}|$ reject. 

Now use the SLPs guessed in (ND3) to construct the corresponding relators (Rel1)--(Rel3) from Construction~\ref{construction:BGKLP}; that is, use $S_{i+1}^{nd}$ in place of $S_{i+1}$ in that construction, use $D^{nd}(\gamma)$ in place of $D(\gamma)$, etc. Now we apply these relators using the guessed elements $T_{i+1}$ in place of $\widehat{\Gamma}_{i+1}$ from Construction~\ref{construction:BGKLP}. If any relation fails in $G_{i+1}$, reject. Similarly apply the same relators to the guessed elements $T'_{i+1}$ in place of $\widehat{\Gamma}_{i+1}$, and if any relation fails in $H_{i+1}$, reject.

For each $1 \leq i \leq m$, let $g_i$ denote the first element of $T_i$ (and, resp., $h-i$ the first element of $T'_i$). The machine checks the co-nondeterministic tries from (coND1) to try all words over $T_1 \cup \dotsb \cup T_{i-1}$ of length $O(\log^2 |G|)$ to see if any are equal to $g_i$, and analogously for $T'_1 \cup \dotsb \cup T'_{i-1}$ and $h_i$ in $H$. If any evaluates to $g_i$ (resp., $h_i$) then that co-nondeterministic branch rejects. For $i=1$, we directly check that both distinguished elements are non-identity. By the Reachability Lemma, these checks ensure that both successive quotients are nontrivial.

In total, the checks took polylogarithmic time.\footnote{And, with the possible exception of constructing the presentations according to the Uniform Short Presentation Conjecture, the rest of the procedure also took only logarithmic space.} This completes the description and complexity analysis of the algorithm in the case of having a uniform short presentation.

\item \textbf{Case 2:} Suppose instead that $SN_{i+1}$ specifies $^2G_{2}(q_{i+1})$. In this case, we use the cohomological framework of Grochow and Qiao \cite{GQCoho} (recalled in Section~\ref{sec:Cohomology}) to decide whether $G_{i+1} \cong H_{i+1}$. Let $\alpha_i \colon G_i \to H_i$ denote the unique isomorphism extending $\Phi_i$. Let $B_{i+1}$ be the generating set for $^2 G_2(q_{i+1})$ from Lemma~\ref{lem:Ree-bijection}, which can be computed from the standard name in $\DTISPpll$. If $|T_{i+1}| \neq |B_{i+1}|$, reject. 

The algorithm proceeds with the following checks:
\begin{itemize}
\item \textbf{Normality and Action Compatibility:} For each $t \in T_{i+1}$ and each $g \in T_{1} \cup \cdots \cup T_{i}$, $tgt^{-1} \in \langle T_{1} \cup \cdots \cup T_{i} \rangle = G_{i}$, we use the word $C^{nd}(g,t)$ from (ND4-Action) and verify that 
\[
tgt^{-1} = C^{nd}(g,t) \qquad \Phi_{i+1}(t) \Phi_i(g) \Phi_{i+1}(t)^{-1} = C^{nd}(\Phi_i(g), \Phi_{i+1}(t)).
\]
This ensures that $G_i$ is normal in $G_{i+1}$, $H_i$ is normal in $H_{i+1}$, and that the map $\Phi_{i+1}$ (if it extends to a homomorphism) is in fact an action compatibility.

\item \textbf{Marked Isomorphism of the Quotients:} Our next step will be to verify that $G_{i+1}/G_{i}$ and $H_{i+1}/H_{i}$ are both isomorphic to $^2G_{2}(q_{i+1})$, and that the map $T_{i+1} \mapsto T_{i+1}'$ induces this isomorphism. Using the co-nondeterministically quantified SLPs from (coND4-Quotient), say we have an SLP $\sigma_1$ in $T_{i+1}$ and $\sigma_2$ in $T_1 \cup \dotsb \cup T_{i}$. The machine checks whether $\sigma_1(T_{i+1}) = \sigma_2(T_1 \cup \dotsb \cup T_{i})$. If not, we consider the guesses ``invalid'', and that co-nondeterministic branch automatically accepts (sic! so that the ultimate accept/reject criterion of the earlier nondeterministic branch depends only on other co-nondeterministic branches). If $\sigma_1(T_{i+1}) = \sigma_2(T_1 \cup \dotsb \cup T_{i})$, then in particular, we have $\sigma_1(T_{i+1}) \in G_i$. In this case, the machine now checks deterministically that $\sigma_1(B_{i+1}) = \text{Id}_7$ in $\GL(7,q_{i+1})$. If not, then this co-nondeterministic branch rejects (and hence, so does its parent nondeterministic branch). The machine then checks similarly for $H_{i+1} / H_i$.

\item \textbf{Cohomology Class Isomorphism.} In order to check that the chosen identification $\Phi_{i+1}$ extends to a cohomology class isomorphism, we use Lemma~\ref{lem:Ree-bijection} to get a section $s_{G,i+1} :{}^2 G_2(q_{i+1})  \to G_{i+1}$ as follows. Given $M \in {}^2 G_2(q_{i+1})$, we compute $\varphi^{-1}(M)$ as in Lemma~\ref{lem:Ree-bijection} to write $M$ as a word $w_M$ in the generating set $B_{i+1}$, $w_M(B_{i+1})=M$. Since we have identified $T_{i+1}$ with $B_{i+1}$, we then use $w_M(T_{i+1})$ as our element of $G_{i+1}$ that maps to $M$ under the given isomorphism $G_{i+1} / G_i \to {}^2 G_2(q_{i+1})$. That is, we define the section $s_{G,i+1} \colon {}^2 G_2(q_{i+1}) \to G_{i+1}$ by $s_{G,i+1}(M) = w_M(T_{i+1})$. On branches of the co-nondeterministic machine for $\varphi^{-1}$ that output $\dagger$, our machine simply accepts, so that its ultimate accept/reject decision only relies on the co-nondeterministic branches that correctly computed $\varphi^{-1}(M)$ (and we have shown that the latter such branches indeed exist, by our definition of $\coNTISPpllMV$ and Lemma~\ref{lem:Ree-bijection}). 

To line up this final check with the statement and notation of Lemma~\ref{lem:GQMain}, we introduce all the relevant notation in our present setting. Let $\alpha \colon G_i \to H_i$ be the unique isomorphism extending $\Phi_i$, and let $\beta \colon G_{i+1} / G_i \to H_{i+1} / H_i$ be the isomorphism of the composition factors induced by $\Phi_{i+1}$. Let $\theta_{G,i+1}$ be the action of $G_{i+1} / G_i$ on $G_i$ by conjugation, using the section $s_{G,i+1}$ above, and analogously for $\theta_{H,i+1}$. By Remark~\ref{rmk:t-trivial}, we may assume the $t$ needed in the definition of pseudo-congruence (Definition~\ref{def:pseudocongruence}) is trivial.

We then define the 2-cocycle $f_{G,i+1} : {}^2 G_2(q_{i+1}) \times  {}^2 G_2(q_{i+1}) \to G_{i}$ by
\[
f_{G,i+1}(q,r) = s_{G,i+1}(q) \cdot s_{G,i+1}(r) \cdot s_{G,i+1}(qr)^{-1}.
\]
Define $f_{H,i+1}$ analogously. We need to verify that for all $q, r \in  {}^2 G_2(q_{i+1})$, the following holds:
\[
\alpha(f_{G,i+1}(q,r)) = f_{H,i+1}^{1, \theta_{H,i+1}}(\beta(q), \beta(r)).
\]

In order to do so, we use the co-nondeterministically tried SLPs $P_{q,r}$ from (coND4-CCI). If $P_{q,r}(T_1 \cup \dotsb \cup T_{i+1}) \neq f_{G,i+1}(q,r)$, then that branch simply accepts, thus effectively leaving the final accept/reject decision in the hands of those co-nondeterministic branches in which $P_{q,r}(T_1 \cup \dotsb \cup T_{i}) = f_{G,i+1}(q,r)$. By construction, the 2-cocycle $f_{G,i+1}(q,r)$ is in $G_i$; hence, by the Babai--Szemerédi Reachability Lemma, there is some $P_{q,r}$ satisfying the former equality. On the co-nondeterministic branch in which such a $P_{q,r}$ is found, the machine then (deterministically) checks whether $P_{q,r}(T'_1 \cup \dotsb \cup T'_{i}) = f_{H,i+1}(q,r)$, and accepts if and only if they are indeed equal. Otherwise, that co-nondeterministic branch rejects.
\end{itemize}

\end{itemize}

\noindent \textbf{Correctness.} We now turn to establishing the correctness of our argument. We will show that if all of the above checks pass, then $G \cong H$; and in fact, the map $\Phi_{m}$ extends to an isomorphism between $G$ and $H$. For this argument, we proceed inductively along the $m$ steps of a composition series of $G$ (note that at no point did the algorithm compute such a composition series; if $G$ or $H$ does not have a composition series of length $m$, then one of the above checks will reject).

The proof is by induction on $0 \leq i \leq m$. We first consider the base case when $i = 0$. In this case, $T_{0} = T_{0}' = \emptyset$. So $G_{0} = G_{0}' = \{1\}$. Now fix $i \geq 0$, and suppose that the map $\Phi_{i}$ extends to an isomorphism $\alpha$ between $G_{i}$ and $H_{i}$. We will now consider the $i+1$ case. 
\begin{itemize}
\item \textbf{Case 1:} Suppose first that $SN_{i+1}$ is not $^2G_{2}(q)$. Because of the relations in (Rel3), we have that $G_i$ is normal in $G_{i+1} := \langle G_i, T_{i+1} \rangle$. Because of the relations in (Rel2), we have that each relator of $SN_{i+1}$ applied to $T_{i+1}$ lies in $G_i$, and hence that $G_{i+1} / G_i$ is a quotient of $SN_{i+1}$ (for convenience, we abuse notation by referring to the standard name $SN_{i+1}$ and the group it specifies interchangeably). 

The verification of the co-nondeterministic bits from (coND1) ensures that at least one element of $T_{i+1}$ is not in $G_i$, and hence that $G_{i+1} / G_i$ is nontrivial, and thus in fact isomorphic to $SN_{i+1}$. Similarly $\langle H_i, T'_{i+1} \rangle / H_i \cong SN_{i+1}$. The exact values of (Rel2) (not modulo $G_i$) along with the exact values of (Rel3) (and the supporting relators from (Rel1)) then imply that $G_{i+1} \cong H_{i+1}$ \cite[Proof of Theorem~8.3]{BGKLP}, and this isomorphism extends the bijection $T_1 \cup \dotsb \cup T_{i+1} \to T'_1 \cup \dotsb \cup T_{i+1}$ described at the start of the proof of correctness. Thus, by induction, $G_{i+1} \cong H_{i+1}$, as claimed.

Finally, by \cite[Theorem~8.3]{BGKLP}, if $G_{i+1} \cong H_{i+1}$, then there are non-deterministic guesses that the algorithm will accept, namely those from Construction~\ref{construction:BGKLP}. Thus there exist non-deterministic guesses that are accepted by the verifier if and only if $G_{i+1} \cong H_{i+1}$.

\item \textbf{Case 2:} Suppose instead that $SN_{i+1} = {}^2G_{2}(q_{i+1})$. The normality check above ensures that $G_{i} \trianglelefteq \langle G_{i}, T_{i+1} \rangle = G_{i+1}$ (and respectively, that $H_{i} \trianglelefteq \langle H_{i}, T_{i+1}' \rangle = H_{i+1}$). 

Next, if the checks in the Marked Isomorphism of the Quotients point pass, then by the Babai--Szemerédi Reachability Lemma, the words corresponding to the co-nondeterministically checked SLPs in which $\sigma_1(T_{i+1})$ was in $G_i$ give a presentation of $G_{i+1} / G_i$ from the generators $T_{i+1}$. The checks then imply that these relators are also valid relators for the generating set $B_{i+1}$ of $^2 G_2(q_{i+1})$, and hence that the map defined by $T_{i+1} \mapsto B_{i+1}$ extends uniquely to a surjective homomorphism $G_{i+1} / G_i \to {}^2 G_2(q_{i+1})$. Similarly, the map $T'_{i+1} \mapsto B_{i+1}$ gives a surjective homomorphism $H_{i+1} / H_i \to {}^2 G_2(q_{i+1})$.

This, on its own, is not yet enough to ensure that the preceding maps are isomorphisms, but we will now argue that if the algorithm succeeds at all stages, then in fact the preceding maps must actually be isomorphisms not merely surjections. For, if at all stages $i$ the checks pass, then for the non-Ree-group composition factors $G_{i+1} / G_{i}$ the machine will have verified that they were indeed isomorphic to the guessed names $SN_{i+1}$. For the factors such that $SN_{i+1}$ was guessed to be a Ree group, what the checks have verified is that $G_{i+1} / G_{i}$ surjects onto the guessed $^2 G_2(q_{i+1})$. In particular, this implies that $|G_{i+1} / G_{i}| \geq |^2 G_2(q_{i+1})|$. But since this is true for all $i$, and the machine has also verified (at the very start of the verification) that $\prod_{i} |SN_i| = |G|$, we cannot have $|G_{i+1} / G_{i}| > |SN_{i+1}| = |^2 G_2(q_{i+1})|$, for then we would have $|G| \geq |G_M| =\prod_i |G_{i} / G_{i-1}| > \prod_i |SN_i| = |G|$, a contradiction. Thus, if all checks at all stages passed, then it must have in fact been the case that the maps $T_{i+1} \mapsto B_{i+1}$ extended uniquely not just to surjective homomorphisms, but in fact to \emph{isomorphisms} $G_{i+1} / G_{i} \stackrel{\cong}{\to} {}^2 G_2(q_{i+1})$, as desired. 

The same argument applies to $H$ \emph{mutatis mutandis}. Thus, since the above conclusion is also true for the maps $T'_{i+1} \mapsto B_{i+1}$, composing one of these isomorphisms with the inverse of the other also ensure that the map $T_{i+1} \mapsto T'_{i+1}$ extends uniquely to an isomorphism of composition factors $G_{i+1} / G_i \to H_{i+1} / H_i$.

We now have isomorphisms $\alpha : G_{i} \cong H_{i}$, $\beta : G_{i+1}/G_{i} \cong H_{i+1}/H_{i}$, and sections $s_{G,i+1} : G_{i+1}/G_{i} \to G_{i+1}$ and $s_{H, i+1} : H_{i+1}/H_{i} \to H_{i+1}$ induced by $T_{i+1}, T_{i+1}'$ respectively. The correctness for the remainder of this case follows from Lemma~\ref{lem:GQMain}.
\end{itemize}

Thus, by induction, our algorithm correctly determines whether $G \cong H$. \\

\noindent \textbf{Complexity.} At each stage, we use $\polylog(n)$ existentially-quantified non-deterministic 
and universally-quantified co-nondeterministic bits, though these can all be guessed/tried in parallel before any verification is done (that is, we do not have $m$ many stages of alternating quantifiers). 

Furthermore, the deterministic computation is $\DTIMEpl$-computable. As $m \leq \lceil \log |G| \rceil$, we obtain our bound of
\[
\exists^{\polylog(n)} \forall^{\polylog(n)}\DTIMEpl,
\]
as desired. \qedhere
\end{proof}

\begin{remark}[On uniformity]
As is visible in the structure of the above proof, if we assume only the (non-uniform) Short Presentation Conjecture, then the same proof gives a non-uniform analogue of Theorem~\ref{thm:MainUpperBound}, namely that \algprobm{Group Isomorphism} is solvable by non-uniform depth-$2\frac{1}{2}$ circuits. Furthermore, if the uniformity in the Uniform Short Presentation Conjecture is upgraded from $\DTIMEpl$ to $\DTISPpll$, then Proposition~\ref{prop:GpI}  upgrades by replacing $\DTIMEpl$ with $\DTISPpll$. In the non-uniform setting, the non-uniform advice at length $n$ consists of short presentations of all finite simple groups of order $\leq n$; note that this only depends on $n$ and not on the input groups $G,H$. Since there are at most two isomorphism types of simple groups of each order (and none of odd order, by the Feit--Thompson Odd Order Theorem), there are at most $n$ such groups (and in fact many fewer, but this bound suffices for the argument), and the total number of bits for this advice is $O(n \polylog(n))$. When the $\DTIMEpl$ machine would need to convert a standard name $SN_i$ into a short presentation, it instead queries that short presentation from the advice.\footnote{For these purposes, we may even assume that the $\DTIMEpl$ machine has ``associative memory'' access to the advice, and can literally look up the presentation by name, rather than worrying about where it is placed in terms of bits. For when it is then converted into a non-uniform Boolean formula, those wires will simply have hard-coded 0s and 1s coming from the short presentations.}
\end{remark}

\section{Lower Bounds}
Under the Uniform Short Presentation Conjecture, the upper bound of Section 3 expresses \algprobm{Group Isomorphism} as the conjunction of a polynomial-size CNF checking validity of the input Cayley tables and a quasipolynomial-size DNF deciding isomorphism under that promise. This suggests asking whether either component can be replaced by a formula of the other type, thus yielding circuits of depth 2. In Section~\ref{sec:DNF}, we show that the validity CNF cannot be replaced by a quasipolynomial-size DNF: any DNF deciding validity, and hence any DNF deciding total \algprobm{Group Isomorphism}, requires at least $2^{n \log n-O(n)}$ terms for groups of order $n$. For the other ``half,'' we prove in Section~\ref{sec:CNF} an $n^{\Omega(\log n)}$ lower bound for CNFs deciding the promise version of \algprobm{Group Isomorphism}. This rules out polynomial-size CNFs but does not rule out quasipolynomial-size CNFs.  Section~\ref{sec:coNP} discusses the significance of such an upper bound and explains an obstruction to obtaining a stronger lower bound.
In Section~\ref{sec:CNF-XOR} we extend the $\mathsf{CNF}$ lower bound to $\ComplexityClass{CNF \circ XOR}$; although the latter result might seem to subsume the former, we keep the proofs separate for several reasons. First, it allows us to introduce the various ideas in the proof one at a time, for easier reading. Second, the results apply to groups with slightly different parameters, and in that sense, neither subsumes the other. Third, it is interesting to contrast the techniques in the two proofs; see Remark~\ref{rmk:twoproofs}.

\subsection{Exponential DNF lower bound} \label{sec:DNF}
We get our exponential lower bound on DNFs for \algprobm{Group Isomorphism} from the same lower bound on DNFs for \algprobm{Validity Testing}: testing whether a given multiplication table is in fact the multiplication table of a group. 

\begin{observation} \label{obs:validity}
Validity testing of Cayley tables reduces to \algprobm{GpI} by a $\mathsf{DLOGTIME}$-uniform projection reduction. Testing whether a multiplication table is a valid quasigroup (=Latin square) similarly reduces to \textsc{Quasigroup Isomorphism} and \algprobm{Latin Square Isotopy}.
\end{observation}

We use Observation~\ref{obs:validity} for our unconditional lower bound below. 
\begin{proof}
If $T$ is a multiplication table, then $(T,T)$ is a yes-instance of \algprobm{GpI}if and only if $T$ is a valid multiplication table of a group. (The same proof works for quasigroups, \emph{mutatis mutandis}.) Thus if \algprobm{GpI} has DNFs of size $O(f(n))$, then so does \algprobm{Validity Testing}. \qedhere
\end{proof}

\begin{theorem} \label{thm:DNF}
Any (uniform or non-uniform) DNF on multiplication tables of order $n$ that is 1 on all Cayley tables of groups and 0 on all other multiplication tables must have at least $n! /n^{\log n} \geq 2^{n \log n - O(n)}$ terms. Hence the same lower bound holds on DNFs for \algprobm{GpI}. 

For Latin squares instead of groups, we get an analogous lower bound of $(n!)^{2n} / n^{n^2} \geq 2^{n^2 \log n - O(n^2)}$.
\end{theorem}

Note that the number of input variables in either case is $n^2 \lceil \log_2 n \rceil =: N$; as a function of the number of input variables the lower bound for groups is $2^{\Omega\left(\sqrt{N \log N}\right)}$, while for quasigroups it is $2^{N(1 - 1/\log N)}$, nearly maximal.

\begin{proof}
The key fact is that there are no two multiplication tables of quasigroups that differ in only one entry, hence \emph{a fortiori} they also cannot differ in only one bit. For if $T$ is the multiplication table of a (quasi)group and $T'$ differs from $T$ in only one entry, say $T(x,y) \neq T'(x,y)$, then the $x$-th row of $T'$ is not a permutation. (If we remove one entry from a row, the value of that entry is uniquely determined as the only entry missing from the rest of the row.) 

Now suppose $\varphi$ is a DNF that decides validity. The total number of variables is $N = n^2 \lceil \log_2 n \rceil$. If any term $t$ of the DNF has width $< N$, then $t$ is $1$ on a sub-cube of dimension $\geq 1$, contradicting the fact from the previous paragraph. 
Thus every term has width exactly $N$, hence is $1$ on exactly one input, so the number of terms must be equal to the number of valid tables. 

In the case of groups, each group of order $n$ has automorphism group of size at most $n^{\log n}$, so there are at least $n! / n^{\log n}$ distinct tables corresponding to each group. Thus the total number of valid tables is at least $n! / n^{\log n} \geq 2^{n \log n - O(n)}$.\footnote{When $n=p^k$ and $p=O(1)$, the number of isomorphism types is $n^{\Omega(\log^2 n)}$ \cite{Higman,Sims}, but this only affects the lower-order terms in the exponent of the asymptotic count of the number of valid Cayley tables.}

In the case of quasigroups, the number of Latin squares of order $n$ is at least $(n!)^{2n} / n^{n^2}$ (e.\,g., \cite[Thm.~17.2]{vanLintWilsonBook}), so we get the same lower bound on the number of terms of a DNF that decides Latin squares. 
\end{proof}

\subsection{Quasi-polynomial CNF lower bound} \label{sec:CNF}

\newcommand{\rank}{\operatorname{rank}}

In this section we prove:

\begin{theorem}\label{cor:two-input} \label{thm:CNF}
Any (uniform or non-uniform) CNF solving the promise \algprobm{Group Isomorphism} problem on multiplication tables of groups of order $n=2^{2k}$ has size at least $n^{\Omega(\log n)}$.
\end{theorem}

The groups we use for the lower bound will be Abelian of exponent $4$, and the constant in the exponent is at least $1/16$.

\newcommand{\cylinder}{cube\xspace}
\newcommand{\cylinders}{cubes\xspace}
\newcommand{\Cylinder}{Cube\xspace}
\newcommand{\Cylinders}{Cubes\xspace}

\begin{proof}[Proof idea]
The proof builds groups $G_M$ from $k \times k$ binary matrices $M$, so that the isomorphism type of $G_M$ is completely determined by $\rank(M)$. Thus \algprobm{Matrix Rank} reduces (with exponential size increase) to \algprobm{Group Isomorphism}; in particular, \algprobm{Group Isomorphism} contains within it the problem of distinguishing matrices of rank $k/2$ from matrices of rank $< k/2$. When a CNF clause is restricted to these Cayley tables, its falsifying set is defined by affine-linear constraints on the rows of $M$ (a feature of our reduction), with distinct rows having independent constraints; such a set is called a row-\cylinder. If such a row-\cylinder avoids all rank $k / 2$ matrices but covers some rank $k / 2-1$ matrix, then it cannot be too large. But there are exponentially many rank $k / 2-1$ matrices, so covering all of them requires $2^{\Omega\left(k^2\right)}$ clauses. Since $k=\Theta(\log n)$, this is exactly $n^{\Omega(\log n)}$.
\end{proof}

Throughout the following, $k$ will denote a natural number and $M$ will denote a $k \times k$ matrix. We
build a group $G_M$ from such a matrix $M$ as follows.

\begin{construction} \label{construction:gulce}
Given $M \in \text{Mat}_k(\FF_2)$, the elements of $G_M$ are pairs $(u,z)\in\mathbb{F}_2^k\times\mathbb{F}_2^k$, so
$|G_M|=2^{2k}$.  Multiplication is defined by
$
  (u,z)\cdot_M(v,w) \;=\; \bigl(u+v,\; z+w+M(u\odot v)\bigr),
$
where $u\odot v$ denotes the coordinate-wise (Hadamard) product,
$(u\odot v)_i=u_iv_i$.
\end{construction}

We use \algprobm{Matrix Rank} over a field $\FF$ to refer to the language 
\[
\{(M,r) : M \text{ is a square matrix over $\FF$ and }\rank(M) = r\}.
\]
We use $I_r \oplus 0_{k-r}$ to denote the $k \times k$ diagonal matrix $\text{diag}(1,1,\dotsc,1,0,0,\dotsc,0)$ of rank exactly $r$.

\begin{proposition}\label{prop:structure}
The group $G_M$ of Construction~\ref{construction:gulce} is an Abelian $2$-group of exponent at most $4$, with isomorphism type
$
  G_M \;\cong\; (\mathbb{Z}_4)^{\rank(M)}
              \times (\mathbb{Z}_2)^{2k-2\rank(M)}.
$
In particular,
$
  G_M \cong G_N \;\Longleftrightarrow\; \rank(M)=\rank(N).
$

Furthermore, for each fixed $r$ the function $(M,r) \mapsto (G_M, G_{I_r \oplus 0_{k-r}})$ is an $\mathbb{F}_2$-affine projection reduction from \algprobm{Matrix Rank} over $\mathbb{F}_2$ from the rank $r$ predicate to \algprobm{Group Isomorphism}, in which each bit of the output depends on at most one row of $M$.
\end{proposition}

\begin{proof}[Proof of Proposition~\ref{prop:structure}]
A direct computation shows
$
  (u,z)^2 \;=\; (0,\,Mu),
$
so the subgroup of squares is
$
  G_M^{(2)} \;:=\; \{x^2: x\in G_M\}
         \;=\; \{(0,Mu): u\in\mathbb{F}_2^k\}
         \;=\; \{0\}\times\operatorname{im}(M),
$
and every element has order dividing $4$.
If $\rank(M)=r$, then $\operatorname{im}(M)$ is an $r$-dimensional subspace
over $\mathbb{F}_2$, giving $|G_M^{(2)}|=2^r$.  The square subgroup detects the number of
$\mathbb{Z}_4$-factors: in $G_M$ that number is exactly $\rank(M)$.  The
stated isomorphism type follows, and two such groups are isomorphic if and only if they
have the same rank.

We have already established that the function $(M,r) \mapsto (G_M, G_{I_r \oplus 0_{k-r}})$ is a many-one reduction (of exponential stretch) from \algprobm{Matrix Rank} over $\FF_2$ to \algprobm{Group Isomorphism}. It remains to show that this reduction is an $\FF_2$-affine projection reduction in which each bit of the output depends on at most one row of $M$. For this, we examine the definition of the product in $G_M$: $(u,z)\cdot_M(v,w) = (u+v, z + w + M(u \odot v))$. The first component, $u+v$, does not depend on $M$, so its coordinates are in fact constants. The $j$-th coordinate of the second component is 
\[
  z_j+w_j+\sum_{i=1}^k M_{ji}\,u_iv_i,
\]
which depends only on row $j$ of $M$, and for each fixed pair $(u,v)$, is $\FF_2$-affine linear in that row of $M$. 
\end{proof}

\begin{remark}[Odd characteristic]
This particular construction only works in characteristic $2$. The same construction in characteristic $p > 2$ has $(u,z)^i = \left(iu, iz + \binom{i}{2}M(u \odot u)\right)$ by induction on $i$, and hence $(u,z)^p = (0,0)$ if $p > 2$, since $\binom{p}{2}$ is divisible by $p$ when $p > 2$, so the resulting group is always an elementary Abelian group of exponent $p$. For odd $p$ we can nonetheless get a similar result using Baer's correspondence \cite{Baer}: given a $k \times k$ matrix $M$ over $\FF_p$, we may define the group $G_{M,p}$ whose underlying set is $\FF_p^{2k} \times \FF_p$ with multiplication defined by $(u,\alpha) \cdot (v,\beta) = \left(u + v, \alpha + \beta + \frac{1}{2} u^\top \begin{pmatrix} 0 & M \\ -M^\top & 0 \end{pmatrix} v\right)$. This group is isomorphic to the direct product of the generalized Heisenberg group of order $p^{2\rank(M) + 1}$ (which is non-Abelian and directly indecomposable) and the Abelian group $\mathbb{Z}_p^{2(k- \rank(M))}$. However, this is no longer an affine projection because it requires encoding mod-$p$ arithmetic using Boolean gadgets. To get around the latter issue, one could instead consider ``CNFs over mod-$p$ atoms'', in which we allow atomic formulae $[x_i=\alpha]$ where $\alpha \in \FF_p$ to take the place of literals, and then build CNFs on top of those atoms. In that case, this again becomes an $\FF_p$-affine projection reduction, and the subsequent arguments will then show that ``CNFs over mod-$p$ atoms''  that solve promise \algprobm{Group Isomorphism} must have size at least $n^{\Omega(\log n)}$.
\end{remark}

\subsubsection{Additional Preliminaries: Row-\Cylinders}

For our lower bound, we view the matrix space as a product over its rows. Namely, we identify $\mathbb{F}_2^{k \times k} \cong\left(\mathbb{F}_2^k\right)^k$ by writing a matrix $M$ as its ordered tuple of rows $M=(m_1,\dots,m_k)$ where $m_j \in \mathbb{F}_2^k$. We call a set of the form: $F=A_1\times\cdots\times A_k$ where each $A_j \subseteq \mathbb{F}_2^k$ is an affine subspace, a row-\cylinder. Such products are cubes in the terminology of \cite{Beame2010} with the additional property of affineness. They are also special cases of cylinder intersections which are mostly used in multi-party communication complexity: indeed, $A_1 \times \cdots \times A_k=\bigcap_{i=1}^k\left(A_1 \times \cdots \times A_{i-1} \times \mathbb{F}_2^k \times A_{i+1} \times \cdots \times A_k\right)$ where the $i$th set in the intersection is a cylinder in the $i$th row-coordinate. We use the name row-\cylinders to emphasize the row-wise product structure.

Row-\cylinders have three elementary properties that we will use repeatedly.
First, they are themselves affine subspaces of
$\mathbb{F}_2^{k\times k}$.  Indeed, if $A_j=a_j+U_j$, where $U_j$ is a
linear subspace of $\mathbb{F}_2^k$, then
$
  F=A+U,
\,\,
  U=U_1\times\cdots\times U_k,
$
where $A$ is the matrix whose $j$-th row is $a_j$.  Consequently,
$
  \dim F=\sum_{j=1}^k \dim U_j
  \,\,\text{and}\,\,
  |F|=2^{\dim F}.
$

Second, any two matrices in a row-\cylinder can be joined by changing one
row at a time.  More precisely, let
$
  X=(x_1,\ldots,x_k),\,\, Y=(y_1,\ldots,y_k)
$
belong to $F$.  For $t=0,\ldots,k$, define
$
  X^{(t)}
  :=
  (y_1,\ldots,y_t,x_{t+1},\ldots,x_k).
$
Every $X^{(t)}$ belongs to $F$, since its $j$-th row belongs to $A_j$:
it is either $y_j\in A_j$ or $x_j\in A_j$.  Thus
$
  X=X^{(0)},X^{(1)},\ldots,X^{(k)}=Y
$
is a path inside $F$ in which consecutive matrices differ in at most one row.

We will also use the following theorem

\begin{theorem}[{\cite[Theorem~6]{deSeguinsPazzis2010}, cf. \cite{Meshulam1985}}] \label{thm:dSP}
Over an arbitrary field $\FF$, if $\mathcal{A}$ is an
affine subspace of $\mathbb{F}^{k\times k}$ and every matrix in
$\mathcal{A}$ has rank at most $s$, then
$
  \dim \mathcal{A}\leq ks.
$
\end{theorem}

\subsubsection{Returning to the lower bound} \label{sec:ReturningToCNF}

\begin{lemma}\label{lem:row-cylinder}
Fix $r \leq k \in \mathbb{N}$. After restricting to the family of pairs of multiplication tables $\{(G_M, G_{I_r \oplus 0_{k-r}}) : M \in \text{Mat}_k(\FF_2)\}$, the falsifying set of each CNF clause has the form
$
  F \;=\; A_1\times A_2\times\cdots\times A_k,
$
where each $A_j\subseteq\mathbb{F}_2^k$ is an affine subspace of possible
choices for row $j$ of $M$. Thus $F$ is a row-\cylinder in the sense
defined above.
\end{lemma}

\begin{proof}
By the second half of Proposition~\ref{prop:structure}, each bit of the multiplication table of $G_M$ (in our chosen way of representing elements) depends $\FF_2$-affine-linearly on at most one row of $M$. 
Therefore each input bit read by a clause is either a constant or an affine linear function of exactly
one row of $M$, and the falsifying set of the clause imposes independent affine linear
constraints on each row.
\end{proof}

\begin{lemma}\label{lem:meshulam}
Let $F=A_1\times\cdots\times A_k$ be a row-\cylinder over a field $\FF$. Then the set of ranks of matrices occurring in $F$ is a contiguous set of integers. 

In particular, if $F$ contains at least one rank-$s$ matrix but no rank-$(s+1)$ matrix, then every matrix in $F$ has rank at most $s$. Consequently, such $F$ satisfy
$
  \dim_{\FF} F \;\leq\; ks.
$
\end{lemma}

\begin{proof}
Suppose $X,Y \in F$ have ranks $r = \rank(X) < \rank(Y) = r'$. By the second property of the row-\cylinders as described above, there is a sequence
$X=X^{0},X^{1},\ldots,X^{k}=Y$
of matrices in $F$ such that consecutive matrices differ in at most
one row. Consequently their ranks differ by at most one. Thus every rank in between $r$ and $r'$ must occur as the rank of at least one matrix in the sequence $X^0, \dotsc, X^k \in F$. This establishes the first claim.

The ``in particular'' then follows immediately. Once we have $\max_{X \in F} \rank(X) \leq s$, \cite[Theorem~6]{deSeguinsPazzis2010} (reproduced as Theorem~\ref{thm:dSP} above) implies that $\dim F \leq ks$, 
completing the proof of the lemma. \qedhere
\end{proof}

The following estimate is standard, but crucial for our lower bound:

\begin{lemma} \label{lem:count}
Over $\FF_q$, there are at least $q^{r(2k-r)-2k}$ matrices of size $k \times k$ and rank exactly $r$.
\end{lemma}

\begin{proof}
The matrix product function $\FF_q^{k \times r} \times \FF_q^{r \times k} \to \FF_q^{k \times k}$ has image consisting exactly of the matrices of rank $\leq r$. The image $AB^\top$ has rank exactly $r$ iff $\rank A = \rank B = r$, and in that case, the pre-image of any such matrix is a set of the form $\{(Ag, g^{-1}B) : g \in \GL(r,q)\}$. Let $N(a,b,r)$ be the number of $a \times b$ matrices of rank exactly $r$. Then we have
\[
N(k,k,r) = N(k,r,r)^2 / |\GL(r,q)|.
\]
Now, we can calculate $N(k,r,r)$ as follows. The $r$ rows must be linearly independent; there are $q^k-1$ choices for the first row, then $q^k-q$ choices for the second row that are linearly independent from the first, then $q^k - q^2$ choices for the next row that are linearly independent from the first two, and so on. Thus 
\begin{align*}
N(k,r,r) & = (q^k-1)(q^k-q)(q^k-q^2) \dotsb (q^k - q^{r-1}) \\
 & = q^{\binom{r}{2}} \cdot (q^k-1)(q^{k-1}-1) \dotsb (q^{k-r+1}-1)
\end{align*}
Similarly, 
\begin{align*}
|\GL(r,q)| = N(r,r,r) 
 & = q^{\binom{r}{2}} (q^r-1) (q^{r-1}-1) \dotsb (q-1) 
\end{align*}
Thus
\begin{align*}
N(k,k,r) & = \frac{q^{2\binom{r}{2}} \cdot [(q^k-1)(q^{k-1}-1) \dotsb (q^{k-r+1}-1)]^2}{q^{\binom{r}{2}} (q^r-1) (q^{r-1}-1) \dotsb (q-1) } \\
 & = q^{\binom{r}{2}} \cdot \frac{[(q^k-1)(q^{k-1}-1) \dotsb (q^{k-r+1}-1)]^2}{(q^r-1) (q^{r-1}-1) \dotsb (q-1)}
\end{align*}
Now, we upper-bound the denominator using $q^i - 1 < q^i$, and we lower-bound the numerator using 
\[
q^i - 1 = q^i(1 - 1/q^i) \geq q^i / q
\]
for all $q \geq 2$ and $i \geq 1$. This gives us the estimate:
\begin{align*}
\log_q N(k,k,r) & \geq \binom{r}{2} + 2\left((k-1) + (k-2) + \dotsb + (k-r+1)\right) - \binom{r+1}{2} \\
 & = \binom{r}{2} - \binom{r+1}{2} + 2\left(\binom{k}{2} - \binom{k-r+1}{2}\right) \\
 & = \frac{1}{2}\left(r(r-1)-r(r+1) + 2k(k-1) - 2(k-r+1)(k-r)\right) \\
 & = r(2k-r) - 2k \qedhere
\end{align*}
\end{proof}

Finally, we now have all the ingredients for the proof of our main CNF lower bound, Theorem~\ref{thm:CNF}:

\begin{proof}[Proof of Theorem~\ref{thm:CNF}]
Let $C=\Gamma_1\wedge\cdots\wedge\Gamma_m$ be a CNF solving the promise \algprobm{GpI} problem. For a matrix $M$ and integer $r = s+1$, we consider applying $C$ to the pair $(G_M, G_{I_r \oplus 0_{k-r}})$ (the output of the reduction from Proposition~\ref{prop:structure}). Let $F_i=\{M:\Gamma_i(G_M,G_{I_r \oplus 0_{k-r}})=0\}$ be the set of matrices $M$ such that the output of the reduction from Proposition~\ref{prop:structure} falsifies the clause $\Gamma_i$.  By
Lemma~\ref{lem:row-cylinder} each $F_i$ is a row-\cylinder.  Correctness of $C$ implies:
\begin{enumerate}
  \item No clause is false on a rank-$r$ matrix:
        $F_i\cap\{M:\rank(M)=r\}=\varnothing$ for every $i$.
  \item Every matrix of rank $\neq r$ falsifies some clause, in particular,
        $\{M:\rank(M) \leq r-1 = s\}\subseteq\bigcup_{i=1}^m F_i$.\footnote{One might be able to sharpen the argument a bit by including here also the matrices of rank $> r$, but that can at most improve the constant in the exponent: for even if we replace our lower bound on $|\bigcup F_i|$ with the the size of the set of all $k \times k$ matrices, our final lower bound would still be of the form $2^{\Omega(k^2)}$, which is asymptotically what we already achieve.}
\end{enumerate}
Condition (1) says that each $F_i$ contains no rank-$r$ matrix.  Thus each $F_i$ that covers at least one rank-$s$ matrix satisfies the hypothesis of Lemma~\ref{lem:meshulam}, giving $|F_i|\leq 2^{ks}$. The $F_i$ that contain only smaller-rank matrices are even smaller (at most $2^{k(s-1)}$), and those that contain no matrices of rank $< r$ do not contribute to the sum at all. 
Next, condition (2) gives us that 
\[
|\{M : \rank(M) \leq s\}| \leq \sum_{i=1}^m |F_i| \leq \sum 2^{ks} = m2^{ks}.
\]
By Lemma~\ref{lem:count}, the left-hand side here is at least $2^{s(2k-s)-2k}$. Putting these together we get
\[
2^{s(2k-s)-2k} \leq m 2^{ks}
\]
hence $m \geq 2^{ks - s^2-2k}$. Since our CNF must work for all ranks $s$, we may choose $s$ to optimize this quantity, which happens at $s=k/2$. Choosing $s=k/2$, we get $m \geq 2^{k^2/2 - (k^2/4)-2k} = 2^{(1/4)k^2-2k}$. Since $n = 2^{2k}$, we have $k = (1/2)\log_2 n$, hence $m \geq n^{\frac{1}{16}(\log_2 n)-1} \geq n^{\Omega(\log n)}$.
\end{proof}

\begin{remark}
Jeřábek also previously established bounds of $\ComplexityClass{\Sigma_{2} TIME}(\log^{2} n)$ for isomorphism testing of Abelian groups, which can be simulated by depth-$2\frac{1}{2}$ $\qACz$-circuits of size $n^{O(\log n)}$ \cite{JerabekTCSStackExchange}.
\end{remark}

Construction~\ref{construction:gulce} has a natural tensor analogue; although one might hope to use a tensor analogue to improve the lower bound, in the remainder of this section we discuss a significant obstacle to this approach. 

\begin{definition}
Let $p$ be an odd prime. For class-$2$ exponent-$p$ Baer groups, the parameter space is the space of skew-symmetric matrix tuples
$
  X \;=\; \bigl(\Lambda^2 V^*\bigr)^{\ell}, \, V=\mathbb{F}_p^k.
$
An element $T\in X$ is an $\ell$-tuple of skew-symmetric matrices
$T=(M^1,\ldots,M^\ell)$, or equivalently a skew-symmetric bilinear map
$b_T:V\times V\to\mathbb{F}_p^\ell$.  The associated Baer group is
$G_T=V\times\mathbb{F}_p^\ell$ with multiplication
$
  (v,z)(v',z') \;=\; \Bigl(v+v',\; z+z'+\tfrac{1}{2}b_T(v,v')\Bigr).
$
\end{definition}

\begin{proposition}[Baer's Correspondence {\cite{Baer}}]\label{prop:baer-iso}
Let $p$ be an odd prime. Let $H=\mathrm{GL}_\ell(p)\times\mathrm{GL}_k(p)$ act on $X$ by simultaneous change of
basis in $\mathbb{F}_p^\ell$ and in $V$.  Then
$
  G_T \;\cong\; G_S \;\Longleftrightarrow\; T\cong S \;\Longleftrightarrow\;
  S = h\cdot T \;\text{ for some } h\in H.
$
\end{proposition}

Proposition~\ref{prop:baer-iso} plays the role that the rank classification
 plays in the matrix setting: it reduces \algprobm{Group
Isomorphism} to a natural orbit problem.  The \algprobm{Skew-Symmetric (or Alternating) Matrix Space Isometry} problem is
accordingly to decide whether a given pair $(T,S) \in X \times X$ satisfies $T \cong S$; that is, the 
yes-set is
$
  \mathcal{Y} \;=\; \{(T,S)\in X\times X : T\cong S\}
             \;=\; \{(T,h\cdot T) : T\in X,\; h\in H\}.
$

\begin{remark}[Obstacle to using tensors in our lower bound strategy] 
There is an important structural difference between the matrix and tensor settings. A fixed tensor orbit has size at most $ |\operatorname{Orb}(R)| \leq |\mathrm{GL}_\ell(p)|\,|\mathrm{GL}_k(p)| = p^{O(k^2+\ell^2)}. $ In the balanced regime $\ell\asymp k$, this is only $p^{O(k^2)}$, whereas the full tensor space has size $ |X| = p^{\ell\binom{k}{2}} = p^{\Theta(k^3)}. $ Thus, a fixed tensor orbit is asymptotically thinner (in the exponent) than a generic rank layer in the matrix setting (viz. there are $p^{k^2}$ matrices of size $k \times k$ and $p^{\Theta(k^2)}$ such matrices of rank $k/2$). 
For this reason, the natural object for a covering argument is not a single fixed orbit in $X$ (in the rank setting we considered only the orbit of a rank-$k/2$ matrix), but rather the two-input isomorphism relation $ \mathcal{Y} := \{(T,S)\in X\times X : T\cong S\}.$ 

One might then try to imitate the rank argument directly on $X\times X$. In canonical Baer labels, Cayley-table entries are affine-linear functions of the tensor coordinates, so the falsifying set of a width-$w$ CNF clause is, as in the matrix case, an affine subspace of $X\times X$ of codimension at most $w$. A tensor analogue of the rank lower bound would therefore follow from showing that every sufficiently large affine subspace of $X\times X$ intersects $\mathcal{Y}$; equivalently, that every yes-free affine subspace has large codimension or covers very little mass. This strategy, however, fails on the full space $X\times X$. Indeed, assume that $k$ is even, fix a nonzero vector $v\in V$, and fix a nondegenerate skew-symmetric form $J\in \Lambda^2 V^*$. Let $\mathcal{A}\subseteq X$ be the linear subspace consisting of tensors $T$ for which $v\in\operatorname{rad}(T)$. This imposes $\ell(k-1)$ linear constraints. Let $\mathcal{B}\subseteq X$ be the affine subspace consisting of tensors $S=(N_1,\ldots,N_\ell)$ with $N_1=\sum_{i=1}^{k/2} (E_{2i,2i+1} - E_{2i+1,2i})$, the standard full-rank skew-symmetric matrix. This has codimension $\binom{k}{2}$, and every $S\in\mathcal{B}$ has trivial radical. Consequently, $\mathcal{A}\times\mathcal{B}$ is a yes-free affine subspace of $X\times X$ (because the radical is an isomorphism invariant) of codimension \[ \ell(k-1)+\binom{k}{2}=O(k^2) \] in the balanced regime $\ell\asymp k$. Thus the tensor analogue of the rank argument cannot be applied directly to all of $X\times X$. 
\end{remark}

\subsection{Quasi-polynomial $\ComplexityClass{CNF \circ XOR}$ lower bound} \label{sec:CNF-XOR}
In this section we use the reduction of Proposition~\ref{prop:structure} together with a new structural argument about certain affine subspaces (which will no longer be row-cubes) to extend the lower bound of Theorem~\ref{thm:CNF} from CNFs to $\textsf{CNF} \circ \textsf{XOR}$:

\begin{theorem}  \label{thm:CNF-XOR}
Any (uniform or non-uniform) $\textsf{CNF}$-of-$\textsf{XOR}$s solving the promise \algprobm{Group Isomorphism} problem on multiplication tables of groups of order $n=2^{2k}$ has size at least $n^{\Omega(\log n)}$.
\end{theorem}

The groups we use for the lower bound are the same as in Theorem~\ref{thm:CNF} (just with a different choice of rank), and the constant in the exponent here is at least $1/8$. The new key algebraic/combinatorial result we will use (essentially in place of Theorem~\ref{thm:dSP} and Lemma~\ref{lem:meshulam}) is:

\begin{theorem} \label{thm:space-missing-k-1}
Let $\FF$ be any field. Suppose $V \subseteq \Mat_k(\FF)$ is an affine subspace that contains at least one invertible matrix, and contains no matrix of rank exactly $k-1$. Then $\dim V \leq \binom{k+1}{2}$.
\end{theorem}

In the proof we talk of the $(-1)$-eigenspace, but we need not assume that $-1 \neq 1$, we just write it this way to make clear that it works over arbitrary fields.

\begin{proof}
Let $V$ be as in the statement. By assumption there is an invertible matrix $A_0 \in V$. Let $V' := A_0^{-1} \cdot V$. Then $V'$ is an affine subspace, with $\dim V' = \dim V$, so it suffices to bound $\dim V'$. Furthermore,  since left multiplication by an invertible matrix preserves ranks, $V'$ also contains no matrix of rank exactly $k-1$, and since $A_0 \in V$, we have $I \in V'$ ($I$=the $k \times k$ identity matrix). Thus $V' = I + V_0$ for a unique linear subspace $V_0 \subseteq \Mat_k(\FF)$.

We now give an alternative characterization of the rank constraint on $V_0$:

\begin{lemma} \label{lem:V0}
For any subspace $V_0 \subseteq \Mat_k(\FF)$, $I + V_0$ contains no matrix of rank exactly $k-1$ if and only if $V_0$ contains no matrix with 1-dimensional $(-1)$-eigenspace.
\end{lemma}

\begin{proof}
 We have $\rank(I+X)=k-1$ iff $\dim\ker(I+X)=1$. But $(I+X)v=0$ iff $v=Iv = -Xv$, so $\rank(I+X)=k-1$ iff the $(-1)$-eigenspace of $X$ is 1-dimensional. 
\end{proof}

In the remainder of the proof of Theorem~\ref{thm:space-missing-k-1}, we will prove by induction on $k$ that any linear subspace of $k \times k$ matrices in which no matrix has a 1-dimensional $(-1)$-eigenspace has dimension at most $\binom{k+1}{2}$. The inductive step will only arise near the end.

Throughout, when we speak of a matrix ``having a 1-dimensional $(-1)$-eigenspace'', we mean that the $(-1)$-eigenspace is exactly 1-dimensional. By the \emph{support} of a matrix $A$ we mean the set of indices $(i,j)$ where $A_{ij} \neq 0$. For each $i=1,\dotsc,k$, let $R_i(V_0)$ denote the subspace of $V_0$ consisting of matrices whose support lies entirely within the $i$-th row.

\begin{lemma} \label{lem:Aii}
If $V_0$ is a vector subspace of $\Mat_k(\FF)$ in which no matrix has 1-dimensional ($-1$)-eigenspace, and $A \in R_i(V_0)$ is not the zero matrix, then $A_{ii} = 0$. 
\end{lemma}

\begin{proof}
Since $A$ is a nonzero matrix in $R_i(V_0)$, we have $\im(A)=\Span\{e_i\}$, since $\im(A)$ is the column-span of $A$, and every column is a scalar multiple of $e_i$. If $A_{ii} \neq 0$, then let $A' = (-1/A_{ii})A$, which is also in $R_i(V_0)$. Then $A' e_i = -e_i$, so $e_i$ is an eigenvector of $A'$ of eigenvalue $-1$, and since $e_i$ is the only nonzero vector in $\im(A) = \im(A')$, $e_i$ is the unique such vector up to scale, contradicting the assumption. This proves the lemma.
\end{proof}

\begin{lemma} \label{lem:Ri}
If $V_0$ is a vector subspace of $\Mat_k(\FF)$ in which no matrix has 1-dimensional ($-1$)-eigenspace, 
then there exists $i$ such that $R_i(V_0)=0$.
\end{lemma}

\begin{proof}
Suppose to the contrary that for every $i$, $R_i(V_0) \neq 0$, that is, that $V_0$ contains a nonzero matrix whose support is contained in the $i$-th row. From this supposition, we will show how to construct a matrix in $V_0$ which has a 1-dimensional $(-1)$-eigenspace, contradicting the assumption on $V_0$ in the statement of the lemma. Let us call such a matrix ``bad'' (just for this proof).

Let $G$ be the directed graph on $k$ vertices that has a directed edge $(i,j)$ iff there exists $A_i \in R_i(V_0)$ with $(A_i)_{ij} \neq 0$. We will show that our supposition implies that $G$ has at least one directed cycle. And then we will show that if $G$ has any directed cycles, then we can use any cycle in $G$ of minimum length to construct a bad matrix $A \in V_0$.

\textbf{$G$ has a directed cycle:} 
Since we have assumed $R_i(V_0)$ is not the zero space, and any $A \in R_i(V_0)$ has $A_{ii} = 0$ by Lemma~\ref{lem:Aii} the graph $G$ has the property that every vertex has at least one outgoing edge, and it has no self-loops. Since $G$ is a finite graph in which every vertex has out-degree at least $1$, by following some sequence of such edges, there must exist at least one directed cycle in $G$. 

\textbf{Constructing a bad matrix from a directed cycle in $G$: }
Now let $C = (i_0,i_1,\dotsc,i_{\ell-1})$ be a directed cycle in $G$ of minimum length $\ell$. For each $j=0,\dotsc,{\ell-1}$, let $A_{j}$ be a matrix in $R_{i_j}(V_0)$ with $(A_j)_{i_j, i_{j+1 \pmod{\ell}}} = -1$. We can find a matrix in $R_{i_j}(V_0)$ with that entry nonzero by the definition of $G$ and the fact that $i_0,i_1,\dotsc,i_{\ell-1}$ is a cycle in $G$; $A_j$ is then the scaled version of such a matrix to make the relevant entry equal to $-1$. Let $A = \sum_{j=0}^{\ell-1} A_j$; we will show that $A$ is bad.

Let $e_i$ denote the standard basis column vector with a $1$ in position $i$ and $0$s everywhere else. We claim that the vector $e := \sum_{j=0}^{\ell-1} e_{i_j}$ spans the entire $(-1)$-eigenspace of $A$.

Note that $Ae$ is the negative of the sum of the columns in $A$ indexed by $i_0,i_1,\dotsc,i_{\ell-1}$. Since $A$, by construction, is zero outside of the rows $i_0,\dotsc,i_{\ell-1}$, when calculating $Ae$ we are effectively focusing on the $\ell \times \ell$ sub-matrix whose rows and columns are both indexed by $\{i_0,i_1,\dotsc,i_{\ell-1}\}$. Denote this sub-matrix by $B$. We claim that $B$ has exactly $\ell$ nonzero entries (corresponding to the cycle $C$). For the subgraph of $G$ induced by the entries of $B$ is a cycle of length $\ell$. If there is \emph{any} other edge $(i_j, i_m)$  between the vertices of $G$ that are in the cycle $C$ (but where the edge is not part of the cycle), then $(i_0,i_1,\dotsc,i_j,i_m,i_{m+1},\dotsc,i_{\ell-1})$ is a cycle strictly shorter than $C$, contradicting the minimality of $C$'s length. Thus $B$ consists of a single cycle, and we get $Ae = -e$. 

To see that $A$ has no other $(-1)$-eigenvector (i.\,e., one that is linearly independent from $e$), note that $A$ is supported on the rows $\{i_0,\dotsc,i_{\ell-1}\}$, so $\im(A)\subseteq \Span\{e_{i_0}\dotsc,e_{i_{\ell-1}}\}$ and thus every eigenvector $v$ of $A$ lies in this span. 
Let $v=\sum_{j=0}^{\ell-1}x_j e_{i_j}$ be any eigenvector of $A$. Since the restriction of $A$ to those coordinates is exactly $-1$ times the adjacency matrix of the directed cycle $C$, the equation $Av=-v$ implies $x_j=x_{j+1 \pmod{\ell}} \,\,\text{for every} \,\,j $. Hence $x_0 = x_1 = \cdots = x_{\ell-1}.$ 

Therefore the only eigenvectors of $A$ are scalar multiples of $e=e_{i_0}+\cdots+e_{i_{\ell-1}}$, and thus $A$ is bad, as claimed. This contradicts our assumption on $V_0$, thus there must exist some $i$ such that $R_i(V_0) = 0$, completing the proof of Lemma~\ref{lem:Ri}. \qedhere
\end{proof}

Finally, we show how to use the inductive hypothesis to complete the proof of Theorem~\ref{thm:space-missing-k-1}. Recall that our inductive hypothesis is that for all $k' < k$, if $V_i$ is a subspace of $k' \times k'$ matrices that contains no matrix whose $(-1)$-eigenspace is 1-dimensional, then $\dim V_i \leq \binom{k'+1}{2}$.

Without loss of generality (by permuting rows if necessary), by Lemma~\ref{lem:Ri} we may assume that $R_k(V_0) = 0$. Let $\phi \colon V_0 \to \FF^k$ be the linear map $\phi(A) = Ae_k$. Then we have
\begin{align*}
\dim V_0 & = \dim \ker \phi + \dim \im \phi \\
 & \leq \dim \ker \phi + k
\end{align*}
Now, $\ker(\phi)$ consists of all matrices in $V_0$ whose last column is the zero vector. Let $V_1$ be the subspace of $(k-1) \times (k-1)$ matrices that occur as the upper-left $(k-1) \times (k-1)$ sub-matrix of matrices in $\ker \phi \subseteq V_0$. Let $\pi \colon \ker \phi \to V_1$ be the natural projection map. 

We claim that $\pi$ is an isomorphism of vector spaces, that is, equivalently, that matrices in $\ker \phi$ are uniquely identified by their first $k-1$ rows. By construction of $V_1$, this holds iff $\pi$ is injective. Suppose $X \in \ker \phi$ has $\pi(X) = 0$. Then the last column of $X$ is zero because $X$ is in $\ker \phi$, and the first $(k-1)$ rows of $X$ are zero because $\pi(X)=0$. But then $X \in R_k(V_0)$, which is zero (by Lemma~\ref{lem:Ri} and permuting rows as needed). Thus $\pi$ is injective, hence is an isomorphism $\ker \phi \to V_1$. 

Now, $V_1$ is a linear subspace of $(k-1) \times (k-1)$ matrices. For a matrix $X$, let $\lambda_{-1}(X)$ denote its $(-1)$-eigenspace.

\begin{lemma} \label{lem:fix}
For $X \in \ker \phi$, $\dim \lambda_{-1}(X) = \dim \lambda_{-1}(\pi(X))$.
\end{lemma}

In fact one can see in the proof of the lemma that the dimension of every eigenspace corresponding to a nonzero eigenvalue is preserved. We will only need this for the $(-1)$-eigenspace.

\begin{proof}
Let $\rho \colon \FF^k \to \FF^{k-1}$ be the projection onto the first $(k-1)$ coordinates, $\rho(v_1,\dotsc,v_k)^\top = (v_1,\dotsc,v_{k-1})^\top$. We will show that $\rho$ gives an isomorphism $\lambda_{-1}(X)$ to $\lambda_{-1}(\pi(X))$. 

Suppose $v \in \FF^k$ is a $(-1)$-eigenvector of $X \in \ker \phi$. Since the last column of $X$ is zero, the last entry of $v$ does not contribute to the first $k-1$ entries of $Xv$. Thus $\pi(X)\rho(v) = -\rho(v)$, so $\rho(v) \in \lambda_{-1}(\pi(X))$, and thus $\rho(\lambda_{-1}(X)) \subseteq \lambda_{-1}(\pi(X))$.

To see that $\rho$ is injective on $\lambda_{-1}(X)$, suppose $v'$ is another $(-1)$-eigenvector of $X$ with the same first $k-1$ coordinates as $v$. Then we have $X(v-v') = -(v-v')$ since $v-v' \in \lambda_{-1}(X)$, but also also $X(v-v')=0$ since $v$ and $v'$ differ only in their last entry, but the last column of $X$ is zero. Thus $\ker \rho = 0$.

To see that $\rho$ is surjective, we will give its inverse. Suppose $v' \in \FF^{k-1}$ is in $\lambda_{-1}(\pi(X))$. Let $r$ be the first $(k-1)$ entries of the last row of $X$. Let $v \in \FF^k$ be the vector that appends to $v'$ the single entry $-r \cdot v'$. Clearly $\rho(v) = v'$, but it remains to see that $v$ is in fact in $\lambda_{-1}(X)$. The last entry of $v$ does not affect the first $(k-1)$ coordinates of $Xv$, since the last column of $X$ is zero, so the first $(k-1)$ coordinates of $Xv$ are precisely $-v'$. Since $X_{kk}=0$, the last entry of $Xv$ is precisely $r \cdot v'$, which is $-v_k$, as desired. Thus $v \in \lambda_{-1}(X)$, so $\rho \colon \lambda_{-1}(X) \to \lambda_{-1}(\pi(X))$ is surjective, hence an isomorphism, completing the proof of the lemma.
\end{proof}

Finally, Lemma~\ref{lem:fix} shows us that our inductive hypothesis applies to $V_1$. By induction, we get a sequence of vector spaces $V_i \subseteq \Mat_{k-i}(\FF)$ satisfying the conclusion of Lemma~\ref{lem:V0}, and we get the recursion
\[
\dim V_i \leq \dim V_{i+1} + (k-i)
\]
Thus $\dim V = \dim V' = \dim V_0 \leq k + (k-1) + (k-2) + \dotsb = \binom{k+1}{2}$, completing the proof of Theorem~\ref{thm:space-missing-k-1}.
\end{proof}

We can now prove our lower bound on $\textsf{CNF} \circ \textsf{XOR}$. The beginning of the proof proceeds as in the proof of Theorem~\ref{thm:CNF}:

\begin{proof}[Proof of Theorem~\ref{thm:CNF-XOR}]
Let $C=\Gamma_1\wedge\cdots\wedge\Gamma_m$ be a $\textsf{CNF}$-of-$\textsf{XOR}$s solving the promise \algprobm{GpI} problem, where each $\Gamma_i$ is an $\textsf{OR}$ of $\textsf{XOR}$s of literals. For a matrix $M$, we consider applying $C$ to the pair $(G_M, G_{I_{k-1} \oplus 0_{1}})$ (the output of the reduction from Proposition~\ref{prop:structure}). Let $F_i=\{M:\Gamma_i(G_M,G_{I_{k-1} \oplus 0_1})=0\}$ be the set of matrices $M$ such that the output of the reduction from Proposition~\ref{prop:structure} falsifies the $\textsf{OR}$-of-$\textsf{XOR}$s $\Gamma_i$.  By Proposition~\ref{prop:structure}, each $F_i$ is an affine subspace of $k \times k$ matrices.
Correctness of $C$ implies:
\begin{enumerate}
  \item No clause is false on a rank-$(k-1)$ matrix:
        $F_i\cap\{M:\rank(M)=k-1\}=\varnothing$ for every $i$.
  \item Every matrix of rank $\neq k-1$ falsifies some clause, in particular,
        $\{M:\rank(M) =k \}\subseteq\bigcup_{i=1}^m F_i$.
\end{enumerate}
Condition (1) says that each $F_i$ contains no rank-$(k-1)$ matrix.  Thus each $F_i$ that covers at least one rank-$k$ matrix satisfies the hypothesis of Theorem~\ref{thm:space-missing-k-1}, giving $|F_i|\leq 2^{(k+1)k/2}$. 
We will only focus on trying to cover the rank-$k$ matrices. Re-index the clauses so that $F_1,\dotsc,F_{m'}$ each contain at least one full-rank matrix.  Condition (2) thus gives us that 
\[
|\{M : \rank(M) =k\}| \leq \sum_{i=1}^{m'} |F_i| \leq \sum 2^{(k+1)k/2} = m' 2^{(k+1)k/2} \leq m 2^{(k+1)k/2}
\]
By the standard count of the size of $\GL_k(\FF_2)$ (reproduced in the middle of the proof of Lemma~\ref{lem:count} above), we have the left-hand side here is at least $2^{k^2 - k}$. Putting these together we get
\[
2^{k^2 - k} \leq m 2^{k^2 / 2 + k/2}
\]
hence $m \geq 2^{k^2 / 2 - 3k/2}$. Since $n = 2^{2k}$, we have $k = (1/2)\log_2 n$, hence $m \geq n^{\frac{1}{8}(\log_2 n)-3/4} \geq n^{\Omega(\log n)}$.
\end{proof}

\begin{remark} \label{rmk:twoproofs}
We find two aspects of this proof worth remarking on, especially in comparison with Theorem~\ref{thm:CNF}. The first is that the two results work on different ranges of ranks: our proof of Theorem~\ref{thm:CNF} gives the asymptotically best lower bound it can when the rank and co-rank are both $\Theta(k)$, that is, when the rank is essentially $k/c$ for $0 < c < 1$, whereas our proof of Theorem~\ref{thm:CNF-XOR} focuses on rank exactly $k-1$. The second is that, while one might naturally expect (in light of Razborov--Smolensky \cite{Razborov, Smolensky87algebraicmethods}) that $\mathsf{Mod}_p$ gates for $p \neq 2$ would not help with isomorphism of $2$-groups, we precisely get a lower bound on $\ComplexityClass{CNF \circ Mod_2}$ circuits for isomorphism of $2$-groups.    
\end{remark}

\subsection{Uniform quasiCNF implies Group Isomorphism is in $\ComplexityClass{coNP}$} \label{sec:coNP}

\begin{observation}
$\DTIME(\polylog(n))$-uniform $2^{\polylog(n)}$-size CNFs can be simulated by 
\[
\forall^{\polylog(n)} \DTIME(n\polylog(n))
\subseteq \ComplexityClass{coNP}.
\]
\end{observation}

\begin{proof}
If $\varphi$ is a $\DTIME(\polylog(n))$-uniform $2^{\polylog(n)}$-size CNF, then $\varphi$ is true on input $x$ iff for all gates $i$, the following deterministic verifier accepts: $V(x,i)$ checks that $i$ is an OR gate of $\varphi$ (using the algorithm describing $\varphi$ from the definition of uniformity); if not, $V(x,i)$ automatically accepts. If $i$ is an OR gate, then for each literal $V$ checks whether that literal occurs in clause $i$ using the uniformity of $\varphi$. For those that do, $V(x,i)$ checks their values in $x$, and then determines whether clause $i$ is satisfied by $x$. Since there are $2n$ literals to check, and the uniformity assumption is that each can be checked to see whether it is an input to clause $i$ in $\polylog(n)$ time, the complexity result follows. 

Finally, we have
$
\forall^{\polylog(n)} \DTIME(n\polylog(n)) \subseteq \forall^{\poly(n)} \ComplexityClass{P} = \ComplexityClass{coNP}.
$
\qedhere
\end{proof}

As a consequence, if we were able to solve promise \algprobm{Group Isomorphism} in 
\[
\forall^{\poly(\log n)}\DTIMEpl,
\]
(one route to attempting to get depth-2 circuits), or even just by uniform quasipolynomial-size CNFs, 
then combining with Lemma~\ref{lem:valid-cayley-tables} would place \algprobm{GpI} into $\coNP$, which remains a significant open question. 

On the question of putting \algprobm{GpI} into $\ComplexityClass{coNP}$: Arvind and Tor\'an \cite{ArvindToran} showed that in the case of solvable groups, \algprobm{GpI} is \emph{almost} in $\ComplexityClass{NP} \cap \ComplexityClass{coNP}$. Precisely, they exhibited an Arthur--Merlin protocol for \algprobm{Group Non-Isomorphism} where Arthur used $O(\log^6 n)$ bits and Merlin used $O(\log^2 n)$ bits. Additionally, Arvind and Tor\'an showed that if $\ComplexityClass{EXP} \not \subseteq \textsf{i.o.-PSPACE}$, then  \algprobm{GpI} for solvable groups belongs to $\ComplexityClass{NP} \cap \ComplexityClass{coNP}$. Arvind and Tor\'an also note that the Uniform Short Presentation Conjecture remains a key  obstacle for extending their work to all finite groups.

\section{Uniformity of the Hulpke--Seress Short Presentations for $^2A_{2}(q) = \PSU_{3}(q)$} \label{app:PSU3}
\newcommand{\exponent}{e}

\begin{proposition}\label{prop:UniformityPSU}
Let $q=p^\exponent>2$, and let
$G:={}^{2}A_2(q)=\PSU_3(q)$.
Given the standard name of $G$, one can construct, in
$\DTIME(\polylog(|G|))$, a presentation of
$G$ of word length $O(\log^2 q)=O(\log^2 |G|).$
\end{proposition}

\subsection{Additional preliminaries for $\PSU_{3}(q)$} \label{sec:PSUPreliminaries}

In this section, we will recall additional preliminaries concerning $\PSU_{3}(q)$. We refer to \cite{HulpkeSeress} for a reference. 

Let $p$ be a prime, and let $q = p^{\exponent}$ for some $\exponent > 0$. The group $\SU_{3}(q)$ is a subgroup of $\SL(3,q^2)$, and $\PSU_3(q) = \SU_3(q) / Z(\SU_3(q))$. The center of $\SU_3(q)$ consists of scalar matrices, hence it is either trivial or has order $3$ (iff $3 | q+1$).
For $a \in \mathbb{F}_{q^{2}}$, we define the trace and norm maps by:
\[
\Tr(a):=a+a^q,\quad
\Nm(a):= \Nm_{\FF_{q^2}/\FF_q}(a) = a^{q+1}.
\]
Note that, as $\FF_{q^2}$ is a degree-$2$ extension of $\FF_q$, the automorphism $a \mapsto a^q$ plays a role analogous to conjugation in the complex numbers, and in this analogy we have $\Tr(a) = a + \overline{a}$ and $\Nm(a)=a\overline{a}$.
Let $S := \{ s \in \mathbb{F}_{q^2} : \Tr(s) = 0 \}$, and $S^{\times} := S \setminus \{0\}$. Note that $S$ is an additive subgroup of $\mathbb{F}_{q^2}$, 
and $S^{\times}$ is a coset of $\mathbb{F}_{q}^{\times}$ in $\mathbb{F}_{q^{2}}^{\times}$. For $\alpha\in\FF_{q^2}^\times$, write $S^\times/\alpha :=\{s/\alpha:s\in S^\times\}.$ The following fact is standard but we will use it a few times so we number it for ease of reference, and give its proof for the reader's convenience.

\begin{fact} \label{fact:trace-surjective}
The trace is a surjective $\FF_q$-linear map $\Tr \colon \FF_{q^2} \to \FF_q$.
\end{fact}

\begin{proof}
Its image lies in $\FF_q$ because $\Tr(a)^q = (a + a^q)^q = a^q + a^{q^2} = a^q + a = \Tr(a)$, and $\FF_q$ consists of all elements of its algebraic closure that satisfy $x^q=x$. It is $\FF_q$-linear because $\Tr(ab) = ab + a^q b^q$, and if $a \in \FF_q$, then $a^q=a$, so we get $a^q b^q = a b^q$, and thus $\Tr(ab) = a\Tr(b)$ for $a \in \FF_q$. Finally, to see it is surjective, we note that as an $\FF_q$-linear map to $\FF_q$, its image must be $0$ or all of $\FF_q$, so it suffices to show it is not always zero. To see this, write $\FF_{q^2} = \FF_q[y] / f(y)$ for some irreducible degree-2 polynomial $f \in \FF_q[y]$. Then every element of $\FF_{q^2}$ can be written uniquely as $a + by$ for $a,b \in \FF_q$. In particular, $y^q = \alpha + \beta y$ for some $\alpha,\beta \in \FF_q$. Then we have $\Tr(by) = by + b^q y^q = by + b^q (\alpha + \beta y) = b^q \alpha + y(b + \beta)$. Choosing $b \neq -\beta$ thus gives $\Tr(by) \neq 0$.
\end{proof}

For a subgroup $A\leq \mathbb{F}_{q^2}^{\times}$, we write
$
A^{(3)}:=\{a^3:a\in A\}
$
for its subgroup of cubes. Unparenthesized powers $A^d$ have their usual meaning as Cartesian powers.

\textbf{Convention.} We denote elements of $\PSU_3(q)$ using $3 \times 3$ matrices, which are representatives of their images modulo scalar matrices; to help remind the reader of this we put such matrices in [square brackets] rather than rounded parentheses.

For $x, y \in \mathbb{F}_{q^{2}}$, define:
\begin{align*}
u(x,y) := \begin{bmatrix}
1 & 0 & 0 \\
y & 1 & 0 \\
x & -y^{q} & 1
\end{bmatrix}.
\end{align*}

\noindent Now fix a generator $\rho$ of $\mathbb{F}_{q^{2}}^{\times}$, and define:
\begin{align*}
h := \begin{bmatrix} 
\rho & 0 & 0 \\ 
0 & \rho^{q-1} & 0 \\
0 & 0 & \rho^{-q}
\end{bmatrix}.
\end{align*}
Put
\[
   d:=\gcd(3,q+1),\quad
   M:=q^2-1,\quad
   N:=\frac{M}{d}.
\]
Then $h$ has order $N$ as an element of $\PSU_3(q)$ (that is, modulo scalar matrices): let $\hat{h}$ denote the same matrix but as an element of $\SU_3(q)$. Then $\hat{h}$ has order $M = q^2-1$, since $\hat{h}_{11}=\rho$ is a generator of $\FF_{q^2}^{\times}$. 
We have $h^a = 1$ in $\PSU_3(q)$ iff $\hat{h}^a$ is scalar in $\SU_3(q)$. The latter happens iff $a \equiv a(q-1) \equiv -qa \pmod{q^2-1}$. The fact that the smallest solution to these equations is $a=N = (q^2-1)/gcd(3,q+1)$ is then a straightforward exercise in modular arithmetic.

The subgroup $B = \langle U, h \rangle$ is known as the (or a) \emph{Borel subgroup} of $\PSU_{3}(q)$, where
\[
U = \{ u(x,y) : x,y \in \FF_{q^2}, x + x^q + y^{q+1} = 0 \} = \{ u(x,y) : \text{Tr}(x) + \text{Nm}(y) = 0 \}.
\]
Observe that $h$ normalizes $U$. In particular, $\langle U, h \rangle$ is a solvable group. Furthermore, with the convention $g^a=a^{-1}ga$, we have:
\[
   u(x,y)^{h}=u(\lambda x,\beta y),
   \quad \text{where} \quad
   \lambda:=\rho^{q+1},\quad
   \beta:=\rho^{2-q}.
\]
The element $\lambda$ is primitive in $\FF_q^\times$, and
$\gcd(M,2-q)=\gcd(3,q+1)=d.$ 

We will now describe the structure of $U$. The subgroup $U$ has order $q^3$ and center:
\[
Z(U) = \{ u(x,0) : x + x^q = 0 \} = \{ u(s,0) : s \in S \}.
\]
We get $|Z(U)| = q$, and $U/Z(U)$ is an elementary Abelian group of order $q^2$. So $Z(U)$ admits an $\mathbb{F}_{p}$-basis $z_1, \ldots, z_{\exponent}$, and there exist elements $v_1, \ldots, v_{2\exponent} \in U$ whose images constitute an $\mathbb{F}_{p}$-basis of $U/Z(U)$. So $\langle U, h \rangle$ is generated by $\{ z_1, \ldots, z_{\exponent}, v_{1}, \ldots, v_{2\exponent},h\}$. In particular, every element $g$ of $\langle U, h \rangle$ can be written uniquely as: 
\begin{align} \label{eq:DecPSU}
g = \left( \prod_{i=1}^{\exponent} z_{i}^{\alpha_{i}} \right) \cdot \left( \prod_{i=1}^{2\exponent} v_{i}^{\beta_{i}} \right) \cdot h^{\gamma},
\end{align}
where for all $i \in [\exponent]$ and $j \in [2\exponent]$, $0 \leq \alpha_i, \beta_j < p$, and $0 \leq \gamma < N$. We refer to the expression on the right-hand side of Equation~(\ref{eq:DecPSU}), which we can think of as a word in the above generating set, as $\text{dec}(g)$ (for ``decomposition''). 

Now $\PSU_{3}(q)$ is generated by $\langle U, h \rangle$ and an extra element:
\[
t := \begin{bmatrix}
0 & 0 & 1 \\
0 & -1 & 0 \\
1 & 0 & 0
\end{bmatrix}.
\]
In particular, $t^2 = 1$ and $h^{t} = h^{-q}$, so $t$ normalizes $\langle h \rangle$. 

Now for any $u=u(x,y) \in U \setminus \{1\}$, there exist unique $g(u), f(u) \in U$ and unique $h(u) \in \langle h \rangle$ such that
\begin{align} \label{eq:PSUSteinberg}
u(x,y)^{t} = g(u) \cdot h(u) \cdot t \cdot f(u).    
\end{align} 
Hulpke and Seress gave the following explicit matrix representations for $f(u), g(u), h(u)$, and $u^{h}$:
\begin{align*}
f(u(x,y)) = \begin{bmatrix}
1 & 0 & 0 \\
-y/x^{q} & 1 & 0 \\
1/x & y^{q}/x & 1
\end{bmatrix},  \quad
g(u(x,y)) = \begin{bmatrix}
1 & 0 & 0 \\
-y/x & 1 & 0 \\
1/x & y^q/x^q & 1
\end{bmatrix} \\
h(u(x,y)) = \begin{bmatrix}
x & 0 & 0 \\
0 & x^{q-1} & 0 \\
0 & 0 & x^{-q}
\end{bmatrix}, \quad
u(x,y)^{h} = \begin{bmatrix}
1 & 0 & 0 \\
y\rho^{2-q} & 1 & 0 \\
x\rho^{q+1} & -y^{q}\rho^{2q-1} & 1
\end{bmatrix}.
\end{align*}
Furthermore, Hulpke and Seress showed that if $3 \nmid (q+1)$, then $\PSU_{3}(q)$ admits a short presentation using only $3$ relators of the form (\ref{eq:PSUSteinberg}). If $3 \mid (q+1)$, then $\PSU_{3}(q)$ admits a short presentation using only $7$ relators of the form (\ref{eq:PSUSteinberg}). We will recall the precise choices for these relators in Construction~\ref{con:ShortPresentationPSU} below.

\subsection{Uniform presentation for the Borel subgroup of $\PSU_3(q)$}
\label{PreliminariesPSU}

In this section, we prove uniformity of a presentation of the Borel subgroup that is essentially that of \cite[Lemma~13]{HulpkeSeress}; the main difference is that where they merely state that generators with certain properties exist (and for their result any generators with those properties will do), we must construct such generators uniformly. The main result of this section is:

\begin{lemma} \label{lem:PSU3-Borel-compressed}
From the standard name (Definition~\ref{def:standard-name}) for $\PSU_3(q)$, in $\DTIME(\polylog(q))$ one can construct a presentation of the Borel subgroup $\langle U,h\rangle$ that has word length $O(\log^2 q)$. 
\end{lemma}

We start by finding elements (as in the previous section) that are a basis for the elementary Abelian group $Z(U)$, and finding elements whose images are a basis for the elementary Abelian group $U/Z(U)$.
Let $\rho\in \mathbb{F}_{q^{2}}^\times$ be the supplied primitive element in the standard name (Def.~\ref{def:standard-name}). 
 Define 
\begin{equation} \label{eq:m-delta}
m_\delta:= 
\begin{cases} 
0,&p=2,\\[1mm] 
(q+1)/2,&p\ne2, 
\end{cases}  
\end{equation}
and let $\delta:=\rho^{m_\delta}$. Then $\delta^q=-\delta$, and hence $S = \ker(\Tr)=\delta\FF_q$. Recall that $\lambda = \rho^{q+1}$. For $1\le i\le\exponent$, put $ z_i:=u(\delta\lambda^{i-1},0). $ Since $\lambda$ has degree $\exponent$ over $\FF_p$, the elements $1,\lambda,\ldots,\lambda^{\exponent-1}$ form an $\FF_p$-basis of $\FF_q$. Hence $z_1,\ldots,z_\exponent$ form an $\FF_p$-basis of $Z(U)$, and $ z_i=z_1^{h^{i-1}}. $ We next choose the generators for $U/Z(U)$. Since $\rho$ has degree $2\exponent$ over $\FF_p$, the elements $ \rho,\rho^2,\ldots,\rho^{2\exponent} $ form an $\FF_p$-basis of $\mathbb{F}_{q^{2}}$. We now compute generators for $U/Z(U)$.  To do so, we distinguish the following cases according to $d:=\gcd(3,q+1)\in\{1,3\}.$

\begin{itemize}
\item \textbf{Case 1:} Suppose that $d=1$.  
Using Gaussian elimination over $\mathbb F_p$, construct
$a_1\in\mathbb F_{q^2}$ satisfying \[ a_1+a_1^q=-\rho^{q+1}, \] and put $v_1:=u(a_1,\rho)$. Such $a_1$ is guaranteed to exist by Fact~\ref{fact:trace-surjective} and the fact that $\rho^{q+1}$ is in $\FF_q$. Then $v_1\in U$, since $\Tr(a_1)+\Nm(\rho) = -\rho^{q+1}+\rho^{q+1} = 0.$

For every integer $b$, $h^{-b}v_1h^b =u\!\left(\lambda^b a_1,\rho\beta^b\right),$ so its second parameter is $ \rho\beta^b=\rho^{\,1+(2-q)b}.$ We have $\rho^{1+(2-q)b} = \rho^{t}$ precisely when 
\[(2-q)b\equiv t-1\pmod M,\]
(where, recall, $M=q^2-1$). Since $d=\gcd(M,2-q)=1$,  for each $1\le t\le2\exponent$ there is a unique $b(t)\in\{0,\ldots,M-1\}$
satisfying the above congruence.  Define $v_t:=h^{-b(t)}v_1h^{b(t)}.$ Then the second parameter of $v_t$ is $\rho^t$.  Since
$\rho,\rho^2,\ldots,\rho^{2\exponent}$ is an $\mathbb F_p$-basis of $\mathbb \FF_{q^2}$, the images of
$v_1,\ldots,v_{2\exponent}$ form an $\mathbb F_p$-basis of $U/Z(U)$.
To compute $b(t)$ from $t$, the inverse of $2-q$ modulo $M$ is computed by the extended Euclidean
algorithm in $O(\log q)$ iterations. 

\item \textbf{Case 2:}
Suppose that $d=3$ and $\exponent\ge2$ (recall $q=p^\exponent$).  For $r\in\{1,2,3\}$, use Gaussian elimination over $\mathbb F_p$ to construct three distinct 
$a_r\in\mathbb F_{q^2}$ satisfying $a_r+a_r^q=-\rho^{r(q+1)}$, which exist by Fact~\ref{fact:trace-surjective}, and put $v_r:=u(a_r,\rho^r)$. 

For $1\le t\le2\exponent$, let $r(t)\in\{1,2,3\}$ be the unique integer satisfying $r(t)\equiv t\pmod 3.$ For every integer $b$, the second parameter of $h^{-b}v_{r(t)}h^b$ is
$\rho^{r(t)}\beta^b =\rho^{\,r(t)+(2-q)b}$. Now $\rho^{\,r(t) + (2-q)b} = \rho^t$ precisely when
$(2-q)b\equiv t-r(t)\pmod M.$ Both sides are divisible by $3$, and $\gcd\!\left(\frac{2-q}{3},\,\frac M3\right)=1.$
Thus the reduced integer congruence
\[
\frac{2-q}{3}\,b(t) \equiv \frac{t-r(t)}{3} \pmod N, \,\, N:=M/3,
\]
has a unique solution $b(t)\in\{0,\ldots,N-1\}$.  

Define $v_t:=h^{-b(t)}v_{r(t)}h^{b(t)}.$ Then the second parameter of $v_t$ is $\rho^t$, so the images of the
$v_t$ form an $\mathbb F_p$-basis of $U/Z(U)$.  

\item \textbf{Case 3:}
Suppose that $d=3$, and $\exponent=1$ (that is, $q=p$).  Then $U/Z(U)$ has
$\mathbb F_p$-dimension two, so a basis cannot contain representatives of all three nonzero $\langle h\rangle$-orbits. For $t\in\{1,2\}$, use Gaussian elimination to construct two distinct 
$a_t\in\mathbb F_{q^2}$ satisfying $a_t+a_t^q=-\rho^{t(q+1)},$ and put $v_t:=u(a_t,\rho^t).$ Such $a_t$ are guaranteed to exist by Fact~\ref{fact:trace-surjective} and the fact that $\rho^{q+1}$ is in $\FF_q$.

Since $\rho,\rho^2$ form an $\mathbb F_p$-basis of
$\mathbb \FF_{q^2}$, the images of $v_1,v_2$ form a basis of $U/Z(U)$. 
\end{itemize}

In order to uniformly compute the relations between our generators, we will need the following uniform method of solving the constructive membership problem in $U$ relative to the generating set we constructed above.

\begin{lemma} 
\label{lem:PSU3-dec} 
Given $g=u(x,y)\in U$ as a matrix, one can compute, in $\DTIME(\polylog(q))$, a word $ \operatorname{dec}_U(g) =z_1^{c_1}\cdots z_\exponent^{c_\exponent} v_1^{b_1}\cdots v_{2\exponent}^{b_{2\exponent}}, \,\, b_i,c_j\in\{0,\ldots,p-1\},$ which evaluates to $g$. Its binary-exponent length is $O(\log q)$. \end{lemma} 

\begin{proof} 
The matrix representing $g$ has entries $ y=g_{21}, \,\, x=g_{31},$ (see Section~\ref{sec:PSUPreliminaries}) so the parameters $x$ and $y$ are read directly from the input matrix. Using Gaussian elimination over $\FF_p$, compute the unique coefficients $b_1,\ldots,b_{2\exponent}\in\FF_p$ such that $ y=\sum_{t=1}^{2\exponent}b_t\rho^t. $ Compute, in the indicated order, the matrix $ W:=\prod_{t=1}^{2\exponent}v_t^{b_t}. $ The second parameter is additive under multiplication in $U$, so the second parameter of $W$ is $ \sum_{t=1}^{2\exponent}b_t\rho^t=y. $ Consequently, $ gW^{-1}=u(s,0) $ for some $s\in S$. The element $s$ is read from the $(3,1)$-entry of the matrix $gW^{-1}$. Using Gaussian elimination again, compute the unique coefficients $c_1,\ldots,c_\exponent\in\FF_p$ satisfying $ s=\sum_{i=1}^{\exponent}c_i\delta\lambda^{i-1}. $ It follows that $g= \left(\prod_{i=1}^{\exponent}z_i^{c_i}\right) \left(\prod_{t=1}^{2\exponent}v_t^{b_t}\right). $ All linear systems have dimension $O(\exponent)$ over $\FF_p$, and all powers are evaluated by repeated squaring. Hence the word is computed in time $\polylog(q)$ and has binary-exponent length $ O(\exponent\log p)=O(\log q). $ 
\end{proof}

For $g\in U$, let $\widehat{\operatorname{dec}}_U(g)$ denote the same word with the
matrix generators replaced by abstract symbols.

\begin{construction}[Uniform presentation of the Borel subgroup] \label{con:SPBorelPSU} 
Let $\mathcal P_B$ have generators
$
\widehat z_1,\ldots,\widehat z_\exponent$,
$\widehat v_1,\ldots,\widehat v_{2\exponent}$,
$\widehat h$.
We add the following relations so that
$\widehat z_1,\ldots,\widehat z_\exponent$ generate $Z(U)$;
$\widehat v_1,\ldots,\widehat v_{2\exponent}$ are lifts of a basis of
$U/Z(U)$; and $\widehat h$ normalizes the subgroup generated by the
$\widehat z_i$ and $\widehat v_t$:

\begin{enumerate}
\item Power and order relations.
\begin{itemize}
\item For $1\le i\le\exponent$, add
$
\widehat z_i^{\,p}=1.
$

\item For $1\le t\le2\exponent$, add
$
\widehat v_t^{\,p}
=\widehat{\operatorname{dec}}_U(v_t^p).
$
Since $U/Z(U)$ has exponent $p$, we have $v_t^p\in Z(U)$.
More explicitly, $v_t^p=1$ in odd characteristic, while
$v_t^2\in Z(U)$ in characteristic two.\footnote{In characteristic 2, the quotient-basis lifts $v_i$ may have order 4, so we use the correct power relations $\widehat{v}_t^2=\widehat{\operatorname{dec}}_U\left(v_t^2\right), \,, v_t^2 \in Z(U)$ rather than asserting that these lifts always have order $p$.}
Consequently,
$\widehat{\operatorname{dec}}_U(v_t^p)$ is a word only in
$\widehat z_1,\ldots,\widehat z_\exponent$.

\item Add
$
\widehat h^{\,N}=1.
$
\end{itemize}

\item ($\widehat{z}_1$ is central) Add
$
[\widehat z_1,\widehat z_i]=1
\quad(1\le i\le\exponent),
\,\,
[\widehat z_1,\widehat v_t]=1
\quad(1\le t\le2\exponent).
$

\item Add the following quotient-commutator relations. Since
$[U,U]\le Z(U)$, every right-hand side below is a word only in the
$\widehat z_i$.
\begin{itemize}
\item If $d=1$, add
$
[\widehat v_1,\widehat v_t]
=\widehat{\operatorname{dec}}_U([v_1,v_t])
\,\,(2\le t\le2\exponent).
$

\item If $d=3$ and $\exponent\ge2$, add
$
[\widehat v_r,\widehat v_t]
=\widehat{\operatorname{dec}}_U([v_r,v_t])
$
for $r\in\{1,2,3\}$ and $1\le t\le2\exponent$.

\item If $d=3$ and $\exponent=1$, add 
$
[\widehat v_1,\widehat v_2]
=\widehat{\operatorname{dec}}_U([v_1,v_2]).
$
\end{itemize}

\item For every $1\le i\le\exponent$, add
$
\widehat h^{-1}\widehat z_i\widehat h
=\widehat{\operatorname{dec}}_U(z_i^h),
$
and for every $1\le t\le2\exponent$, add
$
\widehat h^{-1}\widehat v_t\widehat h
=\widehat{\operatorname{dec}}_U(v_t^h).
$

\item Add the orbit relations
$
\widehat z_i
=\widehat h^{-(i-1)}
\widehat z_1
\widehat h^{\,i-1}
\,\,(2\le i\le\exponent).
$
If $d=1$, also add
$
\widehat v_t
=\widehat h^{-b(t)}
\widehat v_1
\widehat h^{\,b(t)}
\,\,(2\le t\le2\exponent).
$
If $d=3$ and $\exponent\ge2$, also add
$
\widehat v_t
=\widehat h^{-b(t)}
\widehat v_{r(t)}
\widehat h^{\,b(t)}
\,\,
\bigl(t\notin\{1,2,3\}\bigr).
$
\end{enumerate}
\end{construction}

We now have all the ingredients for a uniform presentation of the Borel subgroup:

\begin{proof}[Proof of Lemma~\ref{lem:PSU3-Borel-compressed}]
All displayed relations hold for the indicated matrices, so there is a surjective homomorphism from the presented group onto $B$.

Let $\widehat U$ be the subgroup generated by the $\widehat z_i$ and
$\widehat v_t$.  The conjugation relations and $\widehat h^N=1$ imply
that $\widehat h$ normalizes $\widehat U$.  Since every $\widehat z_i$
is an $\widehat h$-conjugate of $\widehat z_1$, and $\widehat z_1$
commutes with every generator of $\widehat U$, all the $\widehat z_i$ are
central in $\widehat U$.

Suppose first that $d=1$.  Every $\widehat v_i$ is an
$\widehat h$-conjugate of $\widehat v_1$.  Conjugating the commutator relations and expanding commutators of products expresses every $[\widehat v_i,\widehat v_j]$ as a word in the central generators
$\widehat z_1,\ldots,\widehat z_\exponent$.  If $d=3$ and $\exponent\ge2$, the
same argument uses the three representatives
$\widehat v_1,\widehat v_2,\widehat v_3$.  If $d=3$ and $\exponent=1$, the
only quotient commutator relation was included directly. 

The power relations now let us collect every word in $\widehat U$ to the form
$
   \widehat z_1^{c_1}\cdots\widehat z_\exponent^{c_\exponent}
   \widehat v_1^{b_1}\cdots\widehat v_{2\exponent}^{b_{2\exponent}}
$
with $0\le b_i,c_j<p$.
Hence $|\widehat U|\le p^{3\exponent}=q^3$.  The conjugation relations let us move
every occurrence of $\widehat h$ to the right, so the whole presented
group has order at most $q^{3}N=|B|$. The surjection onto $B$ is thus an isomorphism.

We now account for the length. There are $O(\exponent)$ relations. Each word has binary-exponent length $O(\log q)$, and every orbit exponent has $O(\log q)$ bits. Thus the total binary-exponent length is $O(\exponent\log q)\subseteq O(\log^2q)$. By Lemma \ref{lem:PSU3-dec}, every right-hand side is computed using field arithmetic, Gaussian elimination, and the  Euclidean algorithm in time polynomial in $\log q$. By Remark~\ref{rem:presentation-length-conventions} this can be uniformly converted into a presentation with the asymptotically the same word-length.
\end{proof}

\subsection{Proof of Proposition~\ref{prop:UniformityPSU}} \label{sec:uniformityPSU}
To complete the uniform presentation of $\PSU_3(q)$, we 
now make the three and seven additional relations of \cite[Theorems~15--16]{HulpkeSeress} uniform.  We choose their
first coordinates to be explicitly known powers of $\rho$. 

Recall several parameters from above: $M=q^2-1$, $\delta=\rho^{m_{\delta}}$, where $m_{\delta}$ is defined in \eqref{eq:m-delta}.
Let
\[
   \varepsilon_-:=
   \begin{cases}
   0,&p=2,\\[1mm]
   M/2,&p\ne2,
   \end{cases}
\]
so that $-1=\rho^{\varepsilon_-}$.  Put
\[
   x:=\delta\rho,
   \,\,
   s:=-\delta(\rho+\rho^q),
   \,\,
   x':=x+s=-\delta\rho^q.
\]
The element $\rho+\rho^q$ belongs to $\FF_q$ and is nonzero. Indeed, in characteristic two, equality $\rho+\rho^q=0$ would put $\rho$ in
$\FF_q$; in odd characteristic it would imply $\rho^{q-1}=-1$, so the order of $\rho$ would divide $2(q-1)<q^2-1$. Hence $s\in S^\times$. Moreover, $x\in\rho S^\times$
and $ x'=-\delta\rho^q =\delta\rho^{-1}(-\rho^{q+1}) \in S^\times/\rho$. 
Here, $S^\times/\rho$ denotes the set $\{s/\rho:s\in S^\times\}$.
The relevant first coordinates are known powers:
$\delta=\rho^{m_\delta}, \,\, x=\rho^{m_\delta+1}, \,\, x'=\rho^{m_\delta+q+\varepsilon_-}.$

Set $c:=u(\delta,0)=z_1,\,\,b:=u(s,0).$ 
Set $r_0:=-\Tr(x)\in\FF_q^\times.$ Indeed, $x\notin S$, so $\Tr(x)\ne0$. Since $s\in S$ and $x'=x+s$, we also have $-\Tr(x')=r_0.$

Below we use the standard deterministic square-root algorithm for finite
fields of odd order with a supplied quadratic nonresidue: given a known
nonsquare in $\mathbb F_q^\times$, quadratic residuosity testing and
square-root extraction for squares are computable in time polynomial in
$\log q$; see \cite{AMM77}.  When $p$ is odd, the primitive element
$\lambda=\rho^{q+1}$ is such a known nonsquare.

\begin{lemma}
\label{lem:PSU3-norm-lifts-compressed}
Given $r\in\FF_q^\times$, one can compute, in deterministic time
polynomial in $\log q$, an element $L(r)\in \FF_{q^2}^\times$ such that
$\Nm(L(r))=r.$
If $d=3$, one can moreover compute such an element $L_3(r)$ satisfying
$
L_3(r)\in(\FF_{q^2}^\times)^{(3)}.
$
\end{lemma}

\begin{proof}
Suppose first that $p=2$.  Put
$
L(r):=r^{2^{\exponent-1}}\in\mathbb F_q$ (recall $\exponent=\log_p q$).
Then $L(r)^2=r$.  Since $L(r)\in\mathbb F_q$, its relative norm is
$
\Nm_{\mathbb F_{q^2}/\mathbb F_q}(L(r))
   =L(r)^{q+1}
   =L(r)^2
   =r.
$
If $d=3$, then $3\mid(q+1)$, and hence
$
\mathbb F_q^\times
   =\langle\rho^{q+1}\rangle
   \subseteq(\mathbb F_{q^2}^{\times})^{(3)}.
$
Thus $L_3(r):=L(r)$ is already a cube in
$\mathbb F_{q^2}^{\times}$.

Now suppose instead that $p$ is odd.  The element
$\lambda=\rho^{q+1}$ is primitive, and hence nonsquare, in
$\mathbb F_q^\times$.  Using the deterministic square-root algorithm
with this known nonsquare \cite{AMM77}, we proceed as follows.

If $r$ is a square, choose $a\in\mathbb F_q^\times$ satisfying
$a^2=r$ and put
$
L(r):=a.
$ Then we have $\Nm(a)=a^2=r$, as desired. Next, if instead $r$ is a nonsquare, then $r/\lambda$ is a square, because it is the
quotient of two nonsquares.  In this case, choose $a\in\mathbb F_q^\times$ satisfying
$
a^2=r/\lambda
$
and put
$
L(r):=\rho a.
$
Now we also have $\Nm(\rho a)=\lambda a^2=r$, as desired.

Finally, suppose additionally that $d=3$.  Since $3\mid(q+1)$, every element of
$\mathbb F_q^\times$ is a cube in
$\mathbb F_{q^2}^{\times}$.  If $r$ is a square, choose $a^2=r$ and put
$
L_3(r):=a
$
as above.
If instead $r$ is a nonsquare, then $\lambda^3$ is also a nonsquare, because
$\lambda$ is primitive and $3$ is odd.  Hence $r/\lambda^3$ is a
square.  Choose $a\in\mathbb F_q^\times$ satisfying
$
a^2=r/\lambda^3
$
and put
$
L_3(r):=\rho^3a.
$
Here $a$ is a cube in $\mathbb F_{q^2}^{\times}$, and therefore
$\rho^3a$ is a cube.  Moreover,
$
\Nm(\rho^3a)
   =\lambda^3a^2
   =r.
$
All of these operations take time polynomial in $\log q$.
\end{proof}

For the presentation of $\PSU_3(q)$ \cite{HulpkeSeress}, we recall one additional type of relator. For a given element $u \in U$, define the relation:
\begin{equation} \label{eq:R-relator}
\mathcal R(u):\,\,
   \widehat t\,\widehat{\operatorname{dec}}_U(u)\,\widehat t
      =
   \widehat{\operatorname{dec}}_U(g(u)) \cdot 
   \widehat h^{m(u)} \cdot 
   \widehat t \cdot
   \widehat{\operatorname{dec}}_U(f(u)),
\end{equation}
where $m(u)$ is the explicitly known exponent of the first coordinate of $u$.

\begin{construction}
\label{con:ShortPresentationPSU}
Start with $\mathcal P_B$ (from Construction~\ref{con:SPBorelPSU}) and add the generator $\widehat t$. We now add the following relators:
\begin{itemize}
 \item $\widehat t^2 = 1$.
 \item $\widehat t\widehat h\widehat t=\widehat h^{-q}$
 \item We now add relators depending on the value of $d$ (1 or 3).
 \begin{itemize}
     \item \textbf{Case 1:} $d = 1$. Let $\nu := L(r_0)$ (from Lemma~\ref{lem:PSU3-norm-lifts-compressed}), $a:=u(x,\nu)$, $a':=u(x',\nu)=ab$ (where $x,x'$ were defined at the top of Section~\ref{sec:uniformityPSU}).  Add the three relations $\mathcal R(a)$, $\mathcal R(c)$, $\mathcal R(a')$ using \eqref{eq:R-relator}.  Here $ m(a)=m_\delta+1, \,\, m(c)=m_\delta, \,\, m(a')=m_\delta+q+\varepsilon_- $ modulo $M$.

     \item \textbf{Case 2:} $d = 3.$ let $\nu_0:=L_3(r_0),\,\,\kappa:=\rho^{q-1},\,\,\nu_j:=\nu_0\kappa^j\,(j=0,1,2).$ Then $\Nm(\kappa)=1$, and, since $q-1\equiv1\pmod3$, the elements $\nu_0,\nu_1,\nu_2$ lie in the three distinct cosets of $(\mathbb{F}_{q^2}^{\times})^{(3)}$ in $\mathbb{F}_{q^2}^{\times}$. Put $a_j:=u(x,\nu_j), \,\, a'_j:=u(x',\nu_j)=a_jb \,\,(j=0,1,2).$
     
     Add the seven relations $ \mathcal R(c) $ and $ \mathcal R(a_j), \,\, \mathcal R(a'_j) \,\,(j=0,1,2).$ The corresponding exponents are
$
   m(a_j)=m_\delta+1,
\,\,
   m(c)=m_\delta,
\,\,
   m(a'_j)=m_\delta+q+\varepsilon_-
$
modulo $M$.
 \end{itemize}
 \end{itemize}
\end{construction}

\begin{lemma}
\label{lem:PSU3-final-correctness}
The presentation of Construction~\ref{con:ShortPresentationPSU} presents
$\operatorname{PSU}_3(q)$.
\end{lemma}

\begin{proof}
Since $\mathcal P_B$ presents $B$, the arguments of
\cite[Lemmas~5, 7, and 14 and Theorems~15--16]{HulpkeSeress} apply with
our normal-form words $\widehat{\operatorname{dec}}_U$ in place of theirs.  The
indicated matrices satisfy all defining relations and generate
$\operatorname{PSU}_3(q)$, so there is an epimorphism from the
presented group to $\operatorname{PSU}_3(q)$.

The elements $b=u(s,0)$ and $c=u(\delta,0)$ are nonidentity elements of
$Z(U)$.  By \cite[Lemma~8]{HulpkeSeress}, they are conjugate under
$\langle h\rangle$.  By \cite[Lemma~7]{HulpkeSeress}, the relation
$\mathcal R(c)$ therefore implies the relation $\mathcal R(b)$ in the
abstract presentation.

If $d=1$, the presentation consequently implies
$
   \mathcal R(a),\,
   \mathcal R(b),\,
   \mathcal R(ab).
$
The hypotheses of \cite[Theorem~15]{HulpkeSeress} hold because
$x\in\rho S^\times$, $s\in S^\times$, and
$x+s=x'\in S^\times/\rho$.  Hence the presentation defines
$\operatorname{PSU}_3(q)$.

If $d=3$, the elements $\nu_0,\nu_1,\nu_2$ have the same norm and
represent the three cube classes.  The presentation therefore implies
$
   \mathcal R(a_0b),\mathcal R(a_1b),\mathcal R(a_2b).
$
The hypotheses of \cite[Theorem~16]{HulpkeSeress} hold, and the same
conclusion follows.
\end{proof}

Finally, we have essentially completed the proof of Proposition~\ref{prop:UniformityPSU}, the uniform short presentation of $\PSU_3(q)$; here we gather the last remaining pieces of the size and complexity analysis, to complete the proof:

\begin{proof}[Proof of Proposition~\ref{prop:UniformityPSU}]
All field elements are represented as polynomials of degree less than
$2\exponent$ over $\FF_p$.  Addition, multiplication, inversion, 
powering, trace, and norm are computable in time polynomial in
$\exponent\log p=\log q$.  The trace equations and all coordinate expansions
are solved by Gaussian elimination on vector spaces of dimension
$O(\exponent)$ over $\FF_p$.  The exponents $e_t$ are computed by the Euclidean algorithm on $O(\log q)$-bit integers.  Lemma~\ref{lem:PSU3-norm-lifts-compressed} uses only exponentiation and
square-root extraction with the explicitly known nonsquare
$\lambda=\rho^{q+1}$.

By Lemma~\ref{lem:PSU3-Borel-compressed}, the Borel part has
binary-exponent length $O(\log^2q)$.  The Bruhat part adds only
three or seven relations, each of binary-exponent length $O(\log q)$.
Thus the complete presentation has binary-exponent length
$O(\log^2q)$.

Finally, apply the repeated-squaring conversion of
\cite[Remark~1.3]{BGKLP} (reproduced in Remark~\ref{rem:presentation-length-conventions} above) simultaneously to every exponent in the
presentation. This converts the presentation to the word-length
convention used in this paper while preserving the $O(\log^2q)$ bound,
and the conversion is itself computable in time polynomial in $\log q$.
Since
$
   |\PSU_3(q)|
      =\frac{q^3(q^3+1)(q^2-1)}{\gcd(3,q+1)},
$
we have $\log|G|=\Theta(\log q)$.  The resulting presentation has word
length $O(\log^2|G|)$ and is produced in $\DTIMEpl$.
\end{proof}

\section{Uniformity of the Guralnick--Kantor--Kassabov--Lubotzky Short Presentation for the Suzuki groups $^2 B_2(q) = \Sz(q)$} \label{app:suzuki}

Hulpke and Seress \cite[p.~720]{HulpkeSeress} write that J. Thompson told W. Kantor in personal communication that the original paper of Suzuki \cite{suzuki} contains within it, at least implicitly, a short presentation of the Suzuki groups $^2B_2(q) = \Sz(q)$. Guralnick, Kantor, Kassabov, and Lubotzky \cite{GKKLquant} gave a short and bounded\footnote{In the context of the Short Presentation Conjecture, whereas ``short'' means polylogarithmic length, ``bounded'' means $O(1)$ generators and $O(1)$ relators.} presentation of these groups, but do not say anything about uniformity. We show that their presentation can be made uniform. Precisely, we will establish the following.

\begin{proposition}  \label{prop:UniformSuzuki}
Let $k > 0$, and let $q = 2^{2k+1}$. There exists a $\DTIME(\polylog(q))$ algorithm that, given the standard name for $\Sz(q)$, constructs a short presentation for $\Sz(q)$.
\end{proposition}

\subsection{Preliminaries: polynomial notation} \label{sec:PolynomialNotation}

If $\gamma$ lies in an extension field of $\mathbb{F}_{p}$, then
$m_{\gamma}(x)$ denotes its minimal polynomial over $\mathbb{F}_{p}$. If
$\delta \in \mathbb{F}_{p}[\gamma]$, let $f_{\delta;\gamma}(x)\in
\mathbb{F}_{p}[x]$ denote the unique polynomial satisfying
$
f_{\delta;\gamma}(\gamma)=\delta,
\,\,
\deg(f_{\delta;\gamma})<\deg(m_{\gamma}).
$

Following Guralnick, Kantor, Kassabov, and Lubotzky's polynomial notation
\cite[Section~4.3.0, Equations~(4.15)--(4.16)]{GKKLquant}, for any polynomial
$
g(x)=\sum_{i=0}^{e}g_i x^i\in \mathbb{Z}[x]
$
and any two elements $u,h$ in a group, define
$
[[u^{g(x)}]]_h
:=
\prod_{i=0}^{e}(u^{g_i})^{h^i},
$
where the product is taken in increasing order of $i$, and we are using the
convention $x^h=h^{-1}xh$. Equivalently, 
\[
[[u^{g(x)}]]_h
=
u^{g_0}h^{-1}u^{g_1}h^{-1}\cdots h^{-1}u^{g_e}h^e.
\]

We spell out the following effective procedures related to the above because they will be useful for establishing uniformity of the presentation below.

\begin{lemma}\label{lem:EffectivePolynomialNotation}
Let $q=p^e$, and suppose the standard name gives $\mathbb{F}_q$ over
$\mathbb{F}_p$ together with a primitive element $\zeta\in\mathbb{F}_q^\times$.
Then the following can be done in $\mathsf{DTIME}(\polylog(q))$.

\begin{enumerate}
\item Compute the minimal polynomial $m_\zeta(x)$.
\item Given $\delta\in\mathbb{F}_q$, compute the ``coordinate polynomial''
$f_{\delta;\zeta}(x)$ satisfying $f_{\delta;\zeta}(\zeta)=\delta$ and
$\deg(f_{\delta;\zeta})<e$.
\item For every $s \leq O(e)$, compute $x^{p^s}\bmod m_\zeta(x)$. 

\item Given $g(x)=\sum_i g_i x^i$ with $\deg(g)\le e$ and $0\leq g_i<p$, output
the word $[[u^{g(x)}]]_h$ (as an element of a free group in which $u,h$ are two of the generators). For $p=2$, this word has length $O(e)$. 
\end{enumerate}
\end{lemma}

\begin{proof}
1. The standard name realizes $\mathbb{F}_q$ as
$\mathbb{F}_p[Y]/f(Y)$. Compute the unique linear dependence among
$1,\zeta,\zeta^2,\ldots,\zeta^e
$
over $\mathbb{F}_p$. This is Gaussian elimination on an $(e+1)\times e$ matrix
over $\mathbb{F}_p$, hence takes time polynomial in $e$ and $\log p$. (Since $\zeta$ is a primitive element, we have that $1,\zeta,\dotsc,\zeta^{e-1}$ are linearly independent, for otherwise the field $\FF_p(\zeta)$ would be smaller than order $q$.)

2. Express $\delta$ in the basis
$
1,\zeta,\ldots,\zeta^{e-1}
$
over $\mathbb{F}_p$ by solving a linear system. The resulting coefficients are
the coefficients of $f_{\delta;\zeta}(x)$.

3. Compute
$
R_0(x)=x,
\,\,
R_{i+1}(x)=R_i(x)^p\bmod m_\zeta(x).
$
After $s$ iterations, $R_s(x)=x^{p^s}\bmod m_\zeta(x)$. Since all polynomials
have degree $<e$, this takes $\poly(e,\log p)=\polylog(q)$ time. 

4. When $p=2$, each coefficient is $0$ or $1$, so one directly expands
$
[[u^{g(x)}]]_h
$
omitting the factors $u^0$. This has length $O(e)$.
\end{proof}

\subsection{Preliminaries on the Suzuki groups} \label{sec:SuzukiPreliminaries}

\noindent Let $q = 2^{2k+1}$. We follow the notation of \cite{GKKLquant} for ease of reference and comparison, and this material can be found there with further references.
Throughout this section, we will assume $q > 8$. For if $q \leq 8$, we can provide the multiplication table. Let $\mathbb{F}_{q}^{\times} = \langle \zeta \rangle$, and let $\theta : x \mapsto x^{2^{k+1}}$ be the field automorphism whose square is the Frobenius automorphism $x \mapsto x^2 = (x^{\theta})^{\theta}$.
Define $U$ to be the group consisting of  all pairs $(\alpha, \beta) \in \mathbb{F}_{q}^{2}$ with the multiplication rule:
\[
(\alpha, \beta) \cdot (\gamma, \delta) = (\alpha + \gamma, \beta + \delta + \alpha \gamma^{\theta}),
\] 
and $W := Z(U) = [U,U] = \{ (0, \beta) : \beta \in \mathbb{F}_{q} \}$.

For $\epsilon \in \mathbb{F}_{q}^{\times}$, let $h_{\epsilon}$ denote the automorphism of $U$ defined by 
\begin{align} \label{eq:GKKL4.27}
(\alpha, \beta)^{h_{\epsilon}} = (\epsilon \alpha, \epsilon^{\theta+1} \beta).
\end{align}

Then $B = U \rtimes \langle h_{\zeta} \rangle$ is isomorphic to the Borel subgroup of $\Sz(q)$. The cyclic group $\langle h_{\zeta} \rangle \cong \FF_q^\times$ 
 is transitive on the non-trivial elements of both $Z(U)$ and $U/Z(U)$. If $x \in U$ and $v = (i,j,k,\ell) \in \mathbb{Z}^{4}$, we write:
\begin{equation} \label{eq:four-exp}
x^{v} = x^{(i,j,k,\ell)} = x^{h} \text{ where } h = h_{\zeta}^{i} h_{\zeta^{\theta}}^{j} h_{\zeta+1}^{k} h_{\zeta^\theta+1}^{\ell} 
\end{equation}
\cite[Equation~(4.28), p.~29]{GKKLquant}.
 We will need the following set $S_{0}$ of $16$ vectors from $\mathbb{Z}^{4}$ \cite[Table~B.4]{GKKLquant}.
\begin{align*}
S_{0} = \{ &(-1, -1,0,0), (-1, 0,0,0), (0, -1, 0,0), (0, 0, -1, -1), (0,0, -1, 0) \\
&(0,0,0,-1), (0,0,0,0), (0,1,0,0), (1,0,0,0), (1,0,1,0), \\
&(1,1,0,0), (1,1,1,0), (1,1,1,1), (2,1,1,1) \\
&(-2,-2,0,0), (-2,-1,0,0) \}. 
\end{align*}

With this notation in hand, we recall the short, bounded presentation for $\Sz(q)$ from \cite{GKKLquant}. We note that in this presentation, the elements $h_\zeta$ etc. will line up, under isomorphism, with the elements $h_\zeta$ described above, but they are not \emph{assumed} to do so; that is, in the presentation, the elements $h_\zeta,h_{\zeta+1}, h_{\zeta^\theta} h_{\zeta^\theta+1}$ are merely four independent generators, and the relations between them are described explicitly by the relators of the presentation (there are no hidden implicit relations implied by the notation).

\begin{theorem}[{\cite[Section~4.4.3]{GKKLquant}}] \label{thm:GKKLSuzuki}
Let $k > 0$, and let $q = 2^{2k+1}$. $\Sz(q)$ admits the following presentation  with  $7$ generators and $43$ relations. \\

\noindent \textbf{Generators:} $u,w,t, h_{\zeta}, h_{\zeta+1}, h_{\zeta^{\theta}}, h_{\zeta^{\theta}+1}.$ \\

\noindent \textbf{Relations:} In the relations below, $a = h_{\zeta}^{-1} h_{\zeta^{\theta}}$, $w_{1}, \ldots, w_{5}$ are suitable elements of $W := \langle w^{\langle a \rangle} \rangle$ and $u_{1}, u_{2}, u_{3}, u_{\star, 1}, u_{\star, 2}, u_{\star, 3}$ are suitable elements of $U := \langle u^{\langle h_{\zeta} \rangle} \rangle$, determined by our choice of $u$ and $\zeta$. Here and below, the notation $[[x^{g(x)}]]_{h}$ is the polynomial word notation from Section~\ref{sec:PolynomialNotation}. 
\begin{enumerate}
\item $[h_{\star}, h_{\bullet}] = 1$ for all $\star, \bullet \in \{ \zeta, \zeta^{\theta}, \zeta+1, \zeta^{\theta}+1\}$.
\item $w = u^{2}$ and $w^{2} = 1$.
\item $w^{(0,0,-1,1)} = ww^{(-1,1,0,0)}$ (see \eqref{eq:four-exp} above for notation).
\item $w^{(0,0,2,-1)} = ww^{(2,-1,0,0)}$.
\item $[w,w^{(-1,1,0,0)}] = 1$.
\item $[[w^{m_{\zeta}(x)}]]_{a} = 1$. 
\item $w^{(1,0,0,0)} = [[w^{x^{2^{k+1}+1}}]]_{a}$, where $x^{2^{k+1}+1}$ is reduced modulo $m_{\zeta}(x)$ in order to obtain a short word.
\item $u^{(1,0,0,0)} = uu^{(0,0,1,0)}w_{1}$.
\item $u^{(0,1,0,0)} = uu^{(0,0,0,1)}w_{2}$.
\item $[u^{v}, w] = 1$ for every $v \in S_{0}$.
\item $[u, u^{(1,0,0,0)}] = w_{3}$.
\item $[[u^{m_{\zeta}(x)}]]_{h_{\zeta}} = w_{4}$.
\item $u^{(0,1,0,0)} = [[u^{x^{2^{k+1}}}]]_{h_{\zeta}} \cdot w_{5}$, where $x^{2^{k+1}}$ is reduced modulo $m_{\zeta}(x)$ in order to obtain a short word.
\item $t^{2} = 1$.
\item $h_{\star}^{t} = h_{\star}^{-1}$ for all $\star \in \{ \zeta, \zeta+1, \zeta^{\theta}, \zeta^{\theta}+1\}$.
\item $t = u_{1}u_{2}^{t}u_{3}$.
\item $h_{\star}t = u_{\star, 1}u_{\star, 2}^{t} u_{\star, 3}$ for all $\star \in \{ \zeta, \zeta+1, \zeta^{\theta}, \zeta^{\theta}+1\}$.
\end{enumerate}
\end{theorem}

Relations $(1)$--$(13)$ are the Borel relations from \cite[Proposition~4.30]{GKKLquant}, and relations $(14)$--$(17)$ are the additional relations used in \cite[Proposition~4.30 and Section~4.4.3]{GKKLquant}. The presentation itself
and its correctness are due to \cite{GKKLquant}. What we add below is the uniformity, by showing 
that all words appearing in the presentation can be constructed from
the standard name in time $\mathsf{DTIME}(\polylog(q))$.

\subsection{Proof of Proposition~\ref{prop:UniformSuzuki}}

As in the cases of the Ree groups (Lemma~\ref{lem:Ree-bijection}) and $\PSU_3(q)$ (Lemma~\ref{lem:PSU3-dec}), another key to our approach to a uniform presentation here is the following lemma, which solves the constructive membership problem in $U$ relative to a particular generating set. 

\textbf{Convention.} Throughout the remainder of this section we use the notation $U(A,B)$ to denote the element of $U$ that has coordinates $(A,B)$ in the notation from the start of Section~\ref{sec:SuzukiPreliminaries}.

\begin{lemma}\label{lem:SuzukiU}
Let $u = U(1,0)$ and $w = U(0,1)$, and let $h_\zeta,h_{\zeta^\theta}$ be as above. There is a $\DTIME(\polylog(q))$ algorithm which, given $(A,B) \in \FF_q^2$ as input, outputs a word $\text{Word}_{A,B}$ over four generators $\hat{u}, \hat{w}, \hat{h}_{\zeta}, \hat{h}_{\zeta^\theta}$ such that $\text{Word}_{A,B}(u,w,h_\zeta,h_{\zeta^\theta}) = U(A,B)$. 
\end{lemma}

\begin{proof}
Let $a=h_{\zeta}^{-1}h_{\zeta^\theta}$ and let $\hat{a}$ be the corresponding element of the free group, viz. $\hat{a} = \hat{h}_{\zeta}^{-1}\hat{h}_{\zeta^\theta}$.  By the action formula~\eqref{eq:GKKL4.27}, the element $a$ acts on
$W=Z(U)$ as multiplication by $\zeta$. Indeed, $(\zeta^{-1}\zeta^\theta)^{\theta+1} = \zeta^{\theta^2-1} =\zeta.$ Hence, if $D\in\mathbb F_q$ and
$f_{D;\zeta}(x)=\sum_{i=0}^{e-1}d_i x^i$ is the coordinate polynomial satisfying $f_{D;\zeta}(\zeta)=D$, then
$
[[\hat{w}^{f_{D;\zeta}(x)}]]_{\hat{a}}
$
is a word of length $O(e) = O(\log q)$ evaluating to $U(0,D)$.

Similarly, for any $A\in\mathbb F_q$, if
$
f_{A;\zeta}(x)=\sum_{i=0}^{e-1}a_i x^i,
$
then
$
P_A:=[[\hat{u}^{f_{A;\zeta}(x)}]]_{\hat{h}_\zeta}
$
has first coordinate $A$. Say
$
P_A=U(A,C_A),
$
where $C_A$ is computed by simulating the multiplication rule in $U$. Explicitly,
$
C_A=\sum_{0\leq i<j<e}a_i a_j\,\zeta^i(\zeta^j)^\theta.
$
Therefore, for any $A,B\in\mathbb F_q$, the word
\[
\operatorname{Word}_{A,B}
:=
[[\hat{u}^{f_{A;\zeta}(x)}]]_{\hat{h}_\zeta}\,
[[\hat{w}^{f_{B-C_A;\zeta}(x)}]]_{\hat{a}}
\]
evaluates to $U(A,B)$. By Lemma~\ref{lem:EffectivePolynomialNotation}, this
word is computable in $\DTIME(\polylog(q))$ and has length $O(e)=O(\log q)$.
\end{proof}

\noindent We how have all the ingredients in place to prove Proposition~\ref{prop:UniformSuzuki}.

\begin{proof}[Proof~of~Proposition~\ref{prop:UniformSuzuki}]
Let $e=2k+1=\log_2 q$. We construct the presentation from
Theorem~\ref{thm:GKKLSuzuki}. The correctness of this presentation was established in 
\cite[Section~4.4.3]{GKKLquant}. It remains only to show that the relators can be output in $\DTIMEpl$.

We now show the relators can be generated in $\DTIMEpl$.

\begin{itemize}[align=parright, labelwidth=1.7cm]
\item[(1)--(5), (10), (14)--(15)]
Each of these relators is fixed (involves only fixed constants independent of $q$ and have constant word-length). To see this for relation (10), note that (10) 
uses the fixed set
$S_0$ of $16$ vectors (see Section~\ref{sec:SuzukiPreliminaries}). Whenever a word of the form $x^{(i,j,r,s)}$ occurs, we
expand it as
$x^{h_\zeta^i h_{\zeta^\theta}^j h_{\zeta+1}^r h_{\zeta^\theta+1}^s}.
$
All vectors appearing here are fixed, so these relators can be written down in
$\DTIMEpl$.

\item[(6)]
By Lemma~\ref{lem:EffectivePolynomialNotation}, we compute $m_\zeta(x)$ from
the standard name in $\DTIMEpl$. Since $p=2$, all coefficients of
$m_\zeta(x)$ are $0$ or $1$. Thus
\[
[[w^{m_\zeta(x)}]]_a
\]
expands directly as a word of length $O(e)=O(\log q)$ in $w$ and $a$. Since
$a=h_\zeta^{-1}h_{\zeta^\theta}$ has constant length in the generators from Theorem~\ref{thm:GKKLSuzuki},
relation $(6)$ can be output in $\DTIMEpl$.

\item[(7)]
Again by Lemma~\ref{lem:EffectivePolynomialNotation}, we compute
\[
g_{\theta+1}(x):=x^{2^{k+1}+1}\bmod m_\zeta(x)
\]
in $\DTIMEpl$. The polynomial $g_{\theta+1}(x)$ has degree $<e$ and
coefficients in $\mathbb F_2$, so
$
[[w^{g_{\theta+1}(x)}]]_a
$
has length $O(e) = O(\log q)$  and, again by Lemma~\ref{lem:EffectivePolynomialNotation}, can be computed efficiently. Hence relation $(7)$ can be output in $\DTIMEpl$.

\item[(8)--(9), (11)]
The elements $w_1,w_2,w_3$ are the central correction terms appearing in \cite[Proposition~4.30]{GKKLquant}. In the explicit
pair model for $U$ from the start of Section~\ref{sec:SuzukiPreliminaries}, compute
\[
D_1:=(u\,u^{(0,0,1,0)})^{-1}u^{(1,0,0,0)} \qquad
D_2:=(u\,u^{(0,0,0,1)})^{-1}u^{(0,1,0,0)} \qquad
D_3:=[u,u^{(1,0,0,0)}].
\]
By \cite{GKKLquant}, these elements lie in $W=Z(U)$. Define $B_i$ by 
$
D_i=U(0,B_i)\,\, (i=1,2,3).
$
Then output
\[
w_i:=\operatorname{Word}_{0,B_i}.
\]
from Lemma~\ref{lem:SuzukiU}.
This gives words of length $O(e) = O(\log q)$ for $w_1,w_2,w_3$, computable in $\DTIMEpl$.

\item[(12)--(13)] 
For relation $(12)$, compute
\[
D_4:=[[u^{m_\zeta(x)}]]_{h_\zeta}.
\]
By relator $(12)$, this element is the central correction $w_4$.
Define $B_4$ by $D_4=U(0,B_4)$ and output
$
w_4:=\operatorname{Word}_{0,B_4}
$
from Lemma~\ref{lem:SuzukiU}.
For relation $(13)$, compute
\[
g_\theta(x):=x^{2^{k+1}}\bmod m_\zeta(x)
\]
using Lemma~\ref{lem:EffectivePolynomialNotation}, and then compute
\[
D_5:=
\left([[u^{g_\theta(x)}]]_{h_\zeta}\right)^{-1}
u^{(0,1,0,0)}.
\]
By relator $(13)$, this is the central correction $w_5$. Define $B_5$ by 
$D_5=U(0,B_5)$ and output
\[
w_5:=\operatorname{Word}_{0,B_5}.
\]
Thus relations $(12)$ and $(13)$ can thus be output in $\DTIMEpl$.

\item[(16)--(17)]
We now use Suzuki's explicit formulas. Guralnick, Kantor, Kassabov, and Lubotzky \cite[Section~4.4.3]{GKKLquant} state that relation $(16)$ is due
to Suzuki's \cite[Formula~(13)]{suzuki}, and that explicit versions of relations $(16)$ and
$(17)$ are given by Suzuki's \cite[Formulas~(13) and~(38)]{suzuki}. We use these formulas in the notation of the
present section.

To avoid conflict with Suzuki's element $\rho$, put
$
\varrho:=2^k.
$
Then $\varrho\theta=q$ as exponents on $\mathbb F_q$. Suzuki's formula~(38),
translated into the present notation, gives, for every
$\lambda\in\mathbb F_q^\times$,
\begin{equation}\label{eq:SuzukiHtFactorization}
h(\lambda)t
=
U(\lambda^{-\varrho},0)\,
U(0,\lambda^{\varrho+1})^t\,
U(\lambda^{-\varrho},\lambda^{-(\varrho+1)}).
\end{equation}
Equivalently, this identity follows by direct multiplication in Suzuki's
$4\times 4$ matrix model, see Suzuki's formula~(38) \cite[(38)]{suzuki}.

For relation $(16)$, take $\lambda=1$ in
\eqref{eq:SuzukiHtFactorization}. Then
$
t=U(1,0)U(0,1)^tU(1,1).
$
Suzuki also notes in the paragraph preceding \cite[Formula~(39)]{suzuki} that one may take
$\rho=(1,0)$ and $\sigma=(0,1)$ \cite[paragraph preceding formula~(39)]{suzuki}.
Thus in our notation $\rho=u$ and $\sigma=w=u^2$. Since
\[
U(1,1)=U(1,0)U(0,1)=uw,
\]
we may choose
\[
u_1=u,\qquad u_2=w,\qquad u_3=uw.
\]
Hence relation $(16)$ is output as
$
t=u\,w^t\,u\,w.
$

For relation $(17)$, let
\[
\lambda_\star\in\{\zeta,\zeta+1,\zeta^\theta,\zeta^\theta+1\}
\]
be the field element corresponding to $h_\star=h(\lambda_\star)$. By
\eqref{eq:SuzukiHtFactorization}, we may choose
$
u_{\star,1}=U(\lambda_\star^{-\varrho},0),
$
$
u_{\star,2}=U(0,\lambda_\star^{\varrho+1}),
$
and
$
u_{\star,3}=U(\lambda_\star^{-\varrho},
\lambda_\star^{-(\varrho+1)}).
$
The field elements
$
\lambda_\star^{-\varrho},\,\,
\lambda_\star^{\varrho+1},\,\,
\lambda_\star^{-(\varrho+1)}
$
are computable from the standard name by repeated squaring and inversion in
$\mathbb F_q$, hence in $\DTIMEpl$. Once these elements are computed as pairs
$U(A,B)\in U$ we write them using $\operatorname{Word}_{A,B}$ (Lemma~\ref{lem:SuzukiU}). Each
such word has length $O(e)=O(\log q)$ and is computable in $\DTIMEpl$.

Thus relation $(17)$ is output as 
\[
h_\star t
=
\operatorname{Word}_{\lambda_\star^{-\varrho},0}\,
\left(\operatorname{Word}_{0,\lambda_\star^{\varrho+1}}\right)^t\,
\operatorname{Word}_{\lambda_\star^{-\varrho},\lambda_\star^{-(\varrho+1)}}.
\]
Therefore relations $(16)$ and $(17)$ can be output uniformly in $\DTIMEpl$. \qedhere
\end{itemize}

\end{proof}


\section{Conclusion and Open Questions} \label{sec:conclusion}

\paragraph{Upper Bounds.} We placed \algprobm{Group Isomorphism} into the second level of the polylogarithmic-time hierarchy (Theorem~\ref{thm:MainUpperBound}). As a consequence, \algprobm{Group Isomorphism} can be solved by uniform $\textsf{AC}$ circuits of depth-$3$, size $n^{\polylog(n)}$, and polylogarithmic fan-in at the bottom level (i.\,e., depth $2\frac{1}{2}$). This improved the depth over the previous-best: 
Collins, Grochow, Levet, and Wei\ss \, \cite{CGLW} showed that the more general \algprobm{Quasigroup Isomorphism} problem belongs to the third level of the polylogarithmic time hierarchy, and hence can be solved by uniform $\textsf{AC}$ circuits of depth-$3\frac{1}{2}$ and size $n^{O(\log n)}$. 

Extending Theorem~\ref{thm:MainUpperBound} to \algprobm{Quasigroup Isomorphism} seems out of reach. Our technique depends heavily on both the theory of composition series for groups and 
the Classification of Finite Simple Groups.
Classifying all finite simple quasigroups is considered hopeless, and for general quasigroups there is essentially no analogous theory of composition series that can be used the way we use it in this paper. It is even open whether all finite simple quasigroups are $O(1)$-generated. We also point out that our approach of using polylogarithmically short presentations for finite (simple) groups unconditionally does not extend to the setting of quasigroups (Observation~\ref{obs:quasigroups}).

As part of our depth-$2\frac{1}{2}$ upper bound we also show that short presentations of $\PSU_3(q)$ and the Suzuki groups from the literature \cite{HulpkeSeress, GKKLquant} can be made uniform (Propositions~\ref{prop:UniformityPSU} and \ref{prop:UniformSuzuki}, respectively.) Although a short presentation for the Ree groups $^2 G_2(q)$ is not yet available in the literature, in getting around this we develop tools for uniformity analogous to those we used for $\PSU_3(q)$ and the Suzuki groups. Namely, for all three groups we show that, relative to a certain generating set and in a certain coordinate-model of the groups, constructive membership in the Borel subgroup (or its unipotent radical) can be solved in $\DTIMEpl$ (Lemma~\ref{lem:Ree-bijection} for the Ree groups $^2 G_2(q)$; Lemma~\ref{lem:PSU3-dec} for $^2 A_2(q) = \PSU_3(q)$; and Lemma~\ref{lem:SuzukiU} for the Suzuki groups $^2 B_2(q) = \Sz(q)$.) Although the details of the generating sets and coordinate models differ from the three groups, we note the similarities between these lemmata (and their proofs).

\paragraph{Lower bounds.} We also exhibit the first non-trivial circuit lower-bounds against \algprobm{Group Isomorphism}. Precisely, we show that verifying that a multiplication table specifies a group, or even a Latin square (=quasigroup), requires DNFs of exponential size (Theorem~\ref{thm:DNF}). Furthermore, we show that isomorphism testing of Abelian groups requires CNFs of size $n^{\Omega(\log n)}$ (Theorem~\ref{thm:CNF}), and even $\ComplexityClass{CNF \circ XOR}$ circuits of size $n^{\Omega(\log n)}$ (Theorem~\ref{thm:CNF-XOR}). The lower bound for CNFs relies crucially on a result of Seguins Pazzis \cite{deSeguinsPazzis2010} bounding the dimension of an affine linear space of matrices of bounded rank, and the lower bound on \ComplexityClass{CNF \circ XOR} relies crucially on a new result we prove (Theorem~\ref{thm:space-missing-k-1}) bounding the dimension of affine linear spaces avoiding rank $k-1$ (but allowing ranks both above and below $k-1$). At the end of Section~\ref{sec:CNF} we also discuss obstacles to extending our technique from matrices to tensors in attempting to get an $n^{\Omega(\log^2 n)}$ lower bound.

Our lower-bounds also hold more generally for \algprobm{Latin Square Isotopy} (see Remark~\ref{rmk:isotopy} and Corollary~\ref{cor:isotopy}). Collins, Grochow, Levet, and Wei\ss \, \cite{CGLW} previously exhibited the same upper-bound for \algprobm{Latin Square Isotopy} as for \algprobm{Quasigroup Isomorphism}. Precisely, they  showed that \algprobm{Latin Square Isotopy} belongs to the third level of the polylogarithmic time hierarchy, and hence can be solved by uniform $\textsf{AC}$ circuits of depth-$3\frac{1}{2}$ and size $n^{O(\log n)}$. 

Circuit lower bounds are commonly obtained via either low-depth reductions from \algprobm{Parity} or \algprobm{Majority}, or via a switching lemma-type argument. The $\exists^{\log^2 n}\FOLL$ bound for \algprobm{Quasigroup Isomorphism} due to Chattopadhyay, Tor\'an, and Wagner \cite{ChattopadhyayToranWagner} unconditionally rules out the use of low-depth reductions from \algprobm{Parity} or \algprobm{Majority}. In particular, $\exists^{\log^2 n}\FOLL$ does not contain either \algprobm{Parity} or \algprobm{Majority}. Note that neither previous works \cite{ChattopadhyayToranWagner, CGLW} nor our work here rules out the use of switching lemmas more generally. It might be possible to obtain circuit lower-bounds for other algebraic problems in the multiplication table model, through the use of low-depth reductions from validity checking (Theorem~\ref{thm:DNF}) or isomorphism testing of Abelian groups (Theorem~\ref{thm:CNF}). 

\paragraph{Open questions.}
Perhaps the biggest question left open by our work is still to prove that \algprobm{Group Isomorphism} is not in $\ACz$ (or is!). As more approachable next steps, we highlight the following questions:

\begin{question}
Does \algprobm{Group Isomorphism} require CNFs of size $n^{\Omega(\log^2 n)}$? Or more strongly of exponential size?
\end{question}

\begin{question}
Do depth-3 circuits for \algprobm{Group Isomorphism} require super-polynomial size?
\end{question}

Because of our polynomial-size CNF for validity testing of multiplication tables (Lemma~\ref{lem:valid-cayley-tables}), the above two questions are equivalent to their promise versions.

\begin{question}
Does isomorphism testing of Abelian groups, in the multiplication table model (with the promise that the multiplication tables are valid multiplication tables for groups), require DNFs of super-polynomial size?
\end{question}

\section*{Acknowledgments}
\noindent M.L. thanks Martin Kassabov and James B. Wilson for helpful discussions. J. G. thanks Lance Fortnow for helpful discussions on how to name and discuss our low-level circuit classes.

G. K. is supported by the Templeton World Charity Foundation, Inc. (funder DOI: 501100011730) under the grant DOI:10.54224/20650
no. 20650 on “Building Diverse Intelligences through
Compositionality and Mechanism Design”, and additionally by the J. Grochow and R. Frongillo start-up funds at the University of Colorado Boulder.
M. L. was partially supported by a Faculty Research and Development Grant from the College of Charleston. J. G. was supported by NSF CAREER grant CCF-2047756.

\section*{LLM Usage Statement}
The authors declare that no large language models (LLMs) were used in the research, writing, or editing of this work.

\appendix
\addtocontents{toc}{\protect\setcounter{tocdepth}{-5}} 

\section{Mathematica code to verify matrix calculations symbolically in $^2 G_2(q)$} \label{app:mathematica}

We executed the following code using Wolfram 14.3.0.0 \cite{Mathematica} on a MAC OS X ARM (64-bit) platform on 20 Aug 2026.

We use the function \texttt{AntidiagonalMatrix} by Sander Huisman \cite{huisman}:
\begin{verbatim}
ClearAll[AntidiagonalMatrix]
AntidiagonalMatrix[list : (_List | _SparseArray)] := 
 AntidiagonalMatrix[list, 0]
AntidiagonalMatrix[list : (_List | _SparseArray), k_Integer] := 
 Reverse[DiagonalMatrix[list, k], 2]
AntidiagonalMatrix[list : (_List | _SparseArray), k_Integer, 
  n_Integer] := AntidiagonalMatrix[list, k, {n, n}]
AntidiagonalMatrix[list : (_List | _SparseArray), k_Integer, 
  mn : {m_Integer, n_Integer}] := 
 Reverse[DiagonalMatrix[list, k, mn], 2]
\end{verbatim}
We then set up our generating set of $^2 G_2(q)$ from Lemma~\ref{lem:Ree-bijection}:
\begin{verbatim}
AlphaMat := Function[a, {{1, a^L, 0, 0, -a^(3*L + 1), -a^(3*L + 2), a^(4*L + 2) }, 
                         {0, 1, a, a^(L + 1), -a^(2*L + 1), 0, -a^(3*L + 2)}, 
                         {0, 0, 1, a^L, -a^(2*L), 0, a^(3*L + 1)}, 
                         {0, 0, 0, 1, a^L, 0, 0}, 
                         {0, 0, 0, 0, 1, -a, a^(L + 1)},
                         {0, 0, 0, 0, 0, 1, -a^L},
                         {0, 0, 0, 0, 0, 0, 1}}]
BetaMat := Function[b, {{1 , 0, -b^L, 0, -b, 0, -b^(L + 1)}, 
                        {0, 1, 0, b^L, 0, -b^(2*L), 0},
                        {0, 0, 1, 0, 0, 0, b},
                        {0, 0, 0, 1, 0, b^L, 0},
                        {0, 0, 0, 0, 1, 0, b^L},
                        {0, 0, 0, 0, 0, 1, 0},
                        {0, 0, 0, 0, 0, 0, 1}}]
GammaMat := Function[c, {{1 , 0, 0, -c^L, 0, -c, -c^(2*L)},
                         {0, 1, 0, 0, -c^L, 0, c},
                         {0, 0, 1, 0, 0, c^L, 0},
                         {0, 0, 0, 1, 0, 0, -c^L},
                         {0, 0, 0, 0, 1, 0, 0},
                         {0, 0, 0, 0, 0, 1, 0},
                         {0, 0, 0, 0, 0, 0, 1}}]
HMat := Function[lam, 
  DiagonalMatrix[{lam^L, lam^(1 - L), lam^(2*L - 1), 1, 
                  lam^(1 - 2*L), lam^(L - 1), lam^(-L)}]]
TMat := AntidiagonalMatrix[{-1, -1, -1, -1, -1, -1, -1}]
\end{verbatim}
Here we left \texttt{L} (our $\ell$) as a symbol with no constraints. Although some expressions might be able to be simplified by taking into account the true semantics of $\ell$ in the context of the Ree groups, we turned out not to need that, as the elements we wanted appeared even with leaving \texttt{L} as an unconstrained variable. 

Finally, we verify the needed calculations:
\begin{verbatim}
Part[HMat[lam] . AlphaMat[a] . BetaMat[b] . GammaMat[c] . TMat . 
  AlphaMat[x] . BetaMat[y] . GammaMat[z], 7, 1]

> -lam^(-L)

Part[AlphaMat[a] . BetaMat[b] . GammaMat[c] . TMat . AlphaMat[x] . 
  BetaMat[y] . GammaMat[z] , 6, 1]

> a^L

Part[BetaMat[b] . GammaMat[c] . TMat . AlphaMat[x] . BetaMat[y] . GammaMat[z], 
     3, 1]

> -b

Part[GammaMat[c] . TMat . AlphaMat[x] . BetaMat[y] . GammaMat[z], 2, 1]

> -c
\end{verbatim}

\bibliographystyle{alphaurl}
\bibliography{references}

\end{document}